\documentclass[11pt]{article}
\usepackage{fullpage}
\usepackage{amssymb}
\usepackage{amsmath}
\usepackage{amsthm}
\usepackage{multirow}
\usepackage{enumerate}
\usepackage{graphicx}
\usepackage{mathrsfs}
\usepackage[utf8]{inputenc}
\usepackage[noadjust]{cite}

\usepackage{color}
\usepackage[pdfstartview=FitH,pdfpagemode=UseNone,colorlinks=true,citecolor=blue,linkcolor=blue, backref=page]{hyperref}

\usepackage{cleveref}
\usepackage{mathtools}
\usepackage{relsize}
\usepackage{amsfonts}
\usepackage[T1]{fontenc}
\usepackage{mathpazo}
\usepackage{bm}
\usepackage{changepage}
\usepackage[usenames,dvipsnames,svgnames,table]{xcolor}
\usepackage{thmtools}
\usepackage{thm-restate}
\usepackage{subcaption}
\usepackage{enumitem}
\usepackage{algorithm}
\usepackage{algpseudocode}

\usepackage{nicematrix}
\usepackage{arydshln}

\usepackage{tikz}
\usetikzlibrary{positioning}
\usetikzlibrary{calc}
\usetikzlibrary{decorations.shapes}
\usetikzlibrary{decorations.pathreplacing,angles,quotes}
\usepackage{todonotes}

\allowdisplaybreaks

\definecolor{blueviolet}{RGB}{60,50,200}
\definecolor{oliveg}{RGB}{40,200,30}
\hypersetup{colorlinks=true,       
	linkcolor=blueviolet,
	citecolor=oliveg,
}

\newtheorem{theorem}{Theorem}[section]
\theoremstyle{definition}
\newtheorem{definition}[theorem]{Definition}

\newtheorem{proposition}[theorem]{Proposition}
\newtheorem{lemma}[theorem]{Lemma}
\newtheorem{fact}[theorem]{Fact}

\newtheorem{corollary}[theorem]{Corollary}
\newtheorem{claim}[theorem]{Claim}

\newtheorem{remark}[theorem]{Remark}

\makeatletter
\newenvironment{breakablealgorithm}
{
	\refstepcounter{algorithm}
	\hrule height.8pt depth0pt \kern2pt
	\renewcommand{\caption}[2][\relax]{
		{\raggedright\textbf{\fname@algorithm~\thealgorithm} ##2\par}%
		\ifx\relax##1\relax 
		\addcontentsline{loa}{algorithm}{\protect\numberline{\thealgorithm}##2}%
		\else 
		\addcontentsline{loa}{algorithm}{\protect\numberline{\thealgorithm}##1}%
		\fi
		\kern2pt\hrule\kern2pt
	}
}{
	\kern2pt\hrule\relax
}
\makeatother

\renewcommand{\Comment}[1]{\Statex {\begin{center}\textcolor{Gray}{/* #1  */} \end{center}}}

\algrenewcommand\alglinenumber[1]{\footnotesize #1.}

\renewcommand{\bar}{\overline}
\newcommand{\comment}[1]{}

\newcommand{\paren}[1]{\left(#1\right)}

\newcommand{\brac}[1]{\left[#1\right]}

\newcommand{\set}[1]{\left\{#1\right\}}

\newcommand{\ceil}[1]{\left\lceil #1 \right\rceil}
\newcommand{\floor}[1]{\left\lfloor #1 \right\rfloor}
\newcommand{\floorceil}[1]{\left\lfloor #1 \right\rceil}

\newcommand{\abs}[1]{\left\lvert#1\right\rvert}

\newcommand{\vecw}{{\mathbf{w}}}
\newcommand{\vecx}{{\mathbf{x}}}
\newcommand{\vecy}{{\mathbf{y}}}
\newcommand{\vecz}{{\mathbf{z}}}

\newcommand{\F}{\mathbb{F}}

\newcommand{\K}{\mathbb{K}}

\newcommand{\N}{\mathbb{N}}

\newcommand{\V}{\mathbb{V}}

\newcommand{\Z}{\mathbb{Z}}

\newcommand{\cJ}{\mathcal{J}}

\newcommand{\cM}{\mathcal{M}}

\newcommand{\cP}{\mathcal{P}}

\newcommand{\cS}{\mathcal{S}}
\newcommand{\cT}{\mathcal{T}}

\newcommand{\cV}{\mathcal{V}}

\newcommand{\cX}{\mathcal{X}}

\newcommand{\rank}{\mathrm{rank}}

\newcommand{\spacespanned}[1]{\left\langle#1\right\rangle}

\newcommand{\VNP}{\mathsf{VNP}}

\newcommand{\der}[2]{\frac{\partial{#1}}{\partial{#2}}}

\newcommand{\poly}{\textnormal{poly}}

\newcommand{\SP}{\mathsf{SP}}

\newcommand{\APP}{\mathsf{APP}}

\newcommand{\size}{\textnormal{size}}
\newcommand{\fract}[1]{\set{#1}}
\newcommand{\degseq}{\text{\sc Deg-seq}}
\newcommand{\funK}{\text{\sc Upt-K}}

\newcommand{\funJ}{\cJ}
\newcommand{\entropy}{\mathsf{residue}}

\newcommand{\unbiased}{\emph{-unbiased}}
\newcommand{\partialf}{{\mathbf{\boldsymbol\partial}}}
\newcommand{\leaves}{\textnormal{leaves}}
\newcommand{\canon}{\textnormal{can}}
\newcommand{\encoding}{\textnormal{encoding}}
\newcommand{\swap}{\textnormal{swap}}
\newcommand{\Fsm}{{\mathbb{F}}_{sm}}
\newcommand{\SkewP}{\mathsf{SkewP}}
\newcommand{\evalDim}{\mathsf{evalDim}}
\newcommand{\relrk}{\mathsf{relrk}}
\newcommand{\PD}{\mathsf{PD}}

\newcommand{\FF}{\mathbb{F}}
\newcommand{\VV}{\mathbb{V}}
\newcommand{\inbrace}[1]{\left \{ #1 \right \}}

\definecolor{DSgray}{cmyk}{0,0,0,0.7}

\newcommand\numberthis{\addtocounter{equation}{1}\tag{\theequation}}

\date{}
\renewcommand\thefootnote{\fnsymbol{2}}
\begin{document}
	\title{Low-depth arithmetic circuit lower bounds via shifted partials}
	\author{Prashanth Amireddy\thanks{Work done while the author was a research fellow at Microsoft Research India.}\\
	    \normalsize{Harvard University}\\
	    \normalsize{{\tt pamireddy@g.harvard.edu}}
	    \and {Ankit Garg}\\
		\normalsize{Microsoft Research India}\\
		\normalsize{\tt{garga@microsoft.com}}
		\and {Neeraj Kayal}\\
		\normalsize{Microsoft Research India}\\
		\normalsize{{\tt neeraka@microsoft.com}}
	    \and {Chandan Saha\thanks{Partially supported by a MATRICS grant of the Science and Engineering Research Board, DST, India.}}\\
		\normalsize{Indian Institute of Science}\\
		\normalsize{{\tt chandan@iisc.ac.in}}
		\and {Bhargav Thankey\thanks{Supported by the Prime Minister's Research Fellowship, India.}}\\
		\normalsize{Indian Institute of Science}\\
		\normalsize{{\tt thankeyd@iisc.ac.in}}}
	\maketitle
	\thispagestyle{empty}
 
        \renewcommand\thefootnote{\textcolor{red}{\arabic{footnote}}}
    
	\begin{abstract}
		We prove superpolynomial lower bounds for low-depth arithmetic circuits using the \emph{shifted partials} measure \cite{GKKS14, Kayal12a} and the \emph{affine projections of partials} measure \cite{GargKS20, KayalNS20}. The recent breakthrough work of Limaye, Srinivasan and Tavenas \cite{lst1} proved these lower bounds by proving lower bounds for low-depth \emph{set-multilinear} circuits. An interesting aspect of our proof is that it does not require conversion of a circuit to a set-multilinear circuit, nor does it involve random restrictions. We are able to  upper bound the measures for homogeneous formulas directly, without going via set-multilinearity. Our lower bounds hold for the iterated matrix multiplication as well as the Nisan-Wigderson design polynomials. We also define a subclass of homogeneous formulas which we call \emph{unique parse tree} (UPT) formulas, and prove superpolynomial lower bounds for these. This generalizes the superpolynomial lower bounds for \emph{regular} formulas \cite{KayalSS14, FournierLMS15}.
	\end{abstract}

        \newpage
	  \thispagestyle{empty}
	\tableofcontents
	\newpage
	\pagenumbering{arabic}
	
	\addtocontents{toc}{\protect\setcounter{tocdepth}{2}}
	\newtheorem{openProblem}[theorem]{Open Problem}


\section{Introduction} \label{sec:intro}

Arithmetic circuits are a natural model for computing polynomials using the basic operations of addition and multiplication. One of the most fundamental questions about arithmetic circuits is about finding a family of explicit polynomials (if they exist) that cannot be computed by polynomial-sized arithmetic circuits. The existence of such explicit polynomials was conjectured by Valiant in 1979 \cite{Valiant79a} and is the famed $\text{VP}$ vs $\text{VNP}$ conjecture. Arithmetic circuit lower bounds are expected to be easier than Boolean circuit lower bounds. Among many reasons, one is due to the phenomenon of depth reduction. Arithmetic circuits can be converted into low-depth circuits preserving the output polynomial and not blowing up the size too much \cite{ValiantSBR83, AgrawalV08, Koiran12, Tavenas15, GKKS16}. Due to this, strong enough lower bounds even for restrictive models of computation like depth-$3$ circuits or homogeneous depth-$4$ circuits can lead to superpolynomial arithmetic circuit lower bounds.

Arithmetic formulas are an important subclass of arithmetic circuits where the out-degree of every gate is at most $1$. For constant-depth, formulas and circuits are polynomially related. Also, all our results deal with formulas. So we will only refer to formulas from here on. We consider (families of) polynomials having degree at most polynomial in $n$, the number  of inputs/variables. One of the first results studying low-depth arithmetic formulas was that of \cite{NisanW97}, who proved lower bounds for homogeneous depth-$3$ formulas. Progress was stalled for a while, and then various lower bounds for homogeneous depth-$4$ formulas were proven in a series of works \cite{Kayal12a, GKKS14, KayalSS14, FournierLMS15, KumarS14a, KayalLSS17, KumarS17}. There was limited progress for higher-depth formulas, and lower bounds remained open even for depth-$5$ formulas. In a recent breakthrough work, \cite{lst1} proved superpolynomial lower bounds for constant-depth arithmetic formulas. Their lower bounds are of the form $n^{\Omega(\log(n)^{c_\Delta})}$ for a constant $0 < c_\Delta < 1$ depending on the depth $\Delta$ of the formula. The following two open problems naturally emerge out of their work.

\begin{openProblem}\label{open:formulaLB}
Prove superpolynomial lower bounds for general arithmetic formulas. This is interesting even for homogeneous arithmetic formulas.
\end{openProblem}

\begin{openProblem}\label{open:expLB}
Prove exponential lower bounds for constant-depth arithmetic formulas. This is interesting even for homogeneous depth-$5$ formulas.
\end{openProblem}

Towards answering Open Problems 1.1 and 1.2, let us examine the lower bound proof in \cite{lst1} at a high level. Their proof has two main steps: First, they reduce the problem of proving lower bounds for low-depth formulas to the problem of proving lower bounds for low-depth \emph{set-multilinear} formulas; set-multilinear formulas are special homogeneous formulas with an underlying partition of the variables into subsets (see Section \ref{sec: homogeneous circuits}). \cite{lst1} calls such reductions `hardness escalation'. Second, they use an interesting adaptation of the rank of the partial derivatives matrix measure \cite{Nisan91} to prove a lower bound for low-depth set-multilinear formulas. They call this measure \textit{relative rank} ($\relrk$). The effectiveness of the $\relrk$ measure crucially depends on a certain `imbalance' between the sizes of the sets used to define set-multilinear polynomials. The proof in \cite{lst1} raises two natural and related questions: \\
	
	\noindent\textbf{Question 1:} Can we skip the hardness escalation, i.e., the set-multilinearization, step? \\
	
	\noindent\textbf{Question 2:} Can we design a measure that exploits some weakness of homogeneous (but not necessarily set-multilinear) formulas directly? \\
	
	\noindent\textbf{Motivations for studying Question 1:} Set-multilinearization comes with an exponential blow up in size -- a homogeneous, depth-$\Delta$ formula computing a set-multilinear polynomial of degree $d$ can be converted to a set-multilinear formula of depth $\Delta$ and size $d^{O(d)}\cdot s$ (see Lemma 12 in \cite{lst1}). So, an exponential lower bound for low-depth set-multilinear formulas does not seem to imply an exponential lower bound for low-depth homogeneous formulas since we are restricted to work with $d \leq \frac{\log n}{\log \log n}$.\footnote{Note that a homogeneous degree-$d$ polynomial can be computed by a depth-$2$ homogeneous formula of size $n^{O(d)}$ (assuming $d \ll n$). So, to obtain an exponential lower bound, $d$ must be $n^{\Omega(1)}$.} Indeed, it is possible to strengthen and refine the argument in \cite{lst1} to get an exponential lower bound for low-depth set-multilinear formulas (Appendix~\ref{sec:sml-lb}). An approach that evades the hardness escalation step and directly works with homogeneous formulas has the potential to avoid the $d^{O(d)}$ loss and give an exponential lower bound for low-depth homogeneous formulas. For instance, a variant of the shifted partials measure \cite{GKKS14, Kayal12a} yields an exponential lower bound for homogeneous depth-$4$ formulas \cite{KayalLSS17, KumarS17}. If we go via the hardness escalation approach, then we only get a quasi-polynomial lower bound for the same model.
	
	Besides, such a direct argument can be used to prove lower bounds for polynomials which are not necessarily set-multilinear or do not have a 'non-trivial' set-multilinear component. \\
 
	\noindent\textbf{Motivations for studying Question 2:} Typical measures used for proving lower bounds for arithmetic circuits include the partial derivatives measure ($\PD$) \cite{NisanW97, ShpilkaW01}, the rank of the partial derivatives matrix measure (a.k.a. evaluation dimension -- in short, $\evalDim$) \cite{Nisan91, Raz06, RazY09}, the shifted partials measure ($\SP$) and its variants \cite{GKKS14, KayalSS14, KayalLSS17}, the affine projections of partials measure ($\APP$) \cite{GargKS20, KayalNS20}, etc. All these measures are defined for \emph{any} polynomial, which is not necessarily set-multilinear.\footnote{To define $\evalDim$ for a polynomial in $\F[\vecx]$, we just need to partition $\vecx$ into two variable sets $\vecy$ and $\vecz$ such that the $\vecy$ variables are used for evaluations, and the remaining variables are $\vecz = \vecx\setminus \vecy$.} Whereas the $\relrk$ measure used in \cite{lst1}, although very effective, is defined only for set-multilinear polynomials. Measures such as $\PD$, $\SP$, and $\APP$ have the geometrically appealing property that they are invariant under the application of invertible linear transformations on the variables. Since low-depth formulas, as well as low-depth homogeneous formulas, are closed under linear transformations, it is natural to look for measures that do not blow up much on applying linear transformations.
	
	Another important motivation for studying Question 2 is towards learning low-depth homogeneous formulas. While the `hardness escalation' paradigm  of reducing to the set-multilinear case works for proving lower bounds, it is not clear how to exploit it to design learning algorithms. Lower bounds for arithmetic circuits are intimately connected to learning \cite{FortnowK09, Volkovich16, KayalS19, GargKS20}. Hence if we have a lower bound measure that directly exploits the weakness of low-depth homogeneous formulas, it opens up the possibility of new learning algorithms for such models.

\subsection{Our Results}

We give a direct lower bound for low-depth homogeneous formulas via the shifted partials measure which was used in the series of works on homogeneous depth-$4$ lower bounds. While our proof also yields lower bounds only in the low degree setting, the hope is that it could potentially lead to a stronger lower bound in the future.

Consider the shifted partials measure: $\SP_{k,\ell}(f) := \dim \langle \vecx^\ell\cdot\partialf^k (f) \rangle$, where $f$ is a polynomial. That is, $\SP_{k,\ell}(f)$ is the dimension of the space spanned by the polynomials obtained on multiplying degree $\ell$ monomials to partial derivatives of $f$ of order $k$. Also, for convenience, let us denote by $M(n,k) := {n + k - 1 \choose k}$ the number of monomials of degree $k$ in $n$ variables. Then note that for a homogeneous polynomial $f$ of degree $d$, 
$$
\SP_{k,\ell}(f) \le \min\{M(n,k) M(n,\ell), M(n, d - k + \ell)\}.
$$
We show that for polynomials computed by low-depth homogeneous formulas, the shifted partials measure with an appropriate setting of $k$ and $\ell$ is substantially smaller than the above upper bound. At the same time, we exhibit explicit `hard' polynomials for which the shifted partials measure is close to the above bound, hence yielding a lower bound.

\begin{theorem}[\textbf{Lower bound for low-depth homogeneous formulas via shifted partials}]
Let $C$ be a homogeneous formula of size $s$ and product-depth $\Delta$ that computes a polynomial of degree $d$ in $n$ variables. Then for an appropriate value of $k$ and $\ell$, 
$$
\SP_{k,\ell}(C) \le \frac{s  \: 2^{O(d)}}{n^{\Omega(d^{2^{1-\Delta}})}} \: \min\{M(n,k) M(n,\ell), M(n, d - k + \ell)\}.
$$
At the same time, there are homogeneous polynomials $f$ of degree $d$ in $n$ variables (e.g. an appropriate projection of iterated matrix multiplication polynomial, Nisan-Wigderson design polynomial etc.) s.t. 
$$
\SP_{k,\ell}(f) \ge 2^{-O(d)} \: \min\{M(n,k) M(n,\ell), M(n, d - k + \ell)\}.
$$
This gives a lower bound of $s \ge \frac{n^{\Omega(d^{2^{1-\Delta}})}}{2^{O(d)}}$ on the size of homogeneous product-depth $\Delta$ formulas computing these polynomials.
\end{theorem}

\begin{remark} Here are a few remarks about the above theorem.
\begin{enumerate}
    \item We also show that the above lower bound can be derived using the affine projections of partials measure, or even skewed partials (Lemma \ref{lem:ckt_APP}). 
    \item The above lower bound is slightly better than the lower bound of \cite{lst1}. Instead of the $d^{O(d)}$ loss incurred due to converting homogeneous formulas to set-multilinear formulas, our direct analysis only incurs a $2^{O(d)}$ loss.\footnote{In fact, this loss can be brought down to $2^{O(k)}$, but we ignore this distinction because eventually we set $k = \Theta(d)$ to get the lower bounds.} So, for example, for homogeneous product-depth $2$ formulas, our superpolynomial lower bound continues to hold for a higher degree ($\log^2(n)$ vs $\left(\log(n)/ \log \log(n)\right)^2$ in \cite{lst1}). While the improvement maybe insignificant, this hints at something interesting going on with the direct approach.
    \item The dependence on $\Delta$ is slightly better as compared to \cite{lst1}. But this is not because of the direct approach but due to improvements in the underlying number theoretic calculations. In fact the paper \cite{BhargavDS22} improves these calculations even further but similar to \cite{lst1}, they also work with set-multilinear formulas and not directly with homogeneous formulas.
\end{enumerate}
\end{remark}

Lower bounds for general-depth arithmetic formulas are expected to be easier than arithmetic circuit lower bounds. However, despite several approaches and attempts (e.g. via tensor rank lower bounds \cite{Raz13}), we still do not have superpolynomial arithmetic formula lower bounds. However, there has been some success in proving lower bounds for some natural restricted models (apart from the depth restrictions considered above). For example, \cite{KayalSS14} considered the model of \emph{regular} arithmetic formulas. These are formulas which consist of alternating layers of addition ($+$) and multiplication ($\times$) gates s.t. the fanin of all gates in any fixed layer is the same. This is a natural model and the best known formulas for many interesting polynomial families like determinant, permanent, iterated matrix multiplication etc. are all regular. \cite{KayalSS14} proved a superpolynomial lower bound on the size of regular formulas for an explicit polynomial and later \cite{FournierLMS15} proved a tight lower bound for the iterated matrix multiplication polynomial.

We prove superpolynomial lower bounds for a more general model.\footnote{The model in \cite{KayalSS14, FournierLMS15} allowed non-homogeneity but with the formal degree upper bounded by a small constant times the actual degree. However, we only work with homogeneous formulas.} Consider a model of homogeneous arithmetic formulas consisting of alternating layers of addition ($+$) and multiplication ($\times$) gates s.t. the fanin of all addition gates can be arbitrary but fanin of product gates in any fixed layer is the same. We call these \emph{product-regular}. We prove super-polynomial lower bounds for (homogeneous) product-regular formulas. Previously we did not know of lower bounds for even a much simpler model where the fanins of all the product gates are fixed to $2$.

In fact, we prove lower bounds for an even more general model which we call \emph{Unique Parse Tree} (UPT) formulas. A parse tree of a formula is a tree where for every addition gate, one picks exactly one child and for every product gate, we pick all the children. Then we ``short circuit'' all the addition gates. Parse trees capture the way monomials are generated in a formula. We say that a formula is UPT if all its parse trees are isomorphic. A product-regular formula is clearly UPT. In the theorem below, $IMM_{n,\log n}$ is the iterated multiplication polynomial of degree $\log n$ (see Section \ref{subsec: poly families}).

\begin{theorem}\label{thm:uptLB}
Any UPT formula computing $IMM_{n,\log(n)}$ has size at least $n^{\Omega(\log\log(n))}$. A similar lower bound holds for the Nisan-Wigderson design polynomial.
\end{theorem}

\begin{remark} Here are a few remarks about the above theorem.
\begin{enumerate}
\item While (homogeneous) product-regular formulas are restricted to compute polynomials with only certain degrees (e.g. higher product-depth cannot compute prime degrees), (homogeneous) UPT formulas do not suffer from this restriction. For example, see Figure \ref{fig:sub-first} for a UPT formula of product-depth $2$ that computes a degree $3$ polynomial.
\item While this result could possibly also be obtained by defining a similar model in the set-multilinear world, proving a lower bound there and then transporting it back to the homogeneous world, our framework has fewer number of moving parts and hence makes it easier to derive such results.
\end{enumerate}
\end{remark}

\subsection{Techniques and proof overview}

A lot of lower bounds in arithmetic complexity follow the following outline.
\\

\noindent {\bf Step 1: Depth reduction/decomposition.} One first shows that if $f(\vecx)$ is computed by a {\em small} circuit from some restricted subclass of circuits, then there is a corresponding subclass of depth four circuits such that $f(\vecx)$ is also computed by a {\em relatively small} circuit from this subclass\footnote{
   Some major results in the area such as \cite{Raz03, lst1} did not originally proceed via a depth reduction but instead analysed formulas directly. These results can however be restated as first doing a depth reduction and then applying the appropriate measure. 
}. The resulting subclass is of the form:

$$
f(\vecx) = \sum_{i=1}^s \prod_{j=1}^{t_i} Q_{i,j}
$$

Usually there are simple restrictions on the degrees of $Q_{i,j}$'s. For example, they could be upper bounded by some number.
\\

\noindent {\bf Step 2: Employing a suitable set of linear maps.} Let $\F[\vecx]^{=d}$ be the vector space of homogeneous polynomials of degree $d$, 
    $W$ be a suitable vector space, and $\mathrm{Lin}(\F[\vecx]^{=d}, W)$ be the space of linear maps from $\F[\vecx]^{=d}$  to $W$. We choose a suitable set of linear maps $\mathcal{L} \subseteq \mathrm{Lin}(\F[\vecx]^{=d}, W)$. 
    These define a complexity measure
        $$ \mu_{\mathcal{L}}(f) := \dim(\mathcal{L}(f)), \quad \text{where~} \mathcal{L}(f) := \langle \left\{ L(f) \ : \ L \in \mathcal{L} \right\} \rangle. $$

We would like to choose the set of linear maps $\mathcal{L}$ so that they identify some weakness of the terms $\prod_{j=1}^t Q_j$ in the depth-$4$ circuit. That is $\mu_{\mathcal{L}}\left(\prod_{j=1}^t Q_j\right)$ should be much smaller than $\mu_{\mathcal{L}}(f)$ for a generic $f$. For example, if $Q_j$'s are all linear polynomials, then we can choose $\mathcal{L}$ to be the partial derivative maps of order $k$, $\partialf^{k}$. Then $\mu_{\mathcal{L}}\left(\prod_{j=1}^t Q_j\right) \le {t \choose k} \ll {n+k-1 \choose k}$ which is the value for a generic $f$ (for $k \le t/2$). This is the basis of homogeneous depth-$3$ formula lower bounds in \cite{NisanW97}.

For proving lower bounds for bounded bottom fan-in depth-$4$ circuits (i.e. when degree of $Q_j$'s is upper bounded by some number), \cite{GKKS14, Kayal12a} introduced the shifted partials measure and used the linear maps $\mathcal{L} = \vecx^\ell\cdot\partialf^k$. The main insight in their proof was that if we apply a partial derivative of order $k$ on $\prod_{j=1}^t Q_j$ and use the product rule, then at least $t-k$ of the $Q_j$'s remain untouched. This structure can then be exploited by the shifts to get a lower bound. This intuition however completely breaks down for $k \ge t$ (see Appendix \ref{sec:geometric} for more discussion on this). Due to this, progress remain stalled for higher depth arithmetic circuit lower bounds.

\cite{lst1} get around the above obstacle by working with set-multilinear circuits\footnote{
    \cite{lst1} use only a special case of their relative rank framework (with two different set sizes). We only consider this special case here in a slightly different interpretation.} 
which entails working with polynomials over $d$ sets of variables $(\vecx_1,\ldots,\vecx_d)$, $|\vecx_i| = n$. 
Let us use the shorthand $\vecx_S = (\vecx_i)_{i \in S}$. The products they deal with are of the form $\prod_{j=1}^t Q_j(\vecx_{S_j})$ where $S_1, S_2, \ldots, S_t$ form a partition of $[d]$. The set of linear maps they use are $\mathcal{L} = \Pi \circ \partial_{\vecx_A}$ for a subset $A \subseteq [d]$. 
$\Pi$ is a map which sets $n - n_0$ variables in each of the variable sets in $\vecx_{[d]\backslash A}$ to $0$. They observe  (for the appropriate choice of $n_0$) that
\begin{align*}
\mu_{\mathcal{L}}\left(\prod_{j=1}^t Q_j(\vecx_{S_j}) \right) &\le \frac{n^{|A|}}{2^{\frac{1}{2}\sum_{j=1}^t \text{imbalance}_j}}
\end{align*}

Here $\text{imbalance}_j = ||A \cap S_j|\log(n) - |S_j \backslash A| \log(n_0)|$. For the appropriate choice of $n_0$, a generic set-multilinear $f$ satisfies $\mu_{\mathcal{L}}(f) = n^{|A|}$, so that lower bound (on the number of summands of the low depth circuit) obtained is exponential in the total imbalance 
$\sum_{j=1}^t \text{imbalance}_j$. \cite{lst1} observe that this quantity is {\em somewhat large} for the depth four circuits obtained from various circuit models that they consider.

We observe here that that the core of the argument \cite{lst1} allows us to unravel some structure in partial derivatives of order $k$ applied on $\prod_{j=1}^t Q_j$ for values of $k \gg t$. 
We use this  to derive a structure for the partial derivative space of a product $\prod_{j=1}^t Q_j(\vecx)$. Consider a partial derivative operator of order $k$ indexed by a multiset $\alpha$ of size $k$. Then using Leibniz's differentiation rule,
$$
\partial_{\alpha} \prod_{j=1}^t Q_j = \sum_{\alpha_1,\ldots, \alpha_t: \: \sum_{i=1}^t \alpha_i = \alpha} c^{\alpha}_{\alpha_1,\ldots, \alpha_t} \prod_{j=1}^t \partial_{\alpha_j} Q_j
$$
for appropriate constants $c^{\alpha}_{\alpha_1,\ldots, \alpha_t}$'s. In the product $\prod_{j=1}^t \partial_{\alpha_j} Q_j$, we can try to club terms into two groups depending on if the size of $|\alpha_j|$ is small or large. It turns out that the right threshold for $|\alpha_j|$ is $k \ \text{deg}(Q_j)/d$ (i.e. if we divide the order of the derivative proportional to the degrees of the terms). Let 
$$
S := \{j: |\alpha_j| \le k\: \text{deg}(Q_j)/d\}
$$
Define $k_0 := \sum_{j \in S} |\alpha_j|$ and $\ell_0 := \sum_{j \in \bar{S}} (\deg(Q_j) - |\alpha_j|)$. Notice that we can write the product $\prod_{j=1}^t \partial_{\alpha_j} Q_j$ as
$$
P \prod_{j \in S} \partial_{\alpha_j} Q_j
$$
for a degree $\ell_0$ polynomial $P$. Hence $\partial_{\alpha} \prod_{j=1}^t Q_j$ can be written as sums of this form. While it is not immediate (due to the condition on $\alpha_j$'s in $S$), it turns out with a bit more work, one can combine the product of partials into a single partial. 

What can we say about $k_0$ and $\ell_0$? It turns out that the quantity that comes up in calculations is the following: $k_0 + \frac{k}{d-k} \ell_0$ and it satisfies
\begin{align*}
    k_0 + \frac{k}{d-k} \ell_0 \le k.
\end{align*}
Note that $k_0$ is between $0$ and $k$, and $\ell_0$ between $0$ and $d-k$. So the normalization brings $\ell_0$ to the right 'scale'. 

It turns out we can give a better bound in terms of a quantity we call as residue defined as 
\begin{equation*}
		\entropy_k(d_1,\dots,d_t) := \frac{1}{2} \cdot \min \limits_{k_1,\dots,k_t \in \Z} ~\sum \limits_{j=1}^t \abs{k_j - \frac{k}{d} \cdot d_j}.
\end{equation*}
and having the property that:
\begin{proposition} Let $k_0$ and $\ell_0$ be defined as above. Then
    $$
    k_0 + \frac{k}{d-k} \ell_0 \le k - \entropy_k(d_1,\dots,d_t)
    $$
    where $d_j = \text{deg}(Q_j)$.
\end{proposition}
We want to spread the derivatives equally among all terms but cannot due to integrality issues. The residue captures this quantitatively and as described below, is what gives us our lower bounds! 

Combined with the above discussion, we get the following structural lemma about the partial derivative space of $\prod_{j=1}^t Q_j$.

\begin{lemma}
$$\spacespanned{\partialf^{k}\paren{Q_1\cdots Q_t}}  \subseteq
	\sum \limits_{\substack{S \subseteq [t], ~k_0 \in [0..k],~\ell_0 \in [0..(d-k)], \\ k_0 + \frac{k}{d-k}\cdot \ell_0 ~\le~ k - \entropy_k(d_1,\dots,d_t)}} \spacespanned{ \vecx^{\ell_0}\cdot\partialf^{k_0}\paren{\prod_{j\in S} Q_j}}.$$
\end{lemma}

Now we have the choice to utilize the above structure using an additional set of linear maps. Both shifts and projections give similar lower bounds, so let us explain shifts here. Note that there is an intriguing possibility of getting even better lower bounds (in terms of dependence on degree $d$) using other sets of linear maps! Because of the above structural result, we have
$$
\spacespanned{\vecx^{\ell} \cdot \partialf^{k}\paren{Q_1\cdots Q_t}} \subseteq
	\sum \limits_{\substack{S \subseteq [t], ~k_0 \in [0..k],~\ell_0 \in [0..(d-k)], \\ k_0 + \frac{k}{d-k}\cdot \ell_0 ~\le~ k - \entropy_k(d_1,\dots,d_t)}} \spacespanned{ \vecx^{\ell + \ell_0}\cdot\partialf^{k_0}\paren{\prod_{j\in S} Q_j}}
$$
Thus we can upper bound,
$$\SP_{k, \ell}(\paren{Q_1\cdots Q_t}) \le 2^t\cdot d^{2}\cdot \max \limits_{\substack{k_0, \ell_0 \ge 0 \\ k_0 + \frac{k}{d-k}\cdot \ell_0~\le~k - {\entropy_k(d_1,\dots,d_t)}}} M(n, k_0)\cdot M(n,\ell_0 + \ell)$$
It turns out that $2^t d^2$ is usually not a problem and the calculations yield that
\begin{align*}
&\max \limits_{\substack{k_0, \ell_0 \ge 0 \\ k_0 + \frac{k}{d-k}\cdot \ell_0~\le~k - {\entropy_k(d_1,\dots,d_t)}}} M(n, k_0)\cdot M(n,\ell_0 + \ell) \\ 
&\le \frac{ 2^{O(d)}}{n^{\entropy_k(d_1,\dots,d_t)}} \: \min\{M(n,k) M(n,\ell), M(n, d - k + \ell)\}
\end{align*}

Now to upper bound the shifted partial dimension of polynomials computed by low-depth formulas, we give a decomposition for such formulas into sums of products of polynomials (Lemma \ref{lem:low-depth-structural}) where the degree sequences are carefully chosen so that that the residues can be simultaneously lower bounded for all the terms (Lemma \ref{lem:low-k-gamma}). While in a different context, these calculations do bear similarity with related calculations in \cite{lst1}.
\\

\noindent {\bf Step 3: Lower Bounding $\dim(\mathcal{L}(f))$ for an explicit $f$.} As a last step, one shows that for some explicit candidate hard polynomial (such as the iterated matrix multiplication polynomial $IMM_{n, d}$) that $\dim(\mathcal{L}(f))$ is large and thereby obtains a lower bound. It turns out that the calculations involved in this step are easier for \cite{lst1} and somewhat more tedious for our choice of linear maps $\mathcal{L}$. 
\\

\noindent \textbf{Application to other subclasses of formulas.} We observe here that for the subclass of homogeneous formulas that we call UPT formulas, one can do a depth-reduction to obtain a depth-four formula in which all the summands have the same factorization pattern (i.e. the sequence of degrees of the factors in all the summands is that same) - see Lemma \ref{lem:deg-seq}. We further observe (Lemma \ref{lem:uptResLB}) that for any fixed sequence of degrees, there exists a suitable value of the parameter $k$ such that the residue is {\em not too small}. This gives us the superpolynomial lower bound for UPT formulas as stated in Theorem \ref{thm:uptLB}. \\ 

\noindent\textbf{Challenges to using the $\PD$ and the $\SP$ measures.} Let us remark briefly why it is a bit surprising that we are able to prove low-depth lower bounds via shifted partials. \cite{GesmundoL19, Rez92} showed that the $\PD$ measure of the polynomial $(x_1^2 + \cdots + x_n^2)^\frac{d}{2}$ is the maximum possible when the order of derivatives, $k$, is at most $\frac{d}{2}$. Notice that $(x_1^2 + \cdots + x_n^2)^\frac{d}{2}$ can be computed by a homogeneous depth-$4$ formula of size $O(nd)$. So, it is not possible to prove super-polynomial lower bounds for low-depth homogeneous formulas using the $\PD$ measure as it is. One may ask if the $\SP$ measure also has a similar limitation.
  
    Some of the finer separation results in \cite{KumarS14, KumarS19} indicate that the $\SP$ measure (and some of its variants) can be fairly large for homogeneous depth-$4$ and depth-$5$ formulas for some particular choices of the order of the derivatives $k$. Also, the exponential lower bounds for homogeneous depth-4 circuits in \cite{KayalLSS17, KumarS17} use random restrictions along with a variant of the $\SP$ measure. It is not clear how to leverage random restrictions for homogeneous depth-5 circuits -- this is also pointed out in the Appendix of \cite{lst1}. 
    
    Fortunately, \cite{KumarS14, KumarS19} do not rule out the possibility of using $\SP$ for all choices of parameters, like, say, $k \approx \frac{d}{2}$, to prove lower bounds for low-depth homogeneous formulas. But, the original intuition from algebraic geometry that led to the development of the $\SP$ measure (see \cite{GKKS14} Section 2.1) breaks down when $k$ is so large (see Appendix \ref{sec:geometric}). Despite these apparent hurdles, and to our surprise, we overcome these challenges and are able to use $\SP$ with $k \approx \frac{d}{2}$ to prove super-polynomial lower bounds for low-depth homogeneous formulas! To the best of our knowledge, no previous work uses $\SP$ with this high a value of $k$. 

\subsection{Prior Work}
We now give a brief account of some of the known arithmetic circuit lower bounds that are related to our work. \\

\noindent \textbf{General circuits and formulas.} Not much is known about lower bounds for general arithmetic circuits and formulas computing explicit polynomials. Baur and Strassen \cite{BaurS83, S73} proved that any arithmetic circuit computing the power symmetric polynomial ($\mathrm{PSym}_{n,d} := x_1^d + \cdots + x_n^d$) or the elementary symmetric polynomial ($\mathrm{ESym}_{n,d} := \sum_{\substack{S \subseteq [n]: \\ |S| = d}} \prod_{i \in S} x_i$) must have size $\Omega(n\log d)$. There has been no improvement on this bound and their result continues to be the best-known lower bound for general arithmetic circuits. 

The best-known lower bound for general arithmetic formulas is quadratic. \cite{Kalorkoti85} proved that any arithmetic formula computing the polynomial $\sum_{i,j \in [n]}x_i^j y_j$ must have size $\Omega(n^2)$. Recently, \cite{ChatterjeeKSV22} proved an $\Omega(n^2)$ lower bound for arithmetic formulas computing $\mathrm{ESym}_{n, 0.1n}$. They also showed that any `layered' Algebraic Branching Program (ABP) computing $\mathrm{PSym}_{n,n}$ has size $\Omega(n^2)$. ABPs are algebraic analogues of (Boolean) branching programs and, as a model of computation, are known to be at least as powerful as formulas. 
Because of the apparent difficulty of proving lower bounds for general models of computation, restricted classes of circuits like multilinear, homogeneous, and low-depth circuits have received a lot of attention in the last few decades. We now discuss a few results for these models. \\

\noindent \textbf{Multilinear and set-multilinear circuits.} A circuit or formula is said to be multilinear if every gate in it computes a multilinear polynomial. 
\cite{RazSY08} showed a lower bound of $\Omega(n^{4/3}/\log^2 n)$ for syntactically multilinear circuits. This lower bound was improved to an $\Omega(n^2/\log^2 n)$ bound in \cite{AlonKV18}. Unlike the case of general circuits and formulas, a super-polynomial separation is known between multilinear circuits and formulas. \cite{Raz06, RazY08} proved that there exists a polynomial computable by polynomial-size multilinear circuits but can only be computed by multilinear formulas of size $n^{\Omega({\log n})}$. \cite{DvirMPY12} showed a similar lower bound but for a polynomial computable by a polynomial-size multilinear ABP. 

Exponential lower bounds are known for low-depth multilinear circuits. A lower bound of $2^{n^{\Omega(1/\Delta})}$ for multilinear circuits of product-depth $\Delta = o(\log n/\log \log n)$ computing the $n \times n$ permanent and determinant was shown in \cite{RazY09}. \cite{ChillaraL019} proved a lower bound of $2^{\Omega(\Delta d^{1/\Delta})}$ for multilinear circuits of product-depth at most $\Delta \leq \log d$ computing $IMM_{2,d}$. A quasi-polynomial separation between product-depth $\Delta$ multilinear circuits and  product-depth $\Delta+1$ multilinear circuits was proved in \cite{RazY09} and improved to an exponential separation in \cite{ChillaraEL018}.

Notice that the lower bounds mentioned in the previous paragraph are of the form $n^{O(1)}\cdot f(d)$ where $f(d)$ is a superpolynomial but sub-exponential function of the degree. Borrowing terminology from parameterised complexity, \cite{lst1} calls such lower bounds FPT lower bounds. As pointed out in \cite{lst1}, it is unclear if FPT bounds can be used to prove lower bounds for low-depth circuits. \cite{lst1} and later \cite{BhargavDS22} prove a non-FPT  lower bound of $n^{d^{\exp(-\Delta)}}$ for set-multilinear formulas of product-depth $\Delta$ computing $IMM_{n,d}$ when $d = O(\log n)$. These lower bounds are then used to prove super-polynomial lower bounds for low-depth circuits. In \cite{lst2}, a non-FPT lower bound of $(\log n)^{\Omega(\Delta d^{1/\Delta})}$ is proved for set-multilinear formulas of product-depth $\Delta = O(\log d)$ computing $IMM_{n,d}$. They also prove a lower bound of $(\log n)^{\Omega(\log d)}$ for any set-multilinear circuit computing $IMM_{n,d}$. Recently, \cite{KushS22} proved a lower bound of $n^{\Omega\paren{n^{1/\Delta}/\Delta}}$ for set-multilinear formulas of product-depth $\Delta$ computing the Nisan-Wigderson design polynomial. \\

\noindent \textbf{Homogeneous and low-depth circuits.} \cite{ShoupS97, Raz10} proved a lower bound of $\Omega(\Delta n^{1 + 1/\Delta})$ for depth $\Delta$ circuits with multiple output gates. In a classic work \cite{NisanW97}, Nisan and Wigderson showed that any homogeneous depth 3 circuit computing $\mathrm{ESym}_{n,d}$ has size $n^{\Omega(d)}$. A series of papers \cite{Kayal12a, GKKS14, KayalSS14, FournierLMS15, KayalLSS17, KumarS17} resulted in an $n^{\Omega(\sqrt{d})}$ lower bound for homogeneous depth 4 circuits computing the Nisan-Wigderson design polynomial and $IMM_{n,d}$. 

\cite{ShpilkaW01} proved a quadratic lower bound for depth 3 circuits computing elementary symmetric polynomials of degree $\Omega(n)$. This was improved to an almost cubic lower bound in \cite{KayalST16} for a polynomial in VNP. Subsequently, \cite{BalajiLS16, Yau16} proved similar lower bounds for polynomials in VP. \cite{GuptaST20} obtained a lower bound of $\widetilde{\Omega}(n^{2.5})$ for depth 4 circuits computing the Nisan-Wigderson design polynomial. As mentioned before, in a recent breakthrough work \cite{lst1}, Limaye, Srinivasan, and Tavenas proved a lower bound of $n^{\Omega(d^{1/(2^\Delta - 1)}/\Delta)}$ for product-depth $\Delta$ circuits computing $IMM_{n,d}$, $d = O(\log n)$. \cite{BhargavDS22} improved this to a lower bound of $n^{\Omega(d^{1/\phi^{2\Delta}}/\Delta)}$ where $\phi = (\sqrt{5} + 1)/2 \approx 1.618$. \\

\noindent\textbf{Circuits with a few parse trees.}  Circuits with a bounded number of parse trees have been studied before in the non-commutative setting \cite{LagardeMP19, lls19}. \cite{lls19} proved a lower bound of $n^{d^{1/4}}$ for non-commutative circuits having at most $2^{d^{1/4}}$ parse trees. This was an improvement on \cite{LagardeMP19}, which proved a lower bound of $2^{\Omega(n)}$ for non-commutative UPT circuits computing the $n \times n$ permanent and determinant. 

The UPT formulas that we study in this work are also related to regular formulas considered in \cite{KayalSS14, FournierLMS15}. A formula is said to be regular if it has alternating levels of addition and multiplication gates and all gates at the same level have the same fan-in. Recall that we call a formula product-regular if the fan-ins of the addition gates in a formula are arbitrary, but the multiplication gates at the same level are restricted to having the same fan-ins. It is easy to see that UPT formulas are a generalization of homogeneous product-regular formulas. \cite{KayalSS14} obtains a lower bound of $n^{\Omega(\log n)}$ for regular formulas computing a polynomial in VNP. \cite{FournierLMS15} later obtained a lower bound of $n^{\Omega(\log d)}$ for regular formulas computing $IMM_{n,d}$. \\

\noindent {\bf Organization.} After describing preliminaries in Section~\ref{sec:prelim}, we present a structural theorem about the partial derivative space of a product of homogeneous polynomials in Section~\ref{sec:der-structural}. This result is then directly used to upper bound both the SP and APP measures of a product of polynomials. Using this result and a decomposition result for low-depth formulas, we obtain lower bounds for low-depth formulas in Section~\ref{sec:low-depth-lb}. Finally, we prove lower bounds for UPT formulas in Section~\ref{sec:upt-lb}.

	\section{Preliminaries} 
\label{sec:prelim}

In this section, we describe the algebraic models we are interested in, and the complexity measures and polynomials used for proving lower bounds for those models. We begin by establishing standard notations and terminology.

\subsection{Notations}

\noindent {\bf Basics.} 
We will attempt to stick to the following usage of symbols: $C,D$ for circuits; $P,Q$ for polynomials; $i,j$ for indices; $d$ for the degree of a polynomial; $s$ for the size of a formula or a sum-of-products decomposition; $k$ for the order of derivatives and $\ell$ for the order of shifts; $t$ for the number of polynomials; $m$ or $X$ for monic monomials; $\alpha, \tau$ as some real parameters; $\mu, \kappa, \sigma$ for maps; $\cS, \cT$ for spaces (sets) of polynomials; and $\cP,\cT$ for binary trees.

Let $a,b,c$ be real numbers. Then we define the sets $[a..b]:= \set{x \in \Z: x \in [a,b]}$ and $[a]:= [1..a]$. 
For a constant $c\ge 1$ and $b \ge 0$, we say $a \approx_{c} b$ if $a \in \brac{b/c,b}$. We write $a \approx b$ if $a \approx_{c} b$ for some (unspecified) constant $c$. 
All logarithms have base 2 unless specified otherwise. 
We define the integer part of $a$ as $\floor{a}:= \max \set{n \in \Z: n \le a}$, the ceiling of $a$ as $\ceil{a}:= \min \set{n \in \Z: n \ge a}$, the fractional part of $a$ as $\fract{a}:= a - \floor{a}$, the nearest integer of $a$ by $\floorceil{a}$ (i.e., if $\fract{a}\le 1/2$, $\floorceil{a}=\floor{a}$ and otherwise $\floorceil{a}=\ceil{a}$), and the absolute value of $a$ by $\abs{a}$.

The following quantity will be crucially used in the proofs of our lower bounds. Here we think of $d_1,\dots,d_t$ as degrees of certain homogeneous polynomials, $d$ as the degree of the product of those polynomials, and $k$ is the order of partial derivatives used for the complexity measures.

\begin{definition}[$\entropy$]\label{defn:entropy}
	For non-negative integers $d_1,\dots,d_t$ such that $d:= \sum \limits_{i=1}^t d_i \ge 1$ and $k \in [0..(d-1)]$, we define 
	\begin{equation*}
		\entropy_k(d_1,\dots,d_t) := \frac{1}{2} \cdot \min \limits_{k_1,\dots,k_t \in \Z} ~\sum \limits_{i=1}^t \abs{k_i - \frac{k}{d} \cdot d_i}.
	\end{equation*}
    The factor of half has been included in the definition just to make the statements of some of the lemmas in our analysis simple. It is easy to show that $\entropy_k(d_1,\dots,d_t)$ is a real number (but not necessarily an integer) that is at most $\frac{k}{2}$. The minimum is attained when for all $i \in [t]$, $k_i = \floorceil{\frac{k}{d}\cdot d_i}$. When we use $\entropy$ in the analysis of complexity measures, we would also have the following additional constraints that $k_i \ge 0$ and $k_i \le d_i$,  $k_1 + \cdots + k_n = k$, where $k$ shall be the order of derivatives. However, imposing these constraints does not alter the value of $\entropy$ by much, so we omit them.
\end{definition}
 
\noindent {\bf Sets and functions.} When some sets $S_1, \dots, S_t$ are pair-wise disjoint, we write their union as $S_1 \sqcup \dots \sqcup S_t$. For a function $\mu : S \to T$ and a subset $A \subseteq S$, we define the multiset $\mu(A) := \set{\mu(x) : x \in A}$. Clearly $|\mu(A)| = |A|$\footnote{For a multiset $B$, $|B|$ denotes its size, i.e. the number of elements in $B$ counted with their respective multiplicities.}. We denote the power set of a set $S$ by $2^S$. We say that a function $\mu: S \to T$ extends $\kappa: A \to T$ if $A \subseteq S$ and for all $x \in A$, $\mu(x)=\kappa(x)$. Let $\mu: S \to T$ and $\kappa: A \to T$ be functions such that $S \cap A = \emptyset$. Then $\mu \sqcup \kappa: S \sqcup A \to T$ is defined by setting $(\mu \sqcup \kappa)(x) = \mu(x)$ for all $x \in S$ and $(\mu \sqcup \kappa)(x) = \kappa(x)$ for all $x \in A$. \\


\noindent {\bf Binomial coefficients.} For non-negative integers $a,b$, we shall denote the quantity $\binom{a+b-1}{b}$ by $M(a,b)$. Note that $M(a,b)$ is the number of (monic) monomials of degree $b$ over $a$ many variables.
 
The following lemma is useful when dealing with binomial coefficients (see Appendix~\ref{app:lem:binomial-coeffs} for a proof).
 
 \begin{lemma}[{\bf Approximations for $M(a,b)$}]
 \label{lem:binomial-coeffs}
	For positive integers $a\ge b\ge c$ and $d$, we have 
		\begin{enumerate}
			\item $(a/b)^b \le M(a,b) \le (6a/b)^b$,
			\item $(a/2b)^c \le \frac{M(a,b+c)}{M(a,b)} \le (2a/b)^c$,
			\item $\frac{M(c,d)}{M(b,d)} \ge \paren{c/b}^d$.
		\end{enumerate}
\end{lemma}

\noindent {\bf Polynomials, derivatives, and affine projections.} Let $n$ and $n_0$ be positive integers. Define variable sets $\vecx:=\set{x_1,\dots,x_n}$ and $\vecz := \set{z_1,\dots,z_{n_0}}$, where $x_1,\dots,x_n$ and 
$z_1,\dots,z_{n_0}$ are distinct variables. Then $\F[\vecx]$ denotes the set of all multivariate polynomials in $\vecx$-variables over the field $\F$. The degree of a polynomial $P \in \F[\vecx]$ is denoted by $\deg(P)$.

There is a natural one-to-one correspondence between monic monomials over $\vecx$ and multisets over $\vecx$ obtained by associating the monomial $m = \prod_{i \in [n]}x_i^{e_i}$ with the multiset $X$ containing $e_i$ many copies of $x_i$ for all $i \in [n]$. For a monic monomial $m$, its corresponding multiset $X$, and a $P \in \F[\vecx]$, we define $\partial_m P = \partial_X P \in \F[\vecx]$ to be the polynomial obtained by successively taking partial derivatives with respect to all the elements of $X$ (the order of elements does not matter). For a function  $\mu:A \to \vecx$, we may simply denote $\partial_{\mu(A)}P$ by $\partial_{\mu} P$ if the domain of $\mu$ is clear from the context. We will use the following facts about partial derivatives. 

\begin{proposition}[{\bf Sum and product rules of derivatives}]
    \label{prop:ders}
    Let $k$ be a positive integer and $P,Q, Q_1,\dots,\\ Q_t \in \F[\vecx]$. Suppose $\mu$ is a function from $[k]$ to $\vecx$. Then
    \begin{enumerate}
        \item $\partial_{\mu}(P+Q) = \partial_{\mu} P + \partial_{\mu} Q$,
        \item $\partial_{\mu}(Q_1 \cdots Q_t) = \sum_{\substack{\kappa:[t]\to 2^{[k]}\text{~s.t.~}\\ \sqcup_{i \in [t]} \kappa(i) = [k]}}~ \partial_{\mu(\kappa(1))}Q_1 \cdots  \partial_{\mu(\kappa(t))}Q_t$.
    \end{enumerate}
\end{proposition}

The product rule above can be obtained by repeatedly using the fact that $\partial_{x_j}(P \cdot Q) = P \cdot \partial_{x_j} Q + Q \cdot \partial_{x_j} P$ for appropriate polynomials $P,Q$ and index $j$. \\ 

For a non-negative integer $\ell$, we define 
		$$\vecx^\ell := \set{{x_1}^{e_1}\cdots{x_n}^{e_n} :  e_1,\dots,e_n \in \Z_{\ge 0} \text{~and~} e_1 + \dots + e_n = \ell}.$$
For a non-negative integer $k$ and $P \in \F[\vecx]$, we define
		$$\partialf^k P := \set{\partial_{m} P: m \in \vecx^k}.$$
For $P \in \F[\vecx]$, a map $L:\vecx \to \spacespanned{ \vecz}$ and $\cS \subseteq \F[\vecx]$, we define $\pi_L(P)\in \F[\vecz] $ and $\pi_L(\cS)\subseteq \F[\vecz]$ as
		$$\pi_L(P) := P(L(x_1),\dots,L(x_n)) \text{~and~}$$
		$$\pi_L(\cS) :=\set{\pi_L(P): P\in \cS}.$$

We now present some elementary notions regarding polynomials that are needed to formulate our complexity measures.\\

\noindent {\bf Spaces of polynomials.} For $\cS,\cT \subseteq \F[\vecx]$, we define 
		$$\cS\cdot \cT := \set{P\cdot Q : P \in \cS \text{~ and ~} Q\in \cT} \text{~and~}\cS + \cT := \set{P + Q : P \in \cS \text{~ and ~} Q\in \cT}.$$
For a set of polynomials $\cS \subseteq \F[\vecx]$, we define  its {\em span} as $\spacespanned{ \cS } \subseteq \F[\vecx]$ to be the set of all polynomials which can be expressed as a linear combination of some elements in $\cS$. That is,
		$$\spacespanned{ \cS } := \set{P \in \F[\vecx] : \exists t \ge 0, a_1,\dots,a_t \in \F, \text{ and } P_1,\dots,P_t \in \cS \text{~such that~} P=a_1\cdot P_1 + \dots + a_t\cdot P_t }.$$
For a set of polynomials $\cS \subseteq \F[\vecx]$, its \emph{dimension}, denoted by $\dim \cS$, refers to the maximum number of {\em linearly independent} polynomials in $\cS$. It is easy to observe the following relations.


\begin{proposition} 
    \label{prop:sub}
    For any two sets of polynomials $\cS,\cT \subseteq \F[\vecx]$,
    \begin{enumerate}
        \item $0 \in \spacespanned{\cS} $,
        \item $\dim \spacespanned{\cS} \le |\cS|$,
        \item If $\cS \subseteq \cT$, then $\dim \spacespanned{\cS} \le \dim \spacespanned{\cT}$,
        \item $ \spacespanned{\cS + \cT} \subseteq  \spacespanned{\cS} +  \spacespanned{\cT}$ and $\dim \spacespanned{\cS + \cT} \le \dim \spacespanned{\cS} + \dim \spacespanned{\cT}$,
        \item $ \dim \spacespanned{\cS\cdot \cT} \le  \dim \spacespanned{\cS}\cdot  \dim \spacespanned{\cT}$. \label{dim-mult}
    \end{enumerate}
    
\end{proposition}

\subsection{Complexity measures}

We can now define the complexity measures for polynomials that we use to prove our lower bounds: the \emph{shifted partials} ($\SP$) measure and the \emph{affine projections of partials} ($\APP$) measure.
We remark here that both these measures (with different parameters)  have been used in the literature prior to our work -- for example, the shifted partials measure in \cite{GKKS14, Kayal12a} and the affine projections of partials in \cite{GargKS20, KayalNS20}.

\begin{definition}[{\bf Complexity measures}] 
	For a polynomial $P \in \F[\vecx]$, non-negative integers $k,\ell$ and $n_0 \in [n]$, we define
	\begin{itemize}
		\item $\SP_{k,\ell}(P) := \dim \spacespanned{ \vecx^\ell\cdot \partialf^k P}$, 
		\item $\APP_{k,n_0}(P) := \max \limits_{L : \vecx \to \spacespanned{\vecz} } \dim \spacespanned{ \pi_L\paren{\partialf^k P} }$.
	\end{itemize}
\end{definition}

Both the above measures are {\em sub-additive}; this can be argued using Proposition~\ref{prop:sub}.

\begin{proposition}[{\bf Sub-additivity of the measures}]
    \label{prop:sub-additive}
    For two polynomials $P,Q \in \F[\vecx]$, field constants $c_1, c_2 \in \F$, and any parameters $k,\ell,n_0$,
    \begin{enumerate}
        \item $\SP_{k,\ell}(c_1\cdot P+c_2\cdot Q) \le \SP_{k,\ell}(P) + \SP_{k,\ell}(Q) $,
        \item $\APP_{k,n_0}(c_1\cdot P+ c_2\cdot Q) \le \APP_{k,n_0}(P) + \APP_{k,n_0}(Q) $.
    \end{enumerate}
\end{proposition}

\begin{remark} \label{remark: APP vs SkewP}
    The lower bounds that we prove in this work can also be obtained using the skewed partials measure ($\SkewP$) \cite{KayalNS20}, which is a special case of $\APP$. \cite{KayalNS20} used the $\SkewP$ measure to prove an optimal ``non-FPT''\footnote{Borrowing terminology from \cite{lst1}.} lower bound of $n^{\Omega(d)}$ for multilinear depth-3 circuits computing $IMM_{n,d}$. However, we use the more general $\APP$ measure for several reasons: Firstly, $\APP$ has the geometrically appealing feature that it is invariant under the application of invertible linear transformations on the variables. Secondly, there are models for which $\APP$ gives lower bounds but $\SkewP$ does not (see Section \ref{subsec: app vs skewp}). The third reason is that for reconstruction of circuits using the recently proposed learning from lower bounds framework \cite{KayalS19, GargKS20}, $\APP$ might give weaker non-degeneracy conditions than $\SkewP$. Thus using $\APP$, we might be able to learn more circuits from a circuit class than we can learn using the $\SkewP$ measure.

    Also, there is a close connection between $\APP$ and the relative rank ($\relrk$) measure used in \cite{lst1}: Both of them are variants of $\evalDim$ with the added feature of `imbalance'. It is natural to wonder to what extent the imbalance is required. The $\relrk$ measure works with an imbalance between the sizes of the sets involved in a set-multilinear partition. An imbalance or skew between the sizes of variable sets also appears in $\APP$, albeit at a \emph{gross level}: $\APP$ uses two sets -- one for taking derivatives, the other for affine projections -- and there is an imbalance between the sizes of these two sets. Drawing analogy with $\evalDim$, one may also view these two sets as the variables used for evaluations ($\vecy$) and the remaining variables ($\vecz$). It turns out that (for set-multilinear polynomials) the ``finer'' imbalance used in the $\relrk$ measure implies an imbalance -- at a gross level -- between $\vecy$ and $\vecz$.\footnote{\cite{lst1} talks about the (relative) rank of the partial derivatives matrix. The rank of this matrix is $\evalDim$ with respect to an appropriate set $\vecy$.} One may naturally ask if an imbalance at a gross level, like in $\APP$ and its precursor $\SkewP$, is sufficient to prove lower bounds for low-depth circuits.
\end{remark}

\subsection{Algebraic circuits}

In this section, we describe the relevant algebraic models of computation -- homogeneous circuits and unique-parse-tree formulas. 

\subsubsection{Algebraic circuits and formulas}

	An \emph{algebraic/arithmetic circuit} $C$ is a directed acyclic graph (DAG) whose  source nodes (called input gates) are labelled by variables (say, from a set $\vecx$) or constants from an underlying field $\F$, all other nodes are addition ($+$) or multiplication ($\times$) gates, and edges (called wires) are labeled by field constants. Each gate $g$ in $C$ naturally computes a polynomial in $\F[\vecx]$ and the polynomial computed by the (unique) sink node (called the output gate) is said to be the polynomial computed by $C$. If the underlying DAG is a directed rooted tree, the circuit is said to be a {\em formula}.
	
	Whether $C$ refers to a circuit or the polynomial computed by it is understood from the context. For example, $\size(C)$ refers to the number of gates in the circuit $C$, whereas $\deg(C)$ refers to the degree of the polynomial computed by $C$. 
	The {\em depth} of a circuit is the maximum number of addition and multiplication gates in any path in the underlying DAG, and the {\em product-depth} is the maximum number of multiplication gates in any path.\\

\noindent {\bf Sub-circuits and substitutions.} For a gate $g$ in a circuit $C$, the circuit between\footnote{i.e., the sub-graph of $C$ induced by the nodes that lie on any directed path from a source node to $g$.} the input gates and $g$ is called the sub-circuit at $g$ and is denoted by $C_g$. We will denote by $C_{g \gets y}$ the circuit obtained by replacing $g$ with $y$. A sub-circuit of a formula is called a {\em sub-formula}.

A circuit of `low-depth' can be converted to a formula of the same depth without much blow-up in size. Thus, lower bounds against low-depth formulas also give lower bounds against low-depth circuits. In fact, it is shown in \cite{lst1} that to prove lower bounds against low-depth formulas computing `low' degree polynomials, it suffices to prove lower bounds against low-depth {\em homogeneous formulas} (provided the characteristic of the underlying field is large enough). In \cite{lst1}, the authors prove lower bounds for low-depth {\em set-multilinear} formulas which are even restrictive models. We now define these models. Unlike \cite{lst1}, we show lower bounds on homogeneous formulas directly without converting them to set-multilinear formulas.

\subsubsection{Homogeneous circuits} \label{sec: homogeneous circuits}
 
   A polynomial $P$ is said to be {\em homogeneous} if all its monomials (if any) have the same degree i.e., all its monomials have the same number of variables, counted with repetition. {We consider the zero polynomial to be homogeneous}. A circuit $C$ is said to be \emph{homogeneous} if each gate $g$ in $C$ computes a homogeneous polynomial. A homogeneous formula is defined analogously. Furthermore, we may assume without loss of generality that all the input gates in a homogeneous circuit or formula are labeled by variables.

    If there exists a partition of variables as $\vecx= \vecx_1 \sqcup \dots \sqcup \vecx_d$ such that all the monomials in a polynomial $P \in \F[\vecx]$ have exactly one variable from each of the variable sets, then $P$ is said to be {\em set-multilinear} with respect to the partition $\set{\vecx_1, \dots, \vecx_d}$. 
    We shall denote the set of all set-multilinear polynomials over $\set{\vecx_1, \dots, \vecx_d}$ by $\Fsm[\vecx_1, \dots, \vecx_d]$.
    A circuit is said to be {\em set-multilinear} (with respect to $\{\vecx_1,\dots,\vecx_d\}$) if each gate in it computes a set-multilinear polynomial with respect to a subset of $\{\vecx_1,\dots,\vecx_d\}$. Observe that a set-multilinear circuit is also a homogeneous circuit. 

    An arithmetic circuit of size $s$ computing a homogeneous polynomial of degree $d$ can be converted into a homogeneous circuit of size $\poly(s,d)$ computing the same polynomial \cite{S73a}. An arithmetic formula computing a homogeneous polynomial can also be homogenized; however, the size of the resulting homogeneous formula is $d^{O(d)}\poly(s)$ \cite{Raz13}. Notice that when $d = O(\log s/ \log \log s)$ the homogeneous formula has size $\poly(s)$. Unfortunately, the homogenization process in \cite{Raz13} does not preserve the depth of the formula. It can convert a formula of even constant depth to a homogeneous formula of depth as large as $O(\log s)$. Recently \cite{lst1} showed that a product-depth $\Delta$ formula of size $s$ computing a homogeneous, degree $d$ polynomial over a field of characteristic $0$ or more than $d$ can be converted into an equivalent homogeneous formula of product-depth $2\Delta$ and size $2^{O(\sqrt{d})}\poly(s)$.
    
    Irrespective of the above-mentioned homogenisation results, homogeneous circuits and formulas are a natural model of computation, and also many polynomials for which lower bounds are known are homogeneous. Thus, in this article, we focus on homogeneous circuits and formulas. Further, we work in the regime of `low' depth circuits and formulas computing `low' degree polynomials. Since a depth $\Delta$ circuit of size $s$ can be converted into an equivalent formula of size $s^{O(\Delta)}$, we shall work with homogeneous formulas (as opposed to homogeneous circuits) for the remainder of the article. 

\subsubsection{Unique-parse-tree (UPT) formulas}

Next, we formalize certain notions about rooted trees and define a subclass of homogeneous formulas which we call \emph{UPT formulas}\footnote{Our definition for UPT formulas is more general than the model considered in a recent paper by Limaye, Srinivasan and Tavenas \cite{lst3} as we do not impose set-multilinearity.}. For this, we define parse trees of a homogeneous formula; they capture the structure of multiplication gates in the formula.

\begin{definition}[{\bf Parse trees of a homogeneous formula}]
	Given a homogeneous formula $C$ computing a degree-$d$ polynomial, we obtain a parse tree $\cP$ of $C$ as follows: Let $\widetilde{C}$ be the formula obtained by arbitrarily removing all but one sub-formula feeding into each addition gate. Viewing $\widetilde{C}$ as a rooted directed tree (with edges directed away from the root) with internal nodes being multiplication gates, leaves being variables, and ignoring the addition gates (by bypassing them) as well as the field constants labelling the edges, we get a parse tree $\cP$. We also discard the labelling of all the (multiplication) gates and the input gates  in $\cP$. 
	
	Clearly, there are only a finite number of parse trees corresponding to a given formula $C$ and they all have exactly $d$ many leaves\footnote{unless $C$ computes the 0 polynomial} as $C$ is homogeneous. For an empty formula (i.e., 0), the empty tree is a parse tree. Substituting some variables to 0 does not affect UPT-ness of a formula (the same way as it does not affect homogeneity).
\end{definition}

\begin{definition}[{\bf UPT formula}]
	A homogeneous formula $C$ is said to be a \emph{unique-parse-tree (UPT) formula} if all of its parse trees are isomorphic to each other as directed graphs. 
	
	It is easy to observe that any sub-formula of a UPT formula is also a UPT formula.  Moreover, without increasing the size by much, any UPT formula can be converted into a UPT formula in which all the multiplication gates have fain-in exactly 2 -- in other words, all the parse trees are {\em binary trees}. Henceforth, we will work with this additional structure for UPT formulas.
\end{definition}

Formulas with a bounded number of parse trees have been studied before \cite{lls19} in the non-commutative setting. \cite{lls19} proved an exponential lower bound for non-commutative circuits having at most exponentially many parse trees. While the lower bound that we prove in this work is only against formulas containing one parse tree, it is in the much more powerful commutative setting. UPT formulas are also related to regular formulas considered in \cite{KayalSS14}. A formula is said to be regular if it has alternating levels of addition and multiplication gates, and all gates at the same level have the same fan-in. It is easy to see that UPT formulas capture homogeneous formulas wherein the addition gates at the same level can have different fan-ins, and only the multiplication gates at the same level are restricted to having the same fan-ins. Hence UPT formulas are a generalisation of homogeneous regular formulas.

\begin{figure}
\begin{subfigure}{1\textwidth}
  \centering
  \begin{tikzpicture}[level 1/.style={sibling distance=40mm, level  distance=2cm},level 2/.style={sibling distance=30mm, level  distance=1.5cm} ,level 3/.style={sibling distance=20mm, level  distance=1.5cm}, level 4/.style={sibling distance=10mm, level  distance=1cm}]
        \node[circle, draw] {\textcolor{purple}{$+$}} 
            child {node[circle, draw] {\textcolor{purple}{$\times$} } edge from parent[stealth-]
            child {node[circle, draw] {\textcolor{purple}{$\times$}}
            child {node[rectangle, draw] {\textcolor{purple}{$x_3$}}}
            child {node[rectangle, draw] {\textcolor{purple}{$x_4$}}}}
            child {node[rectangle, draw] {\textcolor{purple}{$x_2$}}} }
            child {node[circle, draw] {\textcolor{purple}{$\times$} }
            edge from parent[stealth-]
            child {node[circle, draw] {\textcolor{purple}{$+$}}
            child {node[rectangle, draw]
            {\textcolor{purple}{$x_1$}}}
            child {node[rectangle, draw]
            {\textcolor{purple}{$x_2$}}}}
            child {node[circle, draw] {\textcolor{purple}{$\times$}}
            child {node[rectangle, draw] {\textcolor{purple}{$x_5$}}}
            child {node[rectangle, draw] {\textcolor{purple}{$x_3$}}}}}
            child {node[circle, draw, thin] {\textcolor{purple}{$\times$}} edge from parent[stealth-, ultra thick]
            child {node[rectangle, draw, thin] {\textcolor{purple}{$x_3$}}}
            child {node[circle, draw, thin] {\textcolor{purple}{$+$}}
            child {node[circle, draw, thin] {\textcolor{purple}{$\times$}}
            child {node[rectangle, draw, thin] {\textcolor{purple}{$x_5$}}}
            child {node[rectangle, draw, thin] {\textcolor{purple}{$x_3$}}}}
            child {node[circle, draw, thin] {\textcolor{purple}{$\times$}}
            edge from parent[stealth-, thin]
            child {node[circle, draw] {\textcolor{purple}{$+$}}
            child {node[rectangle, draw] {\textcolor{purple}{$x_6$}}}
            child {node[rectangle, draw] {\textcolor{purple}{$x_4$}}}}
            child {node[circle, draw] {\textcolor{purple}{$+$}}
            child [missing]
            child {node[rectangle, draw, thin] {\textcolor{purple}{$x_1$}}}}}}};
    \end{tikzpicture} 
  \caption{A UPT formula $C$ -- the dark edges produce a parse tree $\cP_1$}
  \label{fig:sub-first}
\end{subfigure}
\begin{subfigure}{.5\textwidth}
  \centering
  \begin{tikzpicture}
       \node [circle, draw] {} 
       child {node[rectangle, draw] {} edge from parent[-stealth]}
       child {node[circle, draw] {} edge from parent[-stealth]
       child {node[rectangle, draw] {}}
       child {node[rectangle, draw] {}}};
    \end{tikzpicture}
    ~~~~
    \begin{tikzpicture}
       \node [circle, draw] {}
       child {node[circle, draw] {} edge from parent[-stealth]
       child {node[rectangle, draw] {}}
       child {node[rectangle, draw] {}}}
       child {node[rectangle, draw] {} edge from parent[-stealth]} ;
    \end{tikzpicture}
  \caption{Two (isomorphic) parse trees of $C$: $\cP_1$ and $\cP_2$}
  \label{fig:sub-second}
\end{subfigure} 
\begin{subfigure}{.5\textwidth}
  \centering
     \begin{tikzpicture}
       \node [circle, draw] {} 
       child {node[rectangle, draw] {}edge from parent[-stealth]}
       child {node[circle, draw] {} edge from parent[-stealth]
       child {node[rectangle, draw] {}}
       child {node[rectangle, draw] {}}};
    \end{tikzpicture}
\caption{The canonical parse tree $\cT(C)$}
  \label{fig:sub-second}
\end{subfigure} 
    \caption{A UPT formula, its parse trees and canonical parse tree}
    \label{fig:upt}
\end{figure}
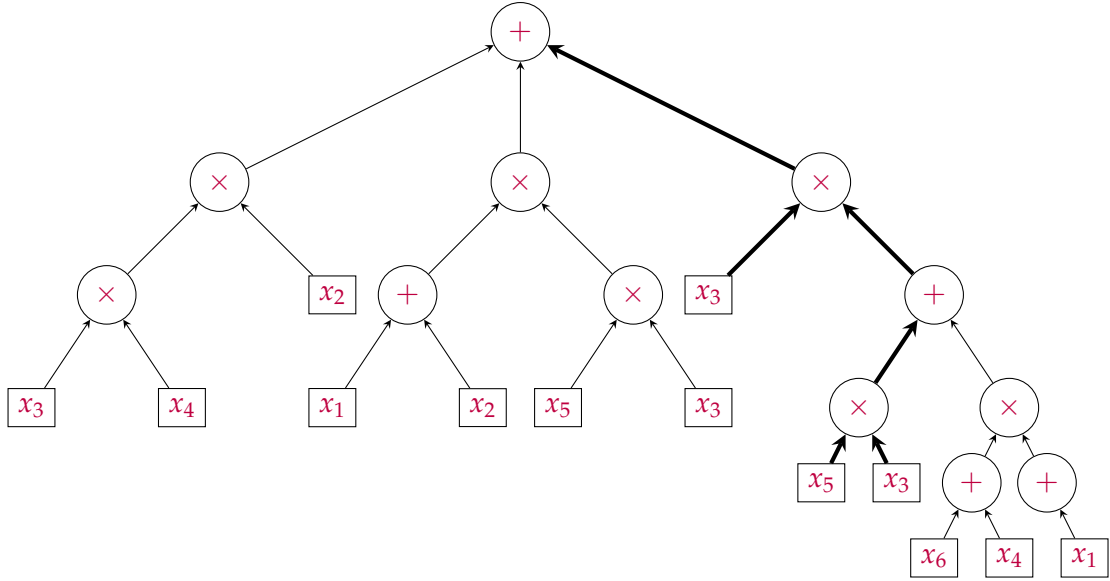
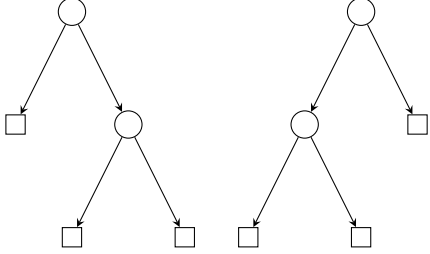
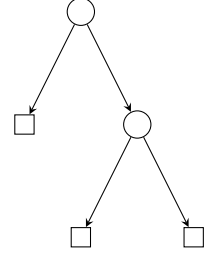

Figure~\ref{fig:upt} gives an example of a UPT formula and two of its parse trees, which can be seen to be isomorphic to each other.\\

\noindent {\bf Binary trees, isomorphism, and canonical trees.}
Unless stated otherwise, every binary tree $\cT$ which we consider in this article will be a rooted, directed (away from the root) binary tree in which all the internal nodes have a \emph{left child} and a \emph{right child}. For a node $v$ of $\cT$, $\cT_v$ denotes the subtree rooted at $v$ and $\leaves(v)$ denotes the number of leaves in $\cT_v$. 

	A binary tree $\cT$ is said to be \emph{right-heavy} if for all internal nodes $v$ with left child $v_L$ and right child $v_R$, we have $\leaves({v_L}) \le \leaves({v_R})$.
	Two binary trees $\cT$ and $\widetilde{\cT}$ are said to be \emph{equal} or {\em identical} if there is a bijection between their nodes preserving the \emph{(parent, left-child)} and \emph{(parent, right-child)} relations. They are said to be \emph{isomorphic} if there exists a bijection preserving the \emph{(parent, child)} relations.
	
	For any given binary tree $\cT$, we define its canonical tree $\canon(\cT)$ using the function in Algorithm~\ref{alg:canon} given in Appendix \ref{app:prelim}. For every node $v$ in $\cT$, that function swaps the subtrees rooted at its left and right children if the former has more leaves than the latter. The only properties of that function we need are mentioned in the following proposition; its proof is given in Section \ref{app:clm:prelim}. These properties can be verified to be true for the parse trees of the formula in Figure~\ref{fig:upt}. 

\begin{proposition}[{\bf Isomorphism means same canonical tree}]
    \label{clm:prelim}
    For any binary trees $\cT$ and $\widetilde{\cT}$,
    \begin{enumerate}
    \item $\canon(\cT)$ is right-heavy and is isomorphic to $\cT$. 
    \item $\cT$ and $\widetilde{\cT}$ are isomorphic if and only if $\canon(\cT) = \canon(\widetilde{\cT})$. Hence, $ \canon(\canon(\cT)) = \canon(\cT)$.
    \item Let $\phi$ denote an isomorphism between $\cT$ and $\canon(\cT)$. Then for a node $v$ in $\cT$, $\canon(\cT_v) = \canon(\cT)_{\phi(v)}$.
    \end{enumerate}
\end{proposition}

\begin{definition}[{\bf Canonical parse tree}]
    For a UPT formula $C$, we define its canonical parse tree as $\cT(C):=\canon(\cP)$, where $\cP$ is an arbitrary parse tree of $C$. The canonical parse tree is a binary tree and is well-defined as all parse trees of a UPT formula are isomorphic. 
\end{definition}

\subsection{Polynomial families} \label{subsec: poly families}

 The polynomials for which we prove the formula lower bounds are the Iterated Matrix Multiplication ($IMM$) polynomial and the Nisan-Wigderson design polynomial ($NW$). \\

\noindent {\bf Iterated Matrix Multiplication.}
    The iterated matrix multiplication, $IMM_{n,d}$ is a polynomial in $N = d\!\cdot\! n^2$ variables defined as the $(1,1)$-th entry of the matrix product of $d$ many $n \times n$ matrices such that the variables in each matrix and across different matrices are all distinct.

    To prove a lower bound for $IMM$, we analyze the shifted partials (and $\APP$) for a different, but related polynomial $P_{\vecw}$ that was introduced in \cite{lst1}. This polynomial is parameterized by a word $\vecw = ( w_1, \dots, w_d )$, a sequence of integers. It was shown in Lemma 22 of \cite{lst1} that $P_{\vecw}$ is a projection of $IMM$. Thus a lower bound for formulas computing $P_\vecw$ also gives a lower bound for formulas computing $IMM$. Both these polynomials can be seen to be homogeneous (in fact, they are set-multilinear), so they can indeed be computed by homogeneous as well as UPT formulas.
    
 \begin{definition}[{\bf Word polynomial $P_{\vecw}$ \cite{lst1}}]
    Given a word $\vecw=( w_1,\dots,w_d ) \in \Z^d$, let $\vecx(\vecw)$ be a tuple of $d$ pairwise disjoint sets of variables $\paren{\vecx_1(\vecw),\dots,\vecx_d(\vecw)}$ with $\abs{\vecx_i(\vecw)} = 2^{|w_i|}$ for all $i \in [d]$. We call a variable set $\vecx_i(\vecw)$ negative if $w_i < 0$ and positive otherwise. As the set sizes are powers of 2, we can map the variables in a set $\vecx_i(\vecw)$ to boolean strings of length $\abs{w_i}$. Let $\sigma:\vecx \to \set{0,1}^*$ be such a mapping.\footnote{Note that $\sigma$ may map a variable from $\vecx_i(\vecw)$ and a variable from $\vecx_j(\vecw)$ to the same string if $i \neq j$.} We extend the definition of $\sigma$ from variables to set-multilinear monomials as follows.
    
    Let $X=x_1\cdots x_r$ be a set-multilinear monomial where $x_i \in \vecx_{\phi(i)}(\vecw)$ and $\phi:[r] \to [d]$ be an increasing function -- in other words, the variables in $X$ are ordered based on the index of the corresponding variable sets. Then, we define a boolean string $\sigma(X):= \sigma(x_1) \circ \dots \circ \sigma(x_r)$, where $\circ$ denotes the concatenation of bits. 
    Let $\cM_+(\vecw)$ and $\cM_-(\vecw)$ denote the set of all (monic) set-multilinear monomials over all the positive sets and all the negative sets, respectively. For two Boolean strings $a,b$, we say $a \sim b$ if $a$ is a prefix of $b$ or vice versa.
    The word polynomial $P_{\vecw} \in \Fsm[\vecx(\vecw)]$ for a word $\vecw$ is defined as
    
    \begin{align*}
        P_{\vecw} := \sum \limits_{\substack{m_+ \in \cM_+(\vecw), ~m_- \in \cM_-(\vecw) \\ \sigma(m_+) ~\sim ~\sigma(m_-)}} m_+\cdot m_-.
    \end{align*}
  \end{definition}

Notice that if $\sum_{w_i \geq 0}|w_i| \leq \sum_{w_i < 0}|w_i|$, then for any $m_+ \in \cM_+$ and $m_- \in \cM_-$, $\sigma(m_+)$ will be a prefix of $\sigma(m_-)$ and if $\sum_{i \in [d]:w_i \geq 0}|w_i| > \sum_{i \in [d]:w_i < 0}|w_i|$, then for any $m_+ \in \cM_+$ and $m_- \in \cM_-$, $\sigma(m_-)$ will be a prefix of $\sigma(m_+)$. We will make use of the following lemma from \cite{lst1} which shows that $IMM$ is at least as hard as $P_{\vecw}$. For this, we recall the notion of {\em unbiased}-ness of $\vecw=(w_1,\dots,w_d)$ from \cite{lst1} -- we say that $\vecw$ is $h$\unbiased~if $\max_{i \in [d]} \abs{w_1 + \dots + w_i} \le h$.

 \begin{lemma}[Lemma 7 in~\cite{lst1}]
    \label{lem:imm-proj}
    Let $\vecw \in [-h..h]^d$ be $h$\unbiased.
    If for some $n \ge 2^h$, $IMM_{n, d}$ has a formula $C$ of product-depth $\Delta$ and size $s$, then $P_{\vecw}$ has a formula $C'$ of product-depth at most $\Delta$ and size at most $s$.
    
    Moreover, if $C$ is homogeneous, then so is $C'$ and if $C$ is UPT, then so is $C'$ with the same canonical parse tree, i.e., $\cT(C') = \cT(C)$.\footnote{Although the lemma in \cite{lst1} is stated for set-multilinear circuits, it also applies to homogeneous formulas and UPT formulas (albeit with a mild blow-up in size) by the same argument.}
 \end{lemma}
 
 
 \noindent {\bf Nisan-Wigderson design polynomial.} For a prime power $q$ and $d \in \N$, let $\vecx = \{x_{1,1}, \ldots, x_{1,q}, \ldots, x_{d,1},\\ \ldots, x_{d,q}\}$. For any $k \in [d]$, the Nisan-Wigderson design polynomial on $qd$ variables, denoted by $NW_{q,d,k}$ or simply $NW$, is defined as follows:
    \[NW_{q,d,k} = \sum_{\substack{h(z) \in \F_q[z]:\\ \deg(h) < k}}~~\prod_{i \in [d]} x_{i, h(i)}.\]
The $IMM$ and the $NW$ polynomials, and their variants, have been extensively used to prove various circuit lower bounds \cite{NisanW97, KayalSS14, KayalLSS17, KumarS17, KayalS16, KayalST16, KayalST16a, ForbesKS16, ChillaraL019, KumarS19, GuptaST20, lst1, KushS22}.
 
	\section{Structure of the space of partials of a product}
\label{sec:der-structural}
In this section, we bound the partial derivative space of a product of homogeneous polynomials. In the following lemma, we show that the space of $k$-th order partial derivatives of a product of polynomials is contained in a sum of shifted partial spaces with shift $\ell_0$ and order of derivatives $k_0$ such that $k_0 + \frac{k}{d-k} \cdot \ell_0$ is `small'. Using this lemma, we upper bound the $\SP$ and $\APP$ measures of a product of homogeneous polynomials. These upper bounds are then used in Sections \ref{sec:low-depth-lb} and \ref{sec:upt-lb} for proving lower bounds for low-depth homogeneous formulas and UPT formulas respectively.

\begin{lemma}[{\bf Upper bounding the partials of a product}]\label{lem:pdspace}
	Let $n$ and $t$ be positive integers and $Q_1,\dots,Q_t$ be non-constant, homogeneous polynomials in $\F[\vecx]$ with degrees $d_1,\dots,d_t$ respectively. Let $d := \deg(Q_1\cdots Q_t) = \sum \limits_{i=1}^t d_i$ and $k < d$ be a non-negative integer. Then,\footnote{As the ranges $k_0\in[0..k]$ and $\ell_0 \in [0..(d-k)]$ follow from $k_0 + \frac{k}{d-k} \cdot \ell_0 \le  k$,  we do not explicitly mention them in the proof.}
	$$\spacespanned{\partialf^{k}\paren{Q_1\cdots Q_t}}  \subseteq
	\sum \limits_{\substack{S \subseteq [t], ~k_0 \in [0..k],~\ell_0 \in [0..(d-k)], \\ k_0 + \frac{k}{d-k}\cdot \ell_0 ~\le~ k - \entropy_k(d_1,\dots,d_t)}} \spacespanned{ \vecx^{\ell_0}\cdot\partialf^{k_0}\paren{\prod_{i\in S} Q_i}}.$$
	\label{lem:der-structure}
\end{lemma}
\begin{proof} We first give some intuition about the proof. Let $m$ be any multilinear monomial (i.e., no variable appears more than once) of degree $k$, and $X$ be the corresponding set of variables. Then, by the product rule $\partial_X(Q_1\cdots Q_t)$ can be expressed as the sum
\begin{equation}\sum_{\substack{\paren{X_1,\dots,X_t}: \\ X_1 \sqcup \dots \sqcup X_t=X}} \partial_{X_{1}}Q_1 \cdots \partial_{X_{t}}Q_t.
\label{eqn:a-sum}\end{equation} 
Note that since the sizes of $X_i$'s should sum to $k$ and the degrees of $Q_i$'s should sum to $d$ in each term of the above summation, some factors are differentiated  `a lot' while  the others are differentiated only `a little'. More specifically, if $\abs{X_i} > \frac{k}{d}\!\cdot\!d_i$, we use the fact that $\partial_{X_i}Q_i \in \spacespanned{\vecx^{d_i-\abs{X_i}}}$ and otherwise we use $\partial_{X_i}Q_i \in \spacespanned{\partialf^{\abs{X_i}} Q_i}$  to conclude that 
$$\partial_{X_1}Q_1 \cdots \partial_{X_t}Q_t \in \spacespanned{\vecx^{\sum_{i \in  \overline{S}} \ell_{0,i}}\cdot \prod_{i \in S} \partialf^{k_{0,i}}Q_i}, $$
where $S:=\set{i \in [t]: |X_i| \le \frac{k}{d}\!\cdot\!d_i} $, $\overline{S} = [t]\setminus S$, $\ell_{0,i} = d_i - |X_i|$ for all $i \in \overline{S}$, and $k_{0,i} = |X_i|$ for all $i \in S$. By the nature of our construction, we can show that $k_0 + \frac{k}{d-k}\!\cdot\!\ell_0 \le k - \entropy_k(d_1,\dots,d_t)$, where $k_0:= \sum_{i \in S} k_{0,i}$ and $\ell_0 := \sum_{i \in \overline{S}} \ell_{0,i}$  (see the calculations at the end of the proof). Now suppose it holds that $\prod_{i \in S} \partialf^{k_{0,i}} Q_i = \partialf^{k_0} \prod_{i\in S} Q_i$. In such a case, we would get the space required in the R.H.S. of the lemma statement, and we would be done. However, this assumption need not be true if $|S| \geq 1$. To get around this issue, we employ an inductive argument on the size of $S$ (see Claim~\ref{clm:cx}). For this argument, it will be helpful to combine certain terms in the sum \eqref{eqn:a-sum} depending on the set of factors that are differentiated a `lot' (see Claim~\ref{clm:ders-as-rs}). We now present the proof in full detail. Since, in general, the variables in $m$ need not be distinct, it will be convenient to think of degree $k$ monomials over $\vecx$ as maps from $[k]$ to $\vecx$.

	For a function $\mu:P\to \vecx$ and any $P' \subseteq P$, recall that $\mu(P')$ refers to the multiset of images of the elements of $P'$ under $\mu$. Thus $|\mu(P')| = |P'|$. Let $\cV$ be the set of polynomials on the R.H.S., i.e., 
	
	\begin{align*}
		\cV := \sum \limits_{\substack{S \subseteq [t], \ k_0 + \frac{k}{d-k}\cdot \ell_0 \\ \le~ k - \entropy_k(d_1,\dots,d_t)}} \spacespanned{ \vecx^{\ell_0}\cdot\partialf^{k_0}\paren{\prod_{i\in S} Q_i}}.
	\end{align*}
	
	\noindent We now argue that for an arbitrary total function $\mu:[k] \to \vecx$, $\partial_{\mu([k])} \paren{\prod \limits_{i\in [t]} Q_i} \in \cV$; the lemma then follows immediately. We use the following identity which is a direct consequence of the product rule for derivatives (Proposition~\ref{prop:ders}):
	
	\begin{align*}
		\partial_{\mu([k])} \paren{\prod \limits_{i\in [t]} Q_i} = \sum_{\substack{\kappa:[t]\to 2^{[k]}\text{~s.t.~}\\ \sqcup_{i \in [t]} \kappa_i = [k]}} \prod \limits_{i \in [t]} {\partial_{\mu(\kappa_i)} Q_i}.
	\end{align*}
	
	\noindent In fact, the product rule yields something general: for any $P \subseteq [k]$, function $\mu:P \to \vecx$, and $S \subseteq [t]$,
	
	\begin{align}
		\partial_{\mu(P)} \paren{\prod \limits_{i\in S} Q_i} = \sum_{\substack{\kappa:S \to 2^P\text{~s.t.~}\\\sqcup_{i \in S} \kappa_i = P}} \prod \limits_{i \in S} {\partial_{\mu(\kappa_i)} Q_i}.
		\label{eqn:pdt-rule}
	\end{align}
	
    \noindent In the above identities we have used $\kappa_i$ as a shorthand for $\kappa(i)$; we shall also do so for the rest of this section.
	
	For an arbitrary $S \subseteq [t]$, recall that we denote $\overline{S} = [t] \setminus S$. Let $\widetilde{\kappa}:\overline{S} \to 2^{[k]}$ be such that $|\widetilde{\kappa}_i| > \frac{k}{d}\!\cdot\! d_i$ for all $i\in \overline{S}$. Then we define a polynomial $R_{S,\widetilde{\kappa}}$ as 
	
	\begin{align}
		R_{S,\widetilde{\kappa}} := \sum_{\substack{\kappa:[t]\to 2^{[k]}\text{~s.t.~}\\  \kappa\text{~extends~}\widetilde{\kappa} \\ \sqcup_{i \in [t]} \kappa_i = [k] \\ \forall i\in S, ~|\kappa_i| \le \frac{k}{d}\cdot d_i}} \prod \limits_{i \in [t]} \partial_{\mu(\kappa_i)}Q_i
		\label{eqn:poly-defn}.
	\end{align}
	
	\noindent The idea is to express any $k$-th order partial derivative of the product $Q_1\cdots Q_t$ in terms of $R_{S,\widetilde{\kappa}}$. Indeed we have the following claim; its proof (which uses the product rule for derivatives) can be found in Section \ref{app:clm:ders-as-rs}.
	
	\begin{claim}
	\label{clm:ders-as-rs}
	    \begin{align}
		\partial_{\mu([k])} \paren{\prod \limits_{i\in [t]} Q_i}
		 & = \sum_{\substack{S \subseteq [t]}}~\sum_{\substack{\widetilde{\kappa}:\overline{S}\to 2^{[k]}\text{ s.t. }\\ \forall i \in \overline{S}, |\widetilde{\kappa}_i|>\frac{k}{d}\cdot d_i}} R_{S,\widetilde{\kappa}}. \nonumber
	    \end{align}
	\end{claim}
    
   \noindent Hence, to show that $\partial_{\mu([k])}\paren{Q_1\cdots Q_t}\in \cV$, it suffices to argue that the polynomials $R_{S,\widetilde{\kappa}}$ are in $\cV$.  We show this by induction on the size of $S$.  In the base case of $|S| = 0$, there does not exist any function $\kappa:[t]\to 2^{[k]}$ that extends $\widetilde{\kappa}$ such that $\set{i\in [t]: |\kappa_i| \le \frac{k}{d}\cdot d_i} = S$ and $\sqcup_{i \in [t]} \kappa_i = [k]$. This is so because $|\kappa_i| = |\widetilde{\kappa}_i| > \frac{k}{d}\cdot d_i$ for all $i \in [t]$ implies that $\sum \limits_{i \in [t]} |\kappa_i| > \sum \limits_{i \in [t]}\frac{k}{d}\cdot d_i = k$, and hence $\sqcup_{i \in [t]} \kappa_i \neq [k]$. So by definition, $R_{S, \widetilde{\kappa}} = 0 \in \cV$.
	
    Suppose that $R_{T, \kappa'} \in \cV$ for all $T \subseteq [n]$ such that $|T| < |S|$.  Let $\widetilde{\kappa}:\overline{S} \to 2^{[k]}$ be any function such that $|\widetilde{\kappa}_i| > \frac{k}{d}\!\cdot\! d_i$ for all $i\in \overline{S}$, and let $\kappa:[t] \to 2^{[k]}$ be a function that extends $\widetilde{\kappa}$ such that
	\begin{align}\label{eqn:s-pi}
		\sqcup_{i \in [t]} \kappa_i = [k] ~ \text{~and~}~
		\set{i\in[t]: \abs{\kappa_i} \le \frac{k}{d}\cdot d_i} = S.
	\end{align}
	Denoting $\sqcup_{i \in \overline{S}} \kappa_i$ by $P_{\overline{S}}$ and $\sqcup_{i \in {S}} \kappa_i$ by $P_{{S}}$,
	
	\begin{flalign*}
		&& \partial_{\mu(\sqcup_{i \in S} \kappa_i)} \prod \limits_{i \in S}Q_i & = \partial_{\mu(P_S)} \prod_{i \in S} Q_i\\
		&&  &= \sum_{\substack{\kappa':S \to 2^{P_S}\text{~s.t.~}\\\sqcup_{i \in S} \kappa'_i = P_S}} \prod \limits_{i \in S} {\partial_{\mu(\kappa'_i)} Q_i}. && \text{(from Equation \eqref{eqn:pdt-rule})}
	\end{flalign*}
	

    \noindent For $U_{S,\kappa} \in \F[\vecx]$ defined as $U_{S, \kappa}:= \paren{\partial_{\mu(P_S)} \prod \limits_{i \in S}Q_i}\cdot \prod \limits_{i \in \overline{S}} \partial_{\mu(\kappa_i)} Q_i$, we have the following claim. It is proved in Section \ref{app:clm:cx}.

	\begin{claim}
	    \label{clm:cx}
	    \begin{align*}
	        R_{S,\widetilde{\kappa}} = U_{S,\kappa} -  \sum \limits_{\substack{T \subsetneq S \text{~and~} \kappa'':S\setminus T \to 2^{P_S} \text{~s.t.}\\ \forall i \in S \setminus T, |\kappa''_i| > \frac{k}{d}\cdot d_i}} R_{T,\kappa'' \sqcup \widetilde{\kappa}}.
	    \end{align*}
	\end{claim}
	
	\noindent When $T \subsetneq S$, by the induction hypothesis, all the terms $R_{T,\kappa'' \sqcup \widetilde{\kappa}}$ in the above expression  are in $\cV$. Therefore, to conclude that $R_{S,\widetilde{\kappa}}\in \cV$, it suffices to show that $U_{S,\kappa} \in \cV$.	From its definition, note that $U_{S,\kappa} \in \spacespanned{ \partialf^{k_0}\paren{\prod \limits_{i\in S} Q_i}\cdot \vecx^{\ell_0} }$ where $k_0 := \abs{\mu(P_S)} = \abs{P_S} = \sum \limits_{i \in S} |\kappa_i|$ and $\ell_0:= \sum \limits_{i \in \overline{S}} \deg(\partial_{\mu(\kappa_i)} Q_i)= \sum \limits_{i \in \overline{S}} (d_i - |\kappa_i|)$. Also,
	
	\begin{align*}
		k - k_0 - \frac{k}{d-k}\cdot \ell_0 & = k - \sum \limits_{i \in S} |\kappa_i| - \frac{k}{d-k}\cdot \sum \limits_{i \in \overline{S}} (d_i - |\kappa_i|) \\
		& = \sum \limits_{i \in \overline{S}} |\kappa_i| - \frac{k}{d-k}\cdot \sum \limits_{i \in \overline{S}} (d_i - |\kappa_i|) \tag{as from \eqref{eqn:s-pi},  $\kappa_1,\dots,\kappa_t$ form a partition of $[k]$ }\\
		& =  \sum \limits_{i \in \overline{S}} |\kappa_i| - \frac{k}{d-k}\cdot (d_i - |\kappa_i|)\\
		& = \sum \limits_{i \in \overline{S}} \frac{d}{d-k}\cdot \paren{  |\kappa_i| - \frac{k}{d}\cdot d_i}\\
		& \ge  \sum \limits_{i \in \overline{S}} |\kappa_i| - \frac{k}{d}.d_i \tag{using $d \ge d-k$ and $|\kappa_i| > \frac{k}{d}\cdot d_i$ iff $i \in \overline{S}$}\\
		& = \frac{1}{2} \paren{\sum \limits_{i \in \overline{S}}  |\kappa_i| - \frac{k}{d}\cdot d_i}  + \frac{1}{2} \paren{\sum \limits_{i \in S}  |\kappa_i| - \frac{k}{d}\cdot d_i}  \\
            &\quad+  \frac{1}{2} \paren{\sum \limits_{i \in \overline{S}}  |\kappa_i| - \frac{k}{d}\cdot d_i}  - \frac{1}{2} \paren{\sum \limits_{i \in S}  |\kappa_i| - \frac{k}{d}\cdot d_i} \\
		& = \frac{1}{2} \paren{\sum \limits_{i \in [t]}  |\kappa_i| - \frac{k}{d}\cdot d_i} +  \frac{1}{2} \paren{\sum \limits_{i \in \overline{S}}  |\kappa_i| - \frac{k}{d}\cdot d_i}  - \frac{1}{2} \paren{\sum \limits_{i \in S}  |\kappa_i| - \frac{k}{d}\cdot d_i}\\
		& = \frac{1}{2} \paren{k-\frac{k}{d}\cdot d} + \frac{1}{2} \paren{\sum \limits_{i \in \overline{S}}  |\kappa_i| - \frac{k}{d}\cdot d_i} - \frac{1}{2} \paren{\sum \limits_{i \in S}  |\kappa_i| - \frac{k}{d}\cdot d_i} \tag{since $|\kappa_i|$'s sum to $k$ and $d_i$'s sum to $d$}\\
		& = \frac{1}{2}\cdot \sum \limits_{i \in [t]} \abs{|\kappa_i| - \frac{k}{d}\cdot d_i} \tag{from ~\eqref{eqn:s-pi}}\\
		& \ge \entropy_k(d_1,\dots,d_t) \tag{from definition of $\entropy$}.
	\end{align*}
	Hence, $U_{S, \kappa} \in \spacespanned {\vecx ^{\ell_0}\cdot \partialf^{k_0} \paren{\prod \limits_{i \in S} Q_i} } \subseteq \cV$ as $k_0 + \frac{k}{d-k}\cdot \ell_0 \le k - \entropy_k(d_1,\dots,d_t)$.
\end{proof}

	We now use the above lemma to upper bound the shifted partials and affine projections of partials measures of a product of polynomials.
	
\begin{lemma}[{\bf Upper bounding $\SP$ and $\APP$ of a product}]
	Let $Q=Q_1\cdots Q_t$ be a homogeneous polynomial in $ \F[x_1,\dots,x_n]$ of degree $d=d_1 + \dots +d_t \ge 1$, where $Q_i$ is homogeneous and $d_i:=\deg(Q_i)$ for $i\in[t]$. Then, for any non-negative integers $k < d$, $\ell \ge 0$, and $n_0 \le n$,
	\begin{enumerate}
	    \item $$\SP_{k, \ell}(Q) \le 2^t\cdot d^{2}\cdot \max \limits_{\substack{k_0, \ell_0 \ge 0 \\ k_0 + \frac{k}{d-k}\cdot \ell_0~\le~k - {\entropy_k(d_1,\dots,d_t)}}} M(n, k_0)\cdot M(n,\ell_0 + \ell),$$
	    \item $$\APP_{k, n_0}(Q) \le 2^t\cdot d^{2}\cdot \max \limits_{\substack{k_0, \ell_0 \ge 0 \\ k_0 + \frac{k}{d-k}\cdot \ell_0~\le~k - {\entropy_k(d_1,\dots,d_t)}}} M(n, k_0)\cdot M(n_0,\ell_0).$$
	\end{enumerate}
	\label{lem:measure-upper-bd}
\end{lemma}

\begin{proof}
	
	We will first upper bound the shifted partials measure.
	From Lemma~\ref{lem:der-structure}, we know that 
	
	\begin{align*}
	    \spacespanned{\partialf^k \paren{Q_1\cdots Q_t}} \subseteq \sum \limits_{\substack{S \subseteq [t]; ~k_0, \ell_0 \ge 0 \\ k_0 + \frac{k}{d-k}\cdot \ell_0 \le k - \entropy_k(d_1,\dots,d_t)}} \spacespanned{\vecx^{\ell_0}\cdot \partialf^{k_0}\paren{ \prod \limits_{i \in S} Q_i}}.
	\end{align*}
	Hence, 
	
	\begin{align}
	    \spacespanned{\vecx^\ell\cdot \partialf^k \paren{Q_1\cdots Q_t}} \subseteq \sum \limits_{\substack{S \subseteq [t]; ~k_0, \ell_0 \ge 0 \\ k_0 + \frac{k}{d-k}\cdot \ell_0 \le k - \entropy_k(d_1,\dots,d_t)}} \spacespanned{\vecx^{\ell_0 + \ell}\cdot \partialf^{k_0}\paren{ \prod \limits_{i \in S} Q_i}}.
	    \label{eqn:der-struct}
	\end{align}
	For a fixed $S\subseteq [t]$ and $k_0,\ell_0$, since $\spacespanned{\vecx^{\ell_0 + \ell}\cdot \partialf^{k_0}\paren{\prod \limits_{i \in S} Q_i}} \subseteq \spacespanned{\vecx^{\ell_0 + \ell}} \cdot \spacespanned{\partialf^{k_0} \cdot \paren{\prod \limits_{i \in S} Q_i}} $, and $\dim \spacespanned{\vecx^{\ell_0 + \ell}} \le \abs{\vecx^{\ell_0 + \ell}} = M(n, \ell_0 + \ell)$ and $\dim \spacespanned{\partialf^{k_0} \paren{\prod \limits_{i \in S} Q_i}} \le \abs{\partialf^{k_0} \paren{\prod \limits_{i \in S} Q_i} } \le \abs{ \vecx^{k_0}} = M(n, k_0)$, we have,
 
    $$\dim \spacespanned{\vecx^{\ell_0 + \ell}\cdot  \partialf^{k_0}\paren{\prod \limits_{i \in S} Q_i}} \le \dim \spacespanned{\vecx^{\ell_0 + \ell}} \cdot  \dim\spacespanned{\partialf^{k_0} \paren{\prod \limits_{i \in S} Q_i}}  \le M(n,\ell_0 + \ell) \cdot M(n,k_0).$$
    
	Adding up the above upper bound over all the $2^t\!\cdot\! d^2$ possible combinations of $S \subseteq [t]$, $k_0\in [0..k]$, and $\ell_0 \in [0..(d-k)]$ in ~\eqref{eqn:der-struct},  we get, 
	$$\SP_{k,\ell}(Q) = \dim \spacespanned{\vecx^\ell\cdot \partialf^k \paren{Q_1\cdots Q_t}} \le 2^t\cdot d^2\cdot  \max \limits_{\substack{k_0, \ell_0 \ge 0 \\ k_0 + \frac{k}{d-k}\cdot \ell_0 \le k - {\entropy_k(d_1,\dots,d_t)}}} M(n, k_0)\cdot M(n,\ell_0 + \ell).$$ 
	
    The details for an upper bound on $\APP$ are similar and can be found in Section \ref{app:lem:measure-upper-bd}.
\end{proof}

	\section{Lower bound for low-depth homogeneous formulas} 
\label{sec:low-depth-lb}

In this section, we present a superpolynomial lower bound for ``low-depth'' homogeneous formulas computing the $IMM$ and $NW$ polynomials. We begin by proving a structural result for homogeneous formulas.

\subsection{Decomposition of low-depth formulas} \label{subsec:homo_decomopositon_lemma}
 We show that any homogeneous formula can be decomposed as a sum of products of homogeneous polynomials of lower degrees, where the number of summands is bounded by the number of gates in the original formula. The decomposition lemma given below bears some resemblance to a decomposition of homogeneous formulas in \cite{HrubesY11b}. In the decomposition in \cite{HrubesY11b}, the degrees of the factors of every summand roughly form a geometric sequence, and hence each summand is a product of a `large' number of factors. Here we show that each summand has `many' low-degree factors. While the lower bound argument in \cite{lst1} does not explicitly make use of such a decomposition, their inductive argument can be formulated as a depth-reduction or decomposition lemma (with slightly different thresholds for the degrees).

\begin{lemma}[{\bf Decomposition of low-depth formulas}]
	\label{lem:low-depth-structural}
	Suppose $C$ is a homogeneous formula of product-depth $\Delta \ge 1$ computing a homogeneous polynomial in $\F[x_1, \ldots, x_n]$ of degree \emph{at least} $d > 0$. 
	Then, there exist homogeneous polynomials $\set{Q_{i,j}}_{i,j}$ in  $\F[x_1, \ldots, x_n]$  such that
	\begin{enumerate} 
		\item $C = \sum \limits_{i=1}^{s} Q_{i,1} \cdots Q_{i,t_i}$,  for some $s \le \size(C)$, and
		\item for all $i\in [s]$, 
		either $$\abs{\set{j\in [t_i]:\deg(Q_{i,j}) =1}} \ge d^{2^{1-\Delta}}~\text{, or}$$
		$$\abs{\set{j\in[t_i]:\deg(Q_{i,j}) \approx_2  d^{2^{1-\delta}}}} \ge d^{2^{1-\delta}}-1~\text{~, for some $\delta \in [2..\Delta]$.}$$
		\label{item:lds-2}
	\end{enumerate}
\end{lemma}

\begin{proof}
    The decomposition is constructed inductively -- at addition gates, we simply add the decompositions of the smaller sub-formulas, whereas the multiplication gates need to be handled more carefully. Consider a multiplication gate $Q_1 \times \dots \times Q_t$. If all the factors ($Q_i$'s) have `low' degrees, we use this expression directly to construct the decomposition. Otherwise, we go deeper into a factor which has a `large' degree, but do not expand the other factors. The thresholds to decide whether a factor is of `low' degree may appear arbitrary (and are indeed so) for this lemma, but we fix them to be $d^{2^{1-\delta}}$ for $\delta \in [2..\Delta]$ as these give us the desired lower bounds.

	Without loss of generality, we may assume that $C$ has alternate layers of addition and multiplication gates. Further, we can assume that the degrees of the polynomials computed by all the multiplication gates that feed into an addition gate are the same as the degree of the polynomial computed by that addition gate. This is so because disconnecting all the multiplication gates that compute polynomials of other degrees does not affect the polynomial computed by the addition gate. Also, for brevity, we will ignore the edge weights in $C$, i.e., we assume that all the field constants on the edges are equal to 1. As scaling with constant factors does not affect the homogeneity of polynomials, this is a valid assumption. Let
	
	$$C = \sum \limits_{i=1}^{u} C_i \text{ , and for $i \in [u],$ } C_i = \prod \limits_{j=1}^{u_i} C_{i,j},$$
	where $u$ and $\set{u_i}_i$ are integers and $\set{C_{i,j}}_{i,j}$ are (homogeneous) sub-formulas of $C$ of product-depth $\Delta-1$. The proof of this lemma is by induction on the product-depth. For $\Delta=1$, for all $i\in[u]$, we have $u_i \geq d$ and for all $j\in\brac{d}$, $\deg(C_{i,j})=1$ , so both the conditions in the lemma statement are met for $Q_{i,j}:=C_{i,j}$.
 
    Suppose that the lemma is true for all homogeneous formulas of product-depth at most $\Delta - 1$, $\Delta \geq 2$. For a formula $C$ with product-depth $\Delta$, we consider the term $C_{i,1}\cdots C_{i,u_i}$ for an arbitrary $i\in[u]$ and analyze the following two cases. \\
	
	\noindent \textbf{Case 1:} There exists some $j^* \in [u_i]$ such that $\deg(C_{i,j^*}) \ge \sqrt{d}$. As the product-depth of $C_{i, j^*}$ is at most $\Delta - 1$, we have the following expression for the polynomial computed by $C_{i,j^*}$ from the induction hypothesis:
	
	\begin{equation}
		C_{i,j^*} = \sum \limits_{\tilde{i}=1}^{\tilde{s}_i} \widetilde{Q}_{i,\tilde{i},1} \cdots \widetilde{Q}_{i,\tilde{i},\tilde{t}_{\tilde{i}}}~ ,
		\label{eqn:delta-1}
	\end{equation} 
	where
        \begin{equation} 
		  \tilde{s}_i \le     \size(C_{i,j^*}) \le \size(C_i),
		\label{eqn:size-bound}
	\end{equation} 
 and $\set{\widetilde{Q}_{i,\tilde{i},\tilde{j}}}_{i,\tilde{i},\tilde{j}}$ are homogeneous polynomials such that for all $\tilde{i}\in \tilde{s}_i$, either

        \begin{eqnarray}
            \abs{\set{\tilde{j}\in [\tilde{t}_{\tilde{i}}]:\deg(\widetilde{Q}_{i,\tilde{i},\tilde{j}}) =1}} \ge \sqrt{d}^{2^{1-(\Delta-1)}} = d^{2^{1-\Delta}}~,\text{or}\label{eqn:linear-forms} \\
            \abs{\set{\tilde{j}\in [\tilde{t}_{\tilde{i}}]:\deg(\widetilde{Q}_{i,\tilde{i},\tilde{j}}) \approx_2 \sqrt{d}^{2^{1-\delta}}}} \ge \sqrt{d}^{2^{1-\delta}}-1~\text{~, for some $\delta \in [2..(\Delta-1)].$} \label{eqn:some-delta}
        \end{eqnarray}
	   
	\noindent Note that since $\sqrt{d}^{2^{1-\delta}}=d^{2^{1-(\delta+1)}}$, ~\eqref{eqn:some-delta} is equivalent to 
	
	\begin{equation}
		\abs{\set{\tilde{j}\in [\tilde{t}_{\tilde{i}}]:\deg(\widetilde{Q}_{i,\tilde{i},\tilde{j}}) \approx_2 d^{2^{1-\delta}}}} \ge d^{2^{1-\delta}}-1~\text{~, for some $\delta \in [3..\Delta].$} \label{eqn:some-delta-plus-1}
	\end{equation}
	Indeed, when $\Delta =2$, the above scenario never arises and the number of linear factors is `large', i.e.,~\eqref{eqn:linear-forms} holds. 
	Denoting $\prod \limits_{j\in [u_i] \setminus \set{j^*}} C_{i,j}$ by $D_{i,j^*}$ and using ~\eqref{eqn:delta-1}, we have
	
	\begin{equation}
		C_i = C_{i,1}\cdots C_{i,u_i} = C_{i,j^*}.D_{i,j^*} = \sum \limits_{\tilde{i}=1}^{\tilde{s}_i} \widetilde{Q}_{i,\tilde{i},1} \cdots \widetilde{Q}_{i,\tilde{i},\tilde{t}_{\tilde{i}}}\cdot D_{i,j^*}.
		\label{eqn:case-1}
	\end{equation}
	Thus, we are able to decompose the sub-formula $C_i$ as a sum of at most $\size(C_i)$ many products.\\
	
	\noindent \textbf{Case 2:} For all $j\in[u_i]$, $\deg(C_{i,j}) < \sqrt{d}$. Consider the polynomials computed by $C_{i,1},\dots,C_{i,u_i}$. Suppose there exists $j_1 \ne j_2 \in [u_i]$ such that $\deg(C_{i,j_1}) < \frac{\sqrt{d}}{2}$ and $\deg(C_{i,j_2}) < \frac{\sqrt{d}}{2}$. Then $\deg(C_{i,j_1}\cdot C_{i,j_2}) < \sqrt{d}$. By repeatedly combining such low degree factors, we can express $C_i= C_{i,1} \cdots C_{i,u_i}$ as
	
	\begin{equation}
		C_i = D_{i,1}\cdots D_{i,v_i}~,
		\label{eqn:low-degree-case}
	\end{equation}
	
	\noindent where $\set{D_{i,j}}_{i,j}$ are homogeneous polynomials such that for all $j \in [v_i]$, we have $\deg(D_{i,j}) < \sqrt{d}$ and there exists at most one  index $j^* \in [v_i]$ such that $\deg(D_{i,j^*}) < \frac{\sqrt{d}}{2}$. In other words, for at least $v_i-1$ indices $j\in [v_i]$, $\deg(D_{i,j}) \approx_2 \sqrt{d}$.
	Using the fact that $C$ is a homogeneous formula,
	
	\begin{equation*}
		d \le \deg(C) = \deg(C_i) = \sum \limits_{j=1}^{v_i} \deg(D_{i,j}) \le v_i\!\cdot\! \sqrt{d}.
		\label{eqn:vi}
	\end{equation*} 
	Therefore, the number of indices $j\in[v_i]$ such that $\deg(D_{i,j}) \approx_2 \sqrt{d}$ is at least $v_i -1 \ge \sqrt{d} -1$.	In other words, 
	
	\begin{equation}
		\abs{\set{j\in[v_i]:\deg(D_{i,j}) \approx_2 d^{2^{1-\delta}}}} \ge d^{2^{1-\delta}}-1~\text{~, for $\delta = 2.$}
		\label{eqn:delta-is-2}
	\end{equation}
	
    Now, expressing $C_i$ for each $i\in[u]$ using~\eqref{eqn:case-1} if $i$ falls under Case 1, and using ~\eqref{eqn:low-degree-case} if $i$ falls under Case 2, we get

    $$C = \sum \limits_{i=1}^u C_i =  \sum \limits_{i=1}^{s} Q_{i,1}  \cdots  Q_{i,t_i},$$
    for polynomials $\set{Q_{i,j}}_{i,j}$ that are  defined appropriately based on $\set{\widetilde{Q}_{i,\tilde{i},\tilde{j}}}_{i,\tilde{i},\tilde{j}}$ and $\set{D_{i,j}}_{i,j}$. Using ~\eqref{eqn:size-bound} and ~\eqref{eqn:low-degree-case}, we get that the number of terms is $s\le \sum \limits_{i=1}^u \size(C_i) \le \size(C)$. Item~\ref{item:lds-2} in the lemma statement directly follows from ~\eqref{eqn:linear-forms}, ~\eqref{eqn:delta-is-2}, or~\eqref{eqn:some-delta-plus-1}.
\end{proof}

\subsection{Low-depth formulas have high residue}

The following lemma gives us a value for the order of derivatives $k$ with respect to which low-depth formulas yield high $\entropy$. Its proof uses Lemma~\ref{lem:low-depth-structural}.

\begin{lemma}[{\bf Low-depth formulas have high residue}]
    Suppose $C$ is a homogeneous formula of product-depth $\Delta \ge 1$ computing a homogeneous polynomial in $\F[x_1,\dots,x_n]$ of degree $d$, where $d^{2^{1-\Delta}} = \omega(1)$.\footnote{In fact the lemma holds as long as $d^{2^{1-\Delta}}$ is greater than a large enough constant.} 
	Then, there exist homogeneous polynomials $\set{Q_{i,j}}_{i,j}$ in  $\F[x_1,\dots,x_n]$  such that $C = \sum \limits_{i=1}^{s} Q_{i,1} \cdots Q_{i,t_i}$,  for some $s \le \size(C)$.   Fixing an arbitrary $i\in[s]$, let $t:=t_i$ and define $d_j := \deg(Q_{i,j})$ for $j\in [t]$. Then, $\entropy_k(d_1,\dots,d_t) \ge \Omega \paren{d^{2^{1-\Delta}}}$, where $k:=\floor{\frac{\alpha\cdot d}{1+\alpha}}$, $\alpha := \sum\limits_{\nu = 0}^{\Delta-1} \frac{(-1)^\nu}{\tau^{2^\nu - 1}}$, and $\tau:=\floor{d^{2^{1-\Delta}}}$.
	\label{lem:low-k-gamma}
\end{lemma}

\begin{proof}
    We will show that the decomposition proven in Lemma~\ref{lem:low-depth-structural} itself satisfies the required property. We first establish a range for the value of $k$ (and $\alpha$) given in the lemma statement. We have $\alpha \le 1$ and 
	
	$$\alpha \ge \sum\limits_{\nu=0}^1 \frac{(-1)^\nu}{\tau^{2^\nu - 1}} = 1-\frac{1}{\tau} = 1 - \frac{1}{\floor{d^{2^{1-\Delta}}}} \ge \frac{1}{2}.$$ Hence, $k\in \brac{\floor{\frac{d}{3}}, \frac{d}{2}} \subseteq \brac{\frac{d}{4}, \frac{d}{2}}$ because $d = \omega(1)$. As $C$ computes a polynomial of degree $d \geq \tau^{2^{\Delta-1}}$, we can apply Item 2 of Lemma \ref{lem:low-depth-structural} to $C$ using $\tau^{2^{\Delta-1}}$ (rather than $d$) as the threshold. Thus, we have that at least one of the following two cases will hold.
	
		\noindent \textbf{Case 1:} $\abs{\set{j\in[t]:d_j=1}} \ge \paren{\tau^{2^{\Delta-1}}}^{2^{1-\Delta}}=\tau$. Then,
		\begin{align*}
		\entropy_k(d_1,\dots,d_t) 
		& =  \frac{1}{2} \cdot \min \limits_{k_1,\dots,k_t \in \Z} \sum \limits_{j\in[t]} \abs{ k_j - \frac{k}{d}\cdot d_j } \\
		& \ge  \frac{1}{2} \cdot \sum\limits_{j\in[t]} \min\set{\fract{\frac{k}{d}\cdot d_j}, 1-\fract{\frac{k}{d}\cdot d_j}} \\ 
		& \ge  \frac{1}{2} \cdot \sum \limits_{j\in[t]:d_j=1} \min\set{\fract{\frac{k}{d}\cdot d_j}, 1-\fract{\frac{k}{d}\cdot d_j}} \\
		& \ge  \frac{1}{2} \cdot \abs{\set{j\in [t]:d_j=1}}\cdot \min\set{\fract{\frac{k}{d}}, 1-\fract{\frac{k}{d}}} \\
		& \ge \tau/8. \tag{as $k/d \in [1/4, 1/2]$}
		\end{align*}
		
	    \noindent \textbf{Case 2:} $\abs{\set{j\in [t]: d_j \approx_2 \paren{\tau^{2^{\Delta -1}}}^{2^{1-\delta}} }} \ge  \paren{\tau^{2^{\Delta -1}}}^{2^{1-\delta}}-1$ for some $\delta \in [2..\Delta]$ (this case cannot occur when $\Delta < 2$). Equivalently, there exists a $\delta \in [0..(\Delta-2)]$ such that $$\abs{\set{j \in [t]: d_j \approx_2 \tau^{2^{\delta}}}} \ge \tau^{2^{\delta}}-1.$$
	   
	    \noindent Let $k_1,\dots,k_t$ be arbitrary non-negative integers such that $k_j \le d_j$ for all $j \in [t]$. Then for any $j\in[t]$ such that $d_j \approx_2 \tau^{2^{\delta}}$, we have
	
	\begin{flalign*}
	    &&\tau^{2^{\delta}-1}\cdot \abs{k_j - \frac{k\cdot d_j}{d}} 
	    & = \tau^{2^{\delta}-1}\cdot \abs{k_j - \frac{d_j}{d}\cdot \floor{\frac{\alpha\cdot d}{1+\alpha}}} \nonumber \\
	    && & \ge \tau^{2^{\delta}-1}\cdot \paren{\abs{k_j - \frac{d_j}{d}\cdot {\frac{\alpha\cdot d}{1+\alpha}}} - \frac{d_j}{d}\cdot \fract{\frac{\alpha\cdot d}{1+\alpha}}} \nonumber\\
	    && & \ge \tau^{2^\delta-1}\cdot \abs{k_j - \frac{\alpha\cdot d_j}{1+\alpha}} - \tau^{2^\delta-1}\cdot \frac{d_j}{d} \nonumber\\
	    && & \ge \tau^{2^\delta-1}\cdot \abs{k_j - \frac{\alpha\cdot d_j}{1+\alpha}} - \frac{\tau^{2^\delta-1}\cdot \tau^{2^\delta}}{d} && \text{(since $d_j \approx_2 \tau^{2^{\delta}}$)}\nonumber\\
	  && & \ge \tau^{2^\delta-1}\cdot \abs{k_j - \frac{\alpha\cdot d_j}{1+\alpha}} - \frac{\paren{d^{2^{1-\Delta}}}^{2^{\delta+1}-1}}{d}\nonumber \\
	    && & \ge \tau^{2^\delta-1}\cdot \abs{k_j - \frac{\alpha\cdot d_j}{1+\alpha}} - \frac{1}{d^{2^{1-\Delta}}} && \text{(as $\delta \le \Delta -2$)}\\
	    && & = \tau^{2^\delta-1}\cdot \abs{k_j - \frac{\alpha\cdot d_j}{1+\alpha}} - o(1) && \intertext{\raggedleft{(if $d^{2^{1-\Delta}} = O(1)$, then the lemma is not interesting)}}
	    && & \ge \frac{1}{2}\cdot \tau^{2^\delta-1}\cdot \abs{k_j - \alpha\cdot (d_j - k_j)} - o(1) && \text{(as $\alpha \le 1$)}  \numberthis\label{eqn:in-terms-eta}
	\end{flalign*}
	
	We use the following claim which is proved in Section \ref{app:entropy-sub}. For $j\in[t]$, let $m_j := d_j -k_j$, note that $m_j$ is a non-negative integer.
		
    \begin{claim}
        \label{clm:entropy-sub}
        $\eta := \tau^{2^\delta-1}\cdot \abs{k_j - \alpha\cdot m_j} \ge \Omega(1)$.
    \end{claim}
	
	Let $k_1, \ldots, k_t \in \Z$ be the such that $\sum \limits_{i=1}^t \abs{k_i - \frac{k}{d} \cdot d_i}$ is minimised. Hence, 
	\begin{align*}
	    \entropy_k(d_1,\dots,d_t) & \ge \frac{1}{2}\sum\limits_{j\in[t]:d_j \approx_2 \tau^{2^\delta}} \abs{k_j - \frac{k}{d}\cdot d_j}\\
	    & \ge \frac{1}{2} \cdot \abs{\set{j\in[t]:d_j \approx_2 \tau^{2^\delta}}}\cdot \min\limits_{j\in[t]:d_j \approx_2 \tau^{2^\delta}} \abs{k_j - \frac{k}{d}\cdot d_j}\\
	    & \ge \Omega\paren{\tau^{2^\delta}}\!\cdot\! \min\limits_{j\in[t]:d_j \approx_2 \tau^{2^\delta}} \abs{k_j - \frac{k}{d}.d_j}\\
	    & = \Omega\paren{\tau}\!\cdot\! \min\limits_{j\in[t]:d_j \approx_2 \tau^{2^\delta}} \tau^{{2^\delta-1}}\cdot \abs{k_j - \frac{k}{d}\cdot d_j}\\
	    & \ge \Omega(\tau)\!\cdot\! \paren{\frac{\eta}{2}-o(1)} \tag{using \eqref{eqn:in-terms-eta}}\\
	    & \ge \Omega\paren{\tau}.
	\end{align*}
	Therefore, $\entropy_k(d_1,\dots,d_t) \ge \Omega(\tau) \ge \Omega\paren{\frac{\tau + 1}{2}} \ge \Omega\paren{\tau + 1} \ge \Omega\paren{d^{2^{1-\Delta}}}$.
\end{proof}


\subsection{High residue implies lower bounds}
\label{subsec:high-residue-implies-lbs}
For a `random' homogeneous degree-$d$ polynomial in $\F[x_1,\dots,x_n]$, if the shift $\ell$ is not too large, we expect the shifted partials measure to be close to the maximum number of operators used to construct the shifted partials space, i.e., $M(n,k)\!\cdot\!M(n,\ell)$. In the lemma below, we derive a lower bound for such polynomials. Explicit examples of such polynomials are given in Section \ref{subsec: explicit hard polynomials}.

\begin{lemma}[{\bf High residue implies lower bounds}]
    Let $P=\sum \limits_{i=1}^s Q_{i,1}\cdots Q_{i,t_i}$ be a homogeneous polynomial in $\F[x_1,\dots,x_n]$ of degree $d$ where $\set{Q_{i,j}}_{i,j}$ are homogeneous and $\SP_{k,\ell}(P) \ge 2^{-O(d)}\!\cdot\! M(n,k)\!\cdot\! M(n,\ell)$ for some $1 \leq k < \frac{d}{2}, n_0 \le n$ and $\ell=\floor{\frac{n\cdot d}{n_0}}$ such that $d \le n_0 \approx {2(d-k)\!\cdot\! \paren{\frac{n}{k}}^{\frac{k}{d-k}}}$.\footnote{Even though $2(d-k)\paren{\frac{n}{k}}^{\frac{k}{d-k}}$ can be greater than $n$, in the analysis in Sections \ref{subsec: explicit hard polynomials}, \ref{subsec: homo putting together}, $n_0$ we will be chosen in such a way that it lies between $d$ and $n$.} If there exists a $\gamma > 0$ such that for all $i \in [s]$, 
    $$\entropy_k(\deg(Q_{i,1}), \dots, \deg(Q_{i,t_i})) \ge \gamma,$$ 
    then $s \ge 2^{-O(d)}\paren{\frac{n}{d}}^{\Omega\paren{\gamma}}$.
    \label{lem:sp-final-lb}
\end{lemma}

\begin{proof}
    Using Lemma~\ref{lem:measure-upper-bd} (Item 1) and the fact that $\SP$ is sub-additive (Proposition~\ref{prop:sub-additive}), we get
    $$\SP_{k,\ell}(P) \le \sum \limits_{i=1}^s \SP_{k,\ell}(Q_{i,1}\cdots Q_{i,t_i}) \le s\cdot 2^t \cdot d^2 \cdot \max \limits_{\substack{k_0,\ell_0 \ge 0 \\ k_0 + \frac{k}{d-k}\cdot \ell_0~\le~k - \gamma}} M(n,k_0)\cdot M(n,\ell+\ell_0),$$
    where $t:= \max_{i} t_i$ is at most $d$.
    On the other hand, by our assumption we have $\SP_{k,\ell}(P) \ge 2^{-O(d)}\cdot M(n,k)\cdot M(n,\ell)$.
    Putting these two together, we get for some integers $k_0\in [0..k],\ell_0 \in [0..(d-k)]$ satisfying
    \begin{align}
        k_0 + \frac{k}{d-k}\cdot \ell_0 \le k - \gamma,
        \label{eqn:d1-d2-upper-bd}
    \end{align}
    that,
    \begin{align}
        s &\ge   2^{-O(d)}\!\cdot\!2^{-t}\!\cdot\!d^{-2}\!\cdot\!\frac{M(n,k)\!\cdot\!M(n,\ell)}{M(n,k_0)\!\cdot\!M(n,\ell+\ell_0)}
		\label{eqn:the-lb}\\
		& \ge 2^{-O(d)}\!\cdot\!\frac{M(n,k)}{M(n,k_0)\!\cdot\!\paren{2n/\ell}^{\ell_0}} \nonumber\\
        \intertext{\raggedleft{(Lemma~\ref{lem:binomial-coeffs} (Item 2) is applicable as $n_0 \ge d$ implies that $n \ge \floor{\frac{n\!\cdot\! d}{n_0}}=\ell$; also $n_0 \leq n$ implies $\ell = \floor{\frac{nd}{n_0}}\geq d \geq \ell_0$)}}
		& \ge  2^{-O(d)}\!\cdot\!\frac{M(n,k)}{M(n,k_0)\!\cdot\!\paren{n_0/\ell_0}^{\ell_0}} \label{eqn:calcs}\\ 
        \intertext{\raggedleft{(because $2n/\ell \leq 4n_0/\ell_0$ as $\ell_0 \leq d$ and absorbing $4^{\ell_0}$ in $2^{-O(d)}$)}}
		& \ge 2^{-O(d)}\!\cdot\!\frac{(n/k)^k}{(6n/{k_0})^{k_0}\!\cdot\!(n_0/{\ell_0})^{\ell_0}}\nonumber\\ \intertext{\raggedleft{(assuming $k_0,\ell_0 \ne 0$ and using Lemma~\ref{lem:binomial-coeffs} (Item 1) as $n \ge n_0 \ge d \ge \max \set{{k_0},{\ell_0}}$; the analysis is easier if any of $k_0$, $\ell_0$ is $0$)}}
		& \ge  2^{-O(d)}\!\cdot\!\frac{(n/k)^k}{(n/{k_0})^{k_0}\!\cdot\!(n_0/{\ell_0})^{\ell_0}} \tag{since $k_0 = O(d)$}\\\nonumber
		& \ge 2^{-O(d)}\!\cdot\!\frac{(n/k)^k}{(n/{k_0})^{k_0}\!\cdot\!\paren{{\frac{2(d-k)}{\ell_0}\!\cdot\!(n/k)^{\frac{k}{d-k}}}}^{\ell_0}} \tag{substituting $n_0$}\\\nonumber
		& = 2^{-O(d)}\!\cdot\!\frac{(n/k)^k}{(n/{k_0})^{k_0}\!\cdot\!(\frac{d-k}{\ell_0})^{\ell_0}\!\cdot\!{(n/k)^{\frac{k\cdot \ell_0}{d-k}}}} \tag{as $\ell_0 = O(d)$}\\\nonumber
		& = 2^{-O(d)}\!\cdot\!\frac{(n/k)^k}{(n/{k})^{k_0}.(k/{k_0})^{k_0}.(\frac{d-k}{\ell_0})^{\ell_0}\!\cdot\!{(n/k)^{\frac{k\!\cdot\!\ell_0}{d-k}}}} \tag{as $n/k_0 = (n/k) . (k/k_0)$}\\\nonumber
		& = 2^{-O(d)}\!\cdot\!\frac{(n/k)^{k-k_0-\frac{k}{d-k}.\ell_0}}{(k/{k_0})^{k_0}\!\cdot\!(\frac{d-k}{\ell_0})^{\ell_0}}\\\nonumber
		& \ge 2^{-O(d)}.\frac{(n/d)^{\gamma}}{(d/{k_0})^{k_0}.(\frac{d}{\ell_0})^{\ell_0}} \tag{using $k < d$ and ~\eqref{eqn:d1-d2-upper-bd}}\\\nonumber
		& = 2^{-O(d)}\!\cdot\!\paren{\frac{n}{d}}^{\gamma}\!\cdot\!\paren{\frac{k_0}{d}}^{k_0}\!\cdot\!\paren{\frac{\ell_0}{d}}^{\ell_0}\\\nonumber
		& \ge 2^{-O(d)}\!\cdot\!\paren{\frac{n}{d}}^{\gamma}\cdot(e^{-1/e})^d\cdot(e^{-1/e})^d \tag{using $x^x \ge e^{-1/e}$ for $x > 0$}\\\nonumber
		& \ge 2^{-O(d)}\!\cdot\!\paren{\frac{n}{d}}^{\Omega(\gamma)}.   \nonumber 
    \end{align}
\end{proof}
We state an analogous lemma with $\APP$ instead of $\SP$ -- the proof given in Section \ref{app:lem:app-final-lb} is similar. 

\begin{lemma}[{\bf High residue implies lower bounds, using $\APP$}]\label{lem:ckt_APP}
    Let $P=\sum \limits_{i=1}^s Q_{i,1}\cdots Q_{i,t_i}$ be a homogeneous polynomial in $\F[x_1,\dots,x_n]$ of degree $d$ where $\set{Q_{i,j}}_{i,j}$ are homogeneous and $\APP_{k,n_0}(P) \ge 2^{-O(d)}\!\cdot\!M(n,k)$ for some $1 \leq k < \frac{d}{2}, n_0 \le n$ such that $d \le n_0 \approx {2(d-k).\paren{\frac{n}{k}}^{\frac{k}{d-k}}}$. If there exists a $\gamma > 0$ such that for all $i \in [s]$, 
    $$\entropy_k(\deg(Q_{i,1}), \dots, \deg(Q_{i,t_i})) \ge \gamma,$$ 
    then $s \ge 2^{-O(d)}\!\cdot\!\paren{\frac{n}{d}}^{\Omega(\gamma)}$.
    \label{lem:app-final-lb}
\end{lemma}

\begin{remark}
    In the above lemmas, although our lower bound appears as $2^{-O(d)}\!\cdot\!n^{\Omega(\gamma)}$, similar calculations actually give a lower bound of $2^{-O(k)}\!\cdot\!n^{\Omega(\gamma)}$ for any choice of $k$ and an appropriate choice of $\ell$ (or $n_0$ in the case of $\APP$). We do not differentiate between the two, as for our applications (i.e., low-depth circuits and UPT formulas), the value of $k$ we choose is $\Theta(d)$. Using our technique, one can recover the known  exponential lower bound for homogeneous depth-4 circuits with bottom fan-in upper bounded by $\sqrt{d}$ proved in \cite{KayalSS14, GargKS20} by noting that for $k=\Theta(\sqrt{d})$ the $\entropy$ is at least $\gamma = \Omega \paren{\sqrt{d}}$. Thus, we get a lower bound of $2^{-O(\sqrt{d})}\cdot n^{\Omega\paren{\sqrt{d}}} \ge n^{\Omega\paren{\sqrt{d}}}$ for some $d=n^{\Theta\paren{1}}$. This is not surprising as the lower bound proofs for this model in \cite{KayalSS14, GargKS20} also use the shifted partials and $\APP$ measures.
    
    Moreover, we observe that the factor of $2^{-O(k)}$ in our lower bounds is likely unavoidable for any choice of $k$ and $\ell$ (or $n_0$ in the case of $\APP$) \textit{using our current estimates for the complexity measures}. The first two factors in the R.H.S. of the inequality~\eqref{eqn:the-lb} are not the bottlenecks (at least for $\Delta =2$), as we can always make sure that $t \le \sqrt{d}$ and the $\SP$ measure is exactly equal to $M(n,k)\cdot M(n,\ell)$ without the extra $2^{-O(d)}$ factor. It is rather the ratio of the binomial coefficients where this factor gets introduced. We refer the reader to the discussion in Section \ref{app:2d-discuss} for more details.
    \label{rem:remove}
\end{remark}

\subsection{The hard polynomials} \label{subsec: explicit hard polynomials}
We shall prove our lower bound for the word polynomial $P_{\vecw}$ introduced in \cite{lst1} as well as for the Nisan-Wigderson design polynomial. In order to do this, we show that the $\SP$ and $\APP$ measures of $P_\vecw$ and the $\SP$ measure of $NW$ are large for suitable choices of $k$, $\ell$ and $n_0$.

\begin{lemma}[{\bf $P_\vecw$ as a hard polynomial}]
	For integers $h,d$ such that $h> 100$ and any parameter $k \in \brac{\frac{d}{30}, \frac{d}{2}}$, there exists an $h$\unbiased~word $\vecw \in [-h..h]^d$, integers $n_0 \le n$, $\ell=\floor{\frac{n\cdot d}{n_0}}$ such that $n_0 \approx 2(d-k)\!\cdot\! \paren{{\frac{n}{k}}}^{\frac{k}{d-k}}$ and the following bounds hold. Here $n$ refers to the number of variables in $P_{\vecw}$, i.e., $n=\sum_{i \in [d]} 2^{\abs{w_i}}$.
	\begin{enumerate}
	    \item $\SP_{k,\ell}(P_{\vecw}) \ge 2^{-O(d)}\!\cdot\! M(n,k)\!\cdot\! M(n,\ell)$,
	    \item $\APP_{k,n_0}(P_{\vecw}) \ge 2^{-O(d)}\!\cdot\! M(n,k)$.
	\end{enumerate}
	\label{lem:hard-poly}
\end{lemma}
The above lemma is proved in Section \ref{app:lem:hard-poly}. The construction of the word $\vecw$ given in that section is similar to the one in~\cite{lst1}, except that potentially we need to work with three different weights instead of the two weights used in ~\cite{lst1} -- this is mainly because we want our lower bound to apply for slightly higher values of $d$ (like $O(\log^2 n)$). 

The following lemma shows that the $\SP$ measure of the Nisan-Wigderson design polynomial is `large' for $k$ as high as $\Theta(d)$, if $\ell$ is chosen suitably. Its proof can be found in Section \ref{subapp:proof_NW_hard_poly}.

\begin{lemma}[\bf{$NW$ as a hard polynomial}]\label{lemma:NW_hard_poly}
    For $n, d \in \N$ such that $120 \leq d \leq \frac{1}{150}\paren{\frac{\log n}{\log \log n}}^2$, let $q$ be the largest prime number between $\floor{\frac{n}{2d}}$ and $\floor{\frac{n}{d}}$\footnote{Bertrand's postulate ensures that such a prime exists.}. For parameters $k \in \brac{\frac{d}{30}, \frac{d}{2} - \frac{\sqrt{d}}{8}}$ and
    $\ell=\floor{\frac{q d^2}{n_0}}$, where $n_0 = 2(d-k)\!\cdot\! \paren{{\frac{q d}{k}}}^{\frac{k}{d-k}}$,\footnote{Note that in this lemma, we do not need $n_0$ to be an integer, it is just used to define $\ell$.} $\SP_{k,\ell}(NW_{q,d,k}) \geq 2^{-O(d)}\cdot M(qd,k)\cdot M(qd,\ell).$ 
\end{lemma}


\noindent\textbf{Non set-multilinear hard polynomial}. An advantage of directly analysing the complexity measures for homogeneous formulas instead of for set-multilinear formulas is that our hard polynomial need not be set multilinear. We now describe a non set-multilinear polynomial $P$ with a large $\APP$ measure; the construction is similar to a polynomial in \cite{GargKS20}. The proof that $\APP$ of $P$ is large is considerably simpler than the proofs of the above two lemmas. 

For any $n, d, \Delta \in \N$ such that $120 \leq d \leq \frac{1}{150}\paren{\frac{\log n}{\log \log n}}^2$, define $k =\floor{\frac{\alpha\cdot d}{1+\alpha}}$, where $\alpha := \sum\limits_{\nu = 0}^{\Delta-1} \frac{(-1)^\nu}{\tau^{2^\nu - 1}}$ and $\tau:=\floor{d^{2^{1-\Delta}}}$. Then, let $n_0 = \floor{2(d-k)\cdot\paren{\frac{n}{k}}^{\frac{k}{d-k}}}$ ($n_0 \leq n$, see Section \ref{subsec: calcs for nsm hard poly}), $n_1 = n - n_0$, $\vecy = \set{x_1,\ldots, x_{n_1}}$ and $\vecz = \set{x_{n_1 + 1}, \ldots, x_n}$. Let $\cM_y$ be the set of all (monic) monomials of degree $k$ in $\vecy$ variables and $\cM_z$ be the set of all (monic) monomials in $\vecz$ variables of degree $d-k$; it can be verified that $|\cM_y| \leq |\cM_z|$. Fix any one-to-one function $\sigma: \cM_y \rightarrow \cM_z$. Then, it is easy to see that for $P_\sigma := \sum_{m \in \cM_y} m\cdot \sigma(m)$, $\APP_{k, n_0}(P) = M(n_1, k)$. It can be verified that for our choice of parameters, this is at least $2^{-O(k)}M(n,k)$. More details can be found in Section \ref{subsec: calcs for nsm hard poly}.

\begin{lemma}[\textbf{Non-set-multilinear hard polynomial}] \label{lemma: nsm hard poly}
    $\APP(P_\sigma) \geq 2^{-O(k)}M(n,k)$.
\end{lemma}

While $P_{\sigma}$ defined above might have a non-trivial set-multilinear component, it can be modified to ensure that there are no multilinear monomials in it. Notice that to prove a lower bound for such a polynomial, we must analyse the measure of a homogeneous formula computing it directly; we can not hope to get a lower bound by going via set-multilinearity as is done in \cite{lst1}.



\subsection{Putting everything together: the low-depth lower bound} \label{subsec: homo putting together}

\begin{theorem}[{\bf Low-depth homogeneous formula lower bound for $IMM$}] 
\label{thm:final-lb}
    For any positive integers $d,n,\Delta$ such that $n = \omega(d)$, 
    any homogeneous formula of product-depth at most $\Delta$ computing $IMM_{n,d}$ over any field $\F$ has size at least $2^{-O(d)}\!\cdot\! n^{\Omega\paren{d^{2^{1-\Delta}}}}$.
    
    In particular, when $d=O(\log n)$, we get a lower bound of $n^{\Omega\paren{d^{2^{1-\Delta}}}}$.
\end{theorem}

    
\begin{proof}
    \noindent  We can assume that $d^{2^{1-\Delta}} = \omega(1)$ and $h:=\floor{\log n} > 100$, as otherwise the lower bound is trivial. Suppose $IMM_{n,d}$ has a homogeneous formula $C$ of product-depth at most $\Delta$. Consider the polynomial $P_{\vecw}$, given by Lemma~\ref{lem:hard-poly}, by setting $k:=\floor{\frac{\alpha\cdot d}{1+\alpha}}$, where $\alpha := \sum\limits_{\nu = 0}^{\Delta-1} \frac{(-1)^\nu}{\tau^{2^\nu - 1}}$ and $\tau:= \floor{d^{2^{1-\Delta}}}$; these parameters are the same as those in Lemma \ref{lem:low-k-gamma}. It is easy to show that $k \in \brac{\frac{d}{4},\frac{d}{2}}$. As $\vecw$ is $h$\unbiased, by Lemma~\ref{lem:imm-proj} there exists a homogeneous formula $C'$ of product-depth at most $\Delta$ computing $P_{\vecw}$ such that $\size(C') \le \size(C)$. Hence, by Lemma~\ref{lem:low-k-gamma}, there exist homogeneous polynomials $\set{Q_{i,j}}_{i,j}$ such that 
    $P_{\vecw} = \sum_{i\in [s]} Q_{i,1}\cdots Q_{i,t_i} $, $s \le \size(C')$
    and $\entropy_k\paren{\deg(Q_{i,1}),\dots,\deg(Q_{i,t_i})} \ge \Omega \paren{d^{2^{1-\Delta}}}$ for $i\in [s]$. Denoting the number of variables in $P_{\vecw}$ by $\widetilde{n}$, Lemma~\ref{lem:hard-poly} guarantees that $n_0 \le 2(d-k)\!\cdot\!\paren{\frac{\widetilde{n}}{k}}^{\frac{k}{d-k}}$, $\ell = \floor{\frac{\widetilde{n}\cdot d}{n_0}}$ and $\SP_{k,\ell}(P_{\vecw}) \ge 2^{-O(d)}\!\cdot\! M(\widetilde{n},k) \!\cdot\! M(\widetilde{n},\ell)$. Therefore, we can apply Lemma~\ref{lem:sp-final-lb} to the same polynomial $P_{\vecw}$ which gives that $s \ge 2^{-O(d)}\!\cdot\! \paren{\frac{\widetilde{n}}{d}}^{\Omega\paren{d^{2^{1-\Delta}}}}$. Hence, $\size(C) \ge \size(C') \ge s \ge 2^{-O(d)}\!\cdot\!n^{\Omega\paren{d^{2^{1-\Delta}}}}$, since $\widetilde{n} \ge 2^h \ge n/2 = \omega(d)$.
\end{proof}

\begin{theorem}[{\bf Low-depth homogeneous formula lower bound for $NW$}] 
\label{thm:final-lb-nw}
    Let $n, d, \Delta$ be positive integers. If $\Delta = 1$, let $d = n^{1-\epsilon}$ for any constant $\epsilon > 0$ and $k = \floor{\frac{d-1}{2}}$. Otherwise, let $d \leq \frac{1}{150}\paren{\frac{\log n}{\log \log n}}^2$, let $\tau = \floor{d^{2^{1-\Delta}}}$, $\alpha = \sum\limits_{\nu = 0}^{\Delta - 1} \frac{(-1)^\nu}{\tau^{2^\nu-1}}$, and $k = \floor{\frac{\alpha\cdot d}{1 + \alpha}}$. In both cases, let $q$ be the largest prime number between $\floor{\frac{n}{2d}}$ and $\floor{\frac{n}{d}}$. Then, any homogeneous formula of product-depth at most $\Delta$ computing $NW_{q,d,k}$ over any field $\F$ has size at least $2^{-O(d)}\!\cdot\! n^{\Omega\paren{d^{2^{1-\Delta}}}}$.
    
    In particular, when $d=O(\log n)$, we get a lower bound of $n^{\Omega\paren{d^{2^{1-\Delta}}}}$.
\end{theorem}

\begin{proof}
    We analyse the cases $\Delta = 1$ and $\Delta \geq 2$ separately. \\

    \noindent\underline{$\Delta = 1.$} Let $C$ be a homogeneous formula of product-depth $1$ computing $NW_{q,d,k}$. Then, $C = \sum_{i \in [s]}\prod_{j \in [d]}Q_{i,j}$, where $Q_{i,j}$ are linear forms. Observe that for any $i\in [k]$, $\partialf^k\paren{\prod_{j \in [d]}Q_{i,j}} \subseteq \spacespanned{\prod_{j \in [d] \setminus S}C_{i,j}: |S| = k}$. Thus, $\dim\spacespanned{\partialf^k\paren{\prod_{j \in [d]}Q_{i,j}}} \leq \binom{d}{k}$. As $\partialf^k C \subseteq \sum_{i \in [s]}\spacespanned{\partialf^k\paren{\prod_{j \in [d]}Q_{i,j}}}$, $\dim\spacespanned{\partialf^k C} \leq s\cdot \binom{d}{k}$. 
    
    On the other hand, $\dim\spacespanned{\partialf^k(NW_{q,d,k})} = \binom{d}{k}\cdot q^k$: For every $S \subseteq [d], |S| = k$, $$T_S:= \set{\prod_{i \in [d] \setminus S}x_{i, h(i)}: h \in \F[z], \deg(h) < k} \subseteq \partialf^k(NW_{q,d,k}).$$ Now, for $h_1 \neq h_2 \in \F[z], \deg(h_1), \deg(h_2) < k$, there exists an $i \in [d] \setminus S$ such that $h_1(i) \neq h_2(i)$ because $|[d]\setminus S| = d-k \geq k+1$. Thus, $\prod_{i \in [d] \setminus S}x_{i, h_1(i)} \neq \prod_{i \in [d] \setminus S}x_{i, h_2(i)}$, and $|T_S| = q^k$. Also, for $S\neq S' \subseteq [d]$, $|S| = |S'| = k$, $T_S$ and $T_{S'}$ are disjoint. Hence, $\dim\spacespanned{\partialf^k(NW_{q,d,k})} \geq \binom{d}{k}\cdot q^k$.\footnote{In fact, it can be shown that this is an equality.}  Thus, $s \geq q^k = \paren{\frac{n}{d}}^{O(d)}$ as $k = \Theta(d)$ and $qd = \Theta(n)$. Because $d \leq n^{1-\epsilon}$, this means that $s \geq n^{O(d)}$. \\
    
    \noindent\underline{$\Delta \geq 2.$}  We can assume that $d^{2^{1-\Delta}} = \omega(1)$, as otherwise the given bound is trivial.  Let $C$ be a homogeneous formula of product-depth at most $\Delta$ computing $NW_{q,d,k}$; $C$ is a formula in $q d$ variables. By Lemma~\ref{lem:low-k-gamma}, there exist homogeneous polynomials $\set{Q_{i,j}}_{i,j}$ such that 
    $NW_{q,d,k} = \sum_{i\in [s]} Q_{i,1}\cdots Q_{i,t_i} $, $s \le \size(C)$,
    and $\entropy_k\paren{\deg(Q_{i,1}),\dots,\deg(Q_{i,t_i})} \ge \Omega \paren{d^{2^{1-\Delta}}}$ for $i\in [s]$. From the proof of Lemma \ref{lem:low-k-gamma}, $k \in \brac{\frac{d}{4}, \frac{d}{2}}$. In fact, as $\frac{k}{d-k} \leq \alpha \leq 1 - \frac{1}{2\tau} \leq 1 - \frac{1}{2\sqrt{d}}$, $k \leq \frac{d}{2} - \frac{\sqrt{d}}{8}$. Thus, Lemma \ref{lemma:NW_hard_poly} guarantees that for $n_0 = 2(d-k)\!\cdot\!\paren{\frac{q d}{k}}^{\frac{k}{d-k}}$ and $\ell = \floor{\frac{q d^2}{n_0}}$, $\SP_{k,\ell}\paren{NW_{q,k,d}} \ge 2^{-O(d)}\!\cdot\! M(qd,k) \!\cdot\! M(qd,\ell)$. Also, it follows from the proof of Lemma \ref{lemma:NW_hard_poly} (see Claim \ref{clm: bounds on ell}) that for $n_0 \leq qd$. So, applying Lemma~\ref{lem:sp-final-lb} to $NW_{k,d,q}$ we get that $s \ge 2^{-O(d)}\!\cdot\! \paren{\frac{qd}{d}}^{\Omega\paren{d^{2^{1-\Delta}}}} = 2^{-O(d)}\!\cdot\! n^{\Omega\paren{d^{2^{1-\Delta}}}}$ as $qd \geq \frac{n}{4}$ and $d = o(n)$. Hence, $\size(C) \geq 2^{-O(d)}\!\cdot\!n^{\Omega\paren{d^{2^{1-\Delta}}}}$.
\end{proof}

\begin{remark}
    Notice that in the above theorem, as $k$ depends on the product-depth $\Delta$, the polynomial $NW_{q,d,k}$ may be different for different values of $\Delta$. However, much like in \cite{KayalSS14}, there is a way to `stitch' all the different $NW$ polynomials for different values of $\Delta$ into a single polynomial $P$ such that any homogeneous formula of product-depth $\Delta$ computing $P$ has size at least $2^{-O(d)}n^{\Omega\paren{d^{2^{1-\Delta}}}}$. See Theorem \ref{thm: UPT lb for NW} for more details.
\end{remark}

In \cite{lst1}, the authors showed how to convert a general circuit of product-depth $\Delta$ computing a homogeneous polynomial to a homogeneous formula of product-depth $2\Delta$ without much increase in the size. Combining Lemma 11 from \cite{lst1} with Theorems ~\ref{thm:final-lb} and \ref{thm:final-lb-nw}, we get the following corollaries.

\begin{corollary}[{\bf Low-depth circuit lower bound for $IMM$}]
    For any positive integers $d,n,\Delta$ such that $n = \omega(d)$, any circuit of product-depth at most $\Delta$ computing $IMM_{n,d}$ over any field $\F$ with characteristic $0$ or more than $d$ has size at least $2^{-O(d)}\!\cdot\!n^{\Omega\paren{\frac{d^{2^{1-2\Delta}}}{\Delta}}}$.
    
    In particular, when $d=O(\log n)$, we get a lower bound of $n^{\Omega\paren{\frac{d^{2^{1-2\Delta}}}{\Delta}}}$.
    \label{cor:gen-lb}
\end{corollary}

\begin{corollary}[{\bf Low-depth circuit lower bound for $NW$}] 
     Let $n, d, \Delta$ be positive integers. If $\Delta = 1$, let $d = n^{1-\epsilon}$ for any constant $\epsilon > 0$ and $k = \floor{\frac{d-1}{2}}$. Otherwise, let $d \leq \frac{1}{150}\paren{\frac{\log n}{\log \log n}}^2$, let $\tau = \floor{d^{2^{1-\Delta}}}$, $\alpha = \sum\limits_{\nu = 0}^{\Delta - 1} \frac{(-1)^\nu}{\tau^{2^\nu-1}}$, and $k = \floor{\frac{\alpha\cdot d}{1 + \alpha}}$. In both cases, let $q$ be the largest prime number between $\floor{\frac{n}{2d}}$ and $\floor{\frac{n}{d}}$. Then, any circuit of product-depth at most $\Delta$ computing $NW_{q,d,k}$ over any field $\F$ of characteristic $0$ or more than $d$ has size at least $2^{-O(d)}\!\cdot\! n^{\Omega\paren{\frac{d^{2^{1-2\Delta}}}{\Delta}}}$.
    
    In particular, when $d=O(\log n)$, we get a lower bound of $n^{\Omega\paren{\frac{d^{2^{1-2\Delta}}}{\Delta}}}$.
\end{corollary}

We note that our lower bounds quantitatively improve on the original homogeneous formula lower bound of \cite{lst1} in terms of the dependence on the degree. While \cite{lst1} gives a lower bound of $d^{O(-d)}\!\cdot\! n^{\Omega\paren{d^{1/{2^{\Delta}-1}}}}$ (as the conversion from homogeneous to {set-multilinear formulas} increases the size by a factor of $d^{O(d)}$), our lower bound is $2^{-O(d)}\!\cdot\!n^{\Omega\paren{d^{2^{1-\Delta}}}}$. Thus, we get slight improvement both in the multiplicative factor (from $d^{O(d)}$ to $2^{O(d)}$) and in the exponent of $n$ (from $d^{\frac{1}{2^{\Delta}-1}}$ to $d^{\frac{1}{{2^{\paren{\Delta-1}}}}}$). We point out what these improvements mean for smaller depths: 

For $\Delta=2$, we observe that our lower bound for homogeneous formulas computing $IMM$ is superpolynomial as long as $d\le \epsilon\cdot \log^2 n$ for a small enough positive constant $\epsilon$, whereas the lower bound in \cite{lst1} does not work beyond $d=O\paren{\paren{\frac{\log n}{\log \log n}}^2}$. In particular, we obtain a lower bound of $n^{\Omega\paren{\log n}}$ for the size of homogeneous depth-5 formulas computing $IMM_{n,d}$ when $d = \Theta(\log^2 n)$.


Finally, we note that for $\Delta=3$ and $d \le \epsilon \cdot \log^{4/3} n$, we get a lower bound of $n^{\Omega\paren{d^{1/4}}}$, as compared to the bound $n^{\Omega{\paren{d^{1/7}}}}$ from \cite{lst1}.


	\section{Lower bound for unique-parse-tree formulas} 
\label{sec:upt-lb}
In this section we show that UPT formulas computing $IMM$ must have a `large' size. We begin by giving a decomposition for such formulas.
\subsection{Decomposition of UPT formulas}
In order to upper bound the $\SP$ (or $\APP$) measure of a UPT formula, we need certain results about binary trees and UPT formulas.
For a given canonical parse tree, we define its \emph{degree sequence} $(d_1,\dots,d_t)$ using the function \degseq~ described in the following algorithm. \\

\begin{breakablealgorithm}
	\caption{Degree sequence of a right-heavy binary tree}
	\label{alg:1}
	\begin{algorithmic}[1]
		\Function{\degseq}{$\cT$}
		\State $v_0 \gets \textnormal{root node of $\cT$}$.
		\If{$v_0$ is a leaf} 
                \State \Return $(1)$. \Comment{returning a singleton tuple}
		\EndIf
		\State $d \leftarrow \leaves(v_0)$, $i\gets 0$.
		\While {$v_i$ is not a leaf}
		\State $v_{i+1} \gets \textnormal{right child of } v_i$, $i \leftarrow i + 1$.
		\EndWhile
		\State $v \gets v_j \text{ corresponding to the largest index $j$ such that } \leaves(v_j) > \frac{d}{3}$.	\label{line:v-node}
		\State $d_1 \gets d-\leaves(v)$.
		\label{line:d-d1}
   		\State \Return $(d_1, \Call{\degseq}{\cT_v})$.\label{alg:1:recurse}
   		\Comment{To avoid a tuple of tuples, we may assume that $(d_1,{\degseq}{(\cT_v)})$ is flattened before returning.}
		\EndFunction
	\end{algorithmic}
\end{breakablealgorithm}~\\~\\
We prove the following lemma in Section \ref{subsec:proof-deg-seq}. The idea here is to `break' the tree at various nodes so that the successive sizes of the smaller trees are far from each other.

\begin{lemma}
	\label{lem:deg-seq}
	For a given canonical parse tree $\cT$ with $d \ge 1$ leaves, let $(d_1,\dots,d_t) := \degseq(\cT)$, where  the function $\degseq$~is given in Algorithm~\ref{alg:1}. Also let $e_i := d - \sum \limits_{j=1}^i d_j$ for $i\in [t]$ and $e_0 := d$.  Then, for all $i\in[t-1]$, $e_{i} \in \bigg(\frac{e_{i-1}}{3}, \frac{2\cdot e_{i-1}}{3}\bigg ]$.  Additionally, $d_{t} = 1$, $e_t=0$, and $\log_{3} d + 1 \leq t \leq \log_{3/2} d + 1$.
\end{lemma}

As mentioned in Section \ref{subsec:homo_decomopositon_lemma}, it was shown in \cite{HrubesY11b} that a homogeneous formula can be expressed as a ``small'' sum of products of homogeneous polynomials such that in each summand, the degrees of the factors roughly form a geometric sequence. We observe that this result can be strengthened for UPT formulas; in particular, we show that for UPT formulas, the ``degree sequences'' of all the summands are identical. The idea of the proof of the following lemma is similar to the proof in \cite{HrubesY11b}; we repeatedly break the formula at gates corresponding to the nodes $v$ obtained from line \ref{line:v-node} of $\degseq(\cT(C))$. The main difference between our proof and the one in \cite{HrubesY11b} is in showing that all the summands obtained in this way have the same degree sequence. For a detailed proof of the following lemma, see Section \ref{subsec:proof-log-pdt}.

\begin{lemma}[{\bf Log-product decomposition of UPT formulas}]
	\label{lem:log-pdt}
	Let $f \in \F[\vecx]$ be a homogeneous polynomial of degree $d \ge 1$ computed by a UPT formula $C$ with canonical parse tree $\cT(C)$. Let $(d_1,\dots,d_t):=\degseq(\cT(C))$. Then there exist an integer $s \le \size(C)$ and homogeneous polynomials $\set{Q_{i,j}}_{i,j}$ where $\deg(Q_{i,j}) = d_j$ for $i\in[s]$, $j\in [t]$, such that $$f =  \sum \limits_{i=1}^{s} Q_{i,1} \cdots Q_{i, t}. $$ 
\end{lemma}

\subsection{UPT formulas have high residue}
Now we show that there exists a value of $k$ that has high $\entropy$ with respect to the degrees of the factors given by the above log-product lemma.

\begin{lemma}\label{lem:uptResLB}[{\bf High $\entropy$ for a degree sequence}]
	For any given canonical parse tree $\cT$ with $d \ge 1$ leaves, let $(d_1,\dots,d_t):=\degseq(\cT)$ and $k:=\funK(d_1,\dots,d_t)$ where the function $\funK$ is described in Algorithm \ref{alg:2}. Then $$\entropy_k(d_1,\dots,d_t)  \ge \frac{\log_3 d - 10}{216}.$$
	\label{lem:gamma-upt-lb}
\end{lemma}
\begin{proof}
The definition of $\entropy$ implies that it is sufficient to show that $\sum \limits_{p=1}^t \abs{k_p - \frac{k}{d}\cdot d_p} \ge \frac{\log_3 d - 10}{108}$ for arbitrary integers $k_1,\dots,k_t$. We now argue that this is indeed the case. Let us assume that $d>3^{10}$ as otherwise, the lower bound on $\entropy$ is trivial.
We import the definitions of $\set{e_i}_{i\in [0..t]}$ (line \ref{line:ei-s}), $m$ (line \ref{line:m-val}), and the function $\funJ$ (line \ref{line:funJ}) from Algorithm~\ref{alg:2}. By Lemma~\ref{lem:deg-seq}, we have $d_{t} = 1$, $e_t=0$, $t \ge \log_3 d + 1 \ge 11$, and for all $i\in[t-1]$, 
\begin{equation}
    e_{i} \in \bigg(\frac{e_{i-1}}{3}, \frac{2\cdot e_{i-1}}{3}\bigg ].
    \label{eqn:e}
\end{equation}
We have the following property of $\funJ$.
\begin{claim}
\label{claim:J}
    For all $i\in [3m]$, we have $3^{i-1} < e_{\funJ(i)} \le 3^i$, hence $\funJ:[3m] \to [t-2]$ is an injective mapping.
\end{claim}	
A proof of the above claim can be found in Section \ref{subsec:proof-J}. We now fix an $i\in [m]$ and $j:=\funJ(3i)\in [t-2]$ at line~\ref{line:j-value} in Algorithm~\ref{alg:2} and continue our analysis. Note that, by the definition of $e_j$ and Equation \eqref{eqn:e},
	
	\begin{algorithm}
		\caption{The value of $k$ for a given sequence of degrees}
		\label{alg:2}
		\begin{algorithmic}[1]
			\Function{\funK}{$d_1,\dots,d_t$}
			\Comment{Returns $k$ which shall be the order of derivatives for the $\SP$ and $\APP$ measures.}
			\State $d = d_1 + \dots + d_t$.
			\For{$i \in [0..t]$}
			\State $e_i  \gets d- \sum \limits_{j=1}^i d_j$. \label{line:ei-s}
			\EndFor
			\State $m \gets \floor{\frac{\log_3 d-1}{3}}$.   \label{line:m-val}
             \Comment{Defining a function $\funJ:[3m] \to [t-2]$.}
			\For{$ i\in [3m]$}
			\State $\funJ(i) \gets \min \set{j \in [0..t]:e_j \le 3^i}$. \label{line:funJ}
			\EndFor
			\State $(a_1,\dots,a_m) \gets \emph{undefined}$.
	    \For{$i \in [m]$}
			\State $j \gets \funJ(3i)$. \label{line:j-value}
			\State $b_0 \gets \paren{\sum \limits_{p=1}^{i-1} \frac{a_p}{3^{3p}} }\cdot d_{j+1}$. \label{line:b0}
			\Comment{$b_1$ defined below is not used in the algorithm but will be useful in the analysis.}
            \State $b_1 \gets \paren{\sum \limits_{p=1}^{i-1} \frac{a_p}{3^{3p}} + \frac{1}{3^{3i}}}\cdot d_{j+1}$. \label{line:b1}
			\If{$\fract{b_0} \in \brac{\frac{1}{18}, \frac{17}{18}}$} \label{line:if}
			\State $a_i \gets 0$.
			\Else \State $a_i \gets 1$.
			\EndIf
			\EndFor
			\Statex
			\State $\alpha \gets \sum \limits_{p=1}^{m} \frac{a_p}{3^{3p}}$ \label{line:alpha} 
			\State $k \gets \floor{\alpha\cdot d}$ \label{line:k-val}
			\State \Return $k$.
			\EndFunction
		\end{algorithmic}
	\end{algorithm}

	\begin{equation}
		d_{j+1}=e_j - e_{j+1} \in \bigg [\frac{e_j}{3},\frac{2e_j}{3}\bigg ) \subseteq \paren{3^{3i-2}, 2\!\cdot\! 3^{3i-1}},
		\label{eqn:r}	
	\end{equation} 
	where the last containment follows from the fact that $e_j \in (3^{3i-1}, 3^{3i}]$ (by applying Claim~\ref{claim:J} for the index $3i$). From line~\ref{line:alpha}, we have 
	$$\alpha = \sum \limits_{p=1}^{m} \frac{a_p}{3^{3p}}.$$
	
	\noindent Then, 
	\begin{equation*}
		\alpha\cdot d_{j+1} = \paren{\sum \limits_{p=1}^m 	\frac{a_p}{3^{3p}} }\cdot d_{j+1} =  \paren{\sum \limits_{p=1}^{i-1} \frac{a_p}{3^{3p}} + \frac{a_{i}}{3^{3i}} + \sum \limits_{p={i+1}}^m \frac{a_p}{3^{3p}} }\cdot d_{j+1} = s_1 + s_2 + s_3,
		\label{eqn:sum}
	\end{equation*}
	where 
	\begin{equation*}
		s_1:=\paren{\sum \limits_{p=1}^{i-1} \frac{a_p}{3^{3p}} }\cdot d_{j+1},
		\label{eqn:s1}
	\end{equation*}
	\begin{equation*}
		s_2:=\paren{\frac{a_{i}}{3^{3i}}}\cdot d_{j+1},
		\label{eqn:s2}
	\end{equation*} and
	\begin{align*}
		s_3 & :=\paren{\sum \limits_{p={i+1}}^m \frac{a_p}{3^{3p}}}\cdot d_{j+1} 
        \le \paren{\sum \limits_{p=i+1}^{\infty} \frac{1}{3^{3p}} }\cdot d_{j+1}  
		\le \paren{\frac{27}{26\cdot 3^{3(i+1)}}}\cdot d_{j+1} 
		 \le \frac{27\cdot 2\cdot 3^{3i-1}}{26\cdot 3^{3(i+1)}}  \tag{using Equation \eqref{eqn:r}}
		& = \frac{1}{39}.
		\label{eqn:s3}
	\end{align*}
	Note that $\abs{b_1 - b_0} = b_1 - b_0 = \frac{d_{j+1}}{3^{3i}} \subseteq \brac{\frac{1}{9}, \frac{8}{9}}$. We now consider two cases $\fract{b_0} \in [\frac{1}{18},\frac{17}{18}]$ and $\fract{b_0} \notin [\frac{1}{18},\frac{17}{18}]$ based on the if-else condition at line~\ref{line:if}.  The following simple claim whose proof can be found in Section \ref{subsec:proof-1} will be helpful in analysing these cases.
		\begin{claim}
			\label{fact:1}
			For real numbers $b_0$ and $b_1$, if $|b_1 - b_0| \in \brac{\frac{1}{9}, \frac{8}{9}}$, then either $\fract{b_0} \in \brac{\frac{1}{18},\frac{17}{18}}$ or $\fract{b_1} \in \brac{\frac{1}{18},\frac{17}{18}}$.
		\end{claim}
\noindent \textbf{Case 1:} $\fract{b_0} \in [\frac{1}{18},\frac{17}{18}]$. In this case, $a_i=0$, so $s_2=0$. Because $s_1 = b_0$, we have
		\begin{align*}
			\fract{\alpha\cdot d_{j+1}} & = \fract{s_1+s_2+s_3} = \fract{b_0 + 0 + s_3} = \fract{\fract{b_0} + s_3}\\
			& = \fract{b_0} + s_3 \in \brac{ \frac{1}{18}-\frac{1}{39}, \frac{17}{18}+\frac{1}{39}}=\brac{\frac{7}{234},\frac{227}{234}}. \tag{as $s_3 \le 1/39$}
		\end{align*}
\textbf{Case 2:} $\fract{b_0} \notin [\frac{1}{18},\frac{17}{18}]$. By Claim \ref{fact:1}, $\fract{b_1} \in [\frac{1}{18},\frac{17}{18}]$ . As $a_i=1$ in this case, 
		\begin{align*}
			\fract{\alpha\cdot d_{j+1}} & = \fract{(s_1+s_2)+s_3} = \fract{b_1 + s_3} = \fract{\fract{b_1} + s_3}\\
			& = \fract{b_1} + s_3 \in \brac{ \frac{1}{18}-\frac{1}{39}, \frac{17}{18}+\frac{1}{39}}=\brac{\frac{7}{234},\frac{227}{234}}. \tag{as $s_3 \le 1/39$}
		\end{align*}	
Thus, in both the cases, we have, for all $i\in [m]$ and $j=\funJ(3i)$ that,
	\begin{equation*}
		\fract{\alpha\cdot d_{j+1}} \in \brac{\frac{7}{234},\frac{227}{234}}
		\label{eqn:fract}.
	\end{equation*}
	As $k_{j+1}$ is an integer,
	\begin{align*}
		\abs{k_{j+1} - \alpha\cdot d_{j+1}} & \ge 
		\abs{\floorceil{\alpha \cdot d_{j+1}}-\alpha \cdot d_{j+1}}\\
		& = \min \set{\fract{\alpha\cdot d_{j+1}}, 1- \fract{\alpha\cdot d_{j+1}}}  \\
		&  \ge \frac{7}{234}.
		\label{eqn:diff}
	\end{align*}
	
	\noindent Now let $\cX := \set{\funJ(3i)+1:i\in [m-2]}\subseteq[t-1]$. Then the above condition translates to: For all $p \in \cX$,
	\begin{equation}
		|k_p - \alpha\cdot d_p| \ge \frac{7}{234}
		\label{eqn:large-fract}
	\end{equation}
	and Equation \eqref{eqn:r} implies that
	\begin{equation}
		d_p \le 2\cdot 3^{3(m-2)-1} \le \frac{d}{3^6}. \label{eqn:small-d}
	\end{equation}

	\noindent We thus have 
	\begin{align*}
		 \sum \limits_{p=1}^t \abs{k_p - \frac{k\cdot d_p}{d}} 
		& \ge \sum \limits_{p\in \cX} \abs{k_p - \frac{\floor{\alpha\cdot d}}{d}\cdot d_p} \tag{using line~\ref{line:k-val}}\\
		& \ge \sum \limits_{p \in \cX}  |k_p - \alpha\cdot d_p| - \abs{\set{\alpha\cdot d}\cdot \frac{d_p}{d}}  \tag{using $|x - y| \ge |x| - |y|$}\\
		& \ge \sum \limits_{p \in \cX}  |k_p - \alpha\cdot d_p| - \frac{d_p}{d} \\
		& \ge  \sum \limits_{p \in \cX}  \frac{7}{234} - \frac{1}{3^6} \tag{using Equations \eqref{eqn:large-fract} and ~\eqref{eqn:small-d}}\\
		& \ge \frac{1}{36}\cdot |\cX|=\frac{m-2}{36} \tag{as $\funJ$ is injective} \\
		& =\frac{1}{36}\cdot \paren{\floor{\frac{\log_3 d-1}{3}}-2}\\
		& \ge \frac{\log_3 d-10}{108}.
	\end{align*}
	
\end{proof}

Hence, we have a lower bound for UPT formulas computing the $IMM$ polynomial as well as a polynomial related to the $NW$ polynomial -- the next section contains the details.

\subsection{Putting everything together: the UPT formula lower bound}
In this section, we prove the following theorem.
\begin{theorem}[{\bf UPT formula lower bound for $IMM$}] \label{thm:upt_lb_IMM}
	For $n \in \N$ and $d \le \epsilon \cdot \log n\cdot \log \log n$, where $\epsilon > 0$ is a small enough constant, any UPT formula computing $IMM_{n,d}$ over any field $\F$ has size $n^{\Omega(\log d)}$.
\end{theorem}

\begin{proof}
	Let $C$ be a UPT formula computing $IMM_{n,d}$, $(d_1,\dots,d_t) := \degseq(\cT(C))$, and $k:=\funK(d_1,\dots,d_t)$. We now delve into the program $\funK(d_1,\dots,d_k)$ in Algorithm~\ref{alg:2} to understand the range of the $k$ that it outputs. Notice that for $i=1$, $b_0$ gets the value 0 (at line~\ref{line:b0}), so the condition in line~\ref{line:if} fails resulting in the value of $a_1$ being 1. Hence, $k = \floor{\alpha\cdot d} \ge \paren{\frac{a_1}{3^3}}\cdot d-1 = \frac{d}{27}-1 \ge \frac{d}{30}$. On the other hand, $k \le \alpha\cdot d \le \paren{\sum \limits_{p=1}^\infty \frac{1}{3^{3p}}}\cdot d \le \frac{d}{26} \le \frac{d}{2}$. Hence Lemma~\ref{lem:hard-poly} is applicable with $k=\funK(\cT(C))$ and $h = \floor{\log n}$ -- giving a polynomial $P_{\vecw}$ in $\widetilde{n}$ (say) variables such that its $\SP$ measure is large. By Lemma~\ref{lem:imm-proj}, this means that there exists a UPT formula $C'$ of similar size computing $P_{\vecw}$ (in fact, $\cT(C') = \cT(C)$). 
	
    Lemma~\ref{lem:log-pdt} and Lemma~\ref{lem:gamma-upt-lb} when put together give
	$$ P_{\vecw} = \sum_{i=1}^s Q_{i,1}\dots Q_{i,t}$$ for some  $s \le \size(C')$ such that for all $i \in s$, $\entropy_k(\degseq(\cT(C'))) \ge \Omega\paren{\log d}$. Here we are using the fact that $k=\funK(\degseq(\cT(C')))$ as $\cT(C) = \cT(C')$. Now since $\vecw$ is obtained by Lemma~\ref{lem:hard-poly}, applying Lemma~\ref{lem:sp-final-lb} with $\gamma=\Omega(\log d)$ gives that $s \ge 2^{-O(d)}\cdot n^{\Omega\paren{\log d}}$. Hence, $\size(C) \ge \size(C') \ge s \ge 2^{-O(d)}\cdot n^{\Omega\paren{\log d}}$.
	
	If $d \leq \epsilon\cdot \log n\cdot \log \log n$ for some  $\epsilon > 0$, then $\frac{d}{\log d} \leq \frac{\epsilon \cdot \log n\cdot \log \log n}{\log \log n + \log \log \log n} \leq \epsilon'\cdot \log n$ for some $0 < \epsilon' \leq \epsilon$. Hence, $d = \epsilon' \cdot \log n\cdot \log d$ and $2^{-O(d)}\cdot n^{\Omega(\log d)}=n^{\Omega(\log d)}$ if $\epsilon$ (and thus $\epsilon'$) is a small enough constant.
\end{proof}

\begin{remark}
    The above theorem can also be derived by using the complexity measure studied in \cite{lst1} along with the observation that the {\em unbounded-depth} set-multilinearization due to \cite{Raz13} (which increases the size by a factor of $2^{O(d)}$) preserves parse trees. 
    
\end{remark}

We also get an analogous theorem for a polynomial related to the $NW$ polynomial.

\begin{theorem} \label{thm: UPT lb for NW}
	Let $n \in \N$, $d \le \epsilon \cdot \log n\cdot \log \log n$, where $\epsilon > 0$ is a small enough constant, and $q$ be the largest prime number between $\floor{\frac{n}{2d}}$ and $\floor{\frac{n}{d}}$. Then, any UPT formula computing $P = \sum\limits_{i = \floor{d/30}}^{\ceil{d/2}}y_i \cdot NW_{q,d,i}$ (where the $y$ variables are distinct from the $\vecx$ variables), over any field $\F$ has size $n^{\Omega(\log d)}$.
\end{theorem}

\begin{proof}
    Let $C$ be any UPT formula computing $P$, $(d_1, \ldots, d_t) := \degseq(\cT(C))$, and $k:=\funK(d_1,\\ \dots,d_t)$. Observe that the formula $C'$ obtained by setting $y_i = 0$ for all $i \in [d]\setminus \set{k}$ in $C$ computes $NW_{q,d,k}(\vecx)$. It follows from the proof of Lemma \ref{lem:log-pdt} that not only is $C'$ a UPT formula, but also that its canonical parse tree $\cT(C')$ is the same $\cT(C)$. Hence, $(d_1, \ldots, d_t) := \degseq(\cT(C'))$ and $k = \funK(\degseq(\cT(C')))$. Then, arguing as in the proof of Theorem \ref{thm:upt_lb_IMM}, we get that $\size(C') \geq (qd)^{\Omega(d)} = n^{\Omega(d)}$. As $\size(C') \leq \size(C)$, this proves the theorem.
\end{proof}
        \section{Conclusion} \label{sec:conclusion}
In this work, we prove superpolynomial lower bounds for low-depth arithmetic circuits as well as UPT formulas using the $\SP$ and $\APP$ measures. Unlike \cite{lst1}, which proves low-depth circuit lower bounds by proving lower bounds for low-depth set-multilinear formulas, we directly upper bound the $\SP$ and $\APP$ of low-depth homogeneous formulas. Our approach has a potential advantage: Since the set-multilinearization process incurs a loss of a factor of $d^{O(d)}$, it is not clear if proving exponential lower bounds for low-depth set-multilinear formulas would yield exponential lower bounds for low-depth homogeneous formulas. A direct approach does not seem to incur such an inherent loss. So, it is conceivable that one might be able to prove exponential lower bounds for low-depth homogeneous formulas or other related models using a direct approach. \\

\noindent\textbf{Problem 1.} Prove exponential lower bounds for low-depth homogeneous arithmetic formulas. Prove exponential lower bounds for low-depth, \textit{multi-r-ic} formulas.\footnote{A formula is said to be multi-r-ic, if the formal degree of every gate with respect to every variable is at most $r$ \cite{KayalS17, KayalST16a}.}\\

The UPT formula lower bound proved in this work is for formulas computing polynomials of degree at most $O(\log n \cdot \log \log n)$. It would be interesting to increase the range of degrees for which our bound works. In the non-commutative setting, exponential lower bounds are known for formulas with exponentially many parse trees \cite{lls19}. It is natural to ask if the same can also be done in the commutative setting.  \\
    
\noindent\textbf{Problem 2.} Prove an $n^{\Omega(\log d)}$ lower bound for UPT formulas for $d = n^{O(1)}$. Prove a superpolynomial lower bound for formulas with ``many'' parse trees. \\
    
Our work also raises the prospect of learning low-depth homogeneous formulas given black-box access using the `learning from lower bounds' paradigm proposed in \cite{GargKS20, KayalS19}. \\

\noindent\textbf{Problem 3.} Obtain efficient learning algorithms for random low-depth homogeneous formulas.\\

To upper bound $\SP$ or $\APP$ of a homogeneous formula $C$, we first show in Section \ref{sec:der-structural} that the space of partial derivatives of $C$ has some structure and then exploit this structure using shifts or affine projections. There might be a better way to exploit the structure of the space of partials and get better lower bounds. Some candidates for an alternate way to exploit this structure are going modulo an appropriately chosen ideal or using random restrictions along with shifts as done in \cite{KayalLSS17, KumarS17} for homogeneous depth-$4$ formulas. Exploring this possibility is also an interesting direction for future work.
	
	
	\bibliographystyle{alpha}
	\bibliography{references}

\begin{thebibliography}{DMPY12}

\bibitem[AKV18]{AlonKV18}
Noga Alon, Mrinal Kumar, and Ben~Lee Volk.
\newblock {Unbalancing Sets and an Almost Quadratic Lower Bound for
  Syntactically Multilinear Arithmetic Circuits}.
\newblock In Rocco~A. Servedio, editor, {\em 33rd Computational Complexity
  Conference, {CCC} 2018, June 22-24, 2018, San Diego, CA, {USA}}, volume 102
  of {\em LIPIcs}, pages 11:1--11:16. Schloss Dagstuhl - Leibniz-Zentrum fuer
  Informatik, 2018.

\bibitem[AV08]{AgrawalV08}
Manindra Agrawal and V.~Vinay.
\newblock {Arithmetic Circuits: {A} Chasm at Depth Four}.
\newblock In {\em 49th Annual {IEEE} Symposium on Foundations of Computer
  Science, {FOCS} 2008, October 25-28, 2008, Philadelphia, PA, {USA}}, pages
  67--75. {IEEE} Computer Society, 2008.

\bibitem[BDS22]{BhargavDS22}
C.~S. Bhargav, Sagnik Dutta, and Nitin Saxena.
\newblock Improved lower bound, and proof barrier, for constant depth algebraic
  circuits.
\newblock In Stefan Szeider, Robert Ganian, and Alexandra Silva, editors, {\em
  47th International Symposium on Mathematical Foundations of Computer Science,
  {MFCS} 2022, August 22-26, 2022, Vienna, Austria}, volume 241 of {\em
  LIPIcs}, pages 18:1--18:16. Schloss Dagstuhl - Leibniz-Zentrum f{\"{u}}r
  Informatik, 2022.

\bibitem[BLS16]{BalajiLS16}
Nikhil Balaji, Nutan Limaye, and Srikanth Srinivasan.
\newblock An almost cubic lower bound for {\(\Sigma\)}{\(\Pi\)}{\(\Sigma\)}
  circuits computing a polynomial in {VP}.
\newblock {\em Electronic Colloquium on Computational Complexity {(ECCC)}},
  23:143, 2016.

\bibitem[BS83]{BaurS83}
Walter Baur and Volker Strassen.
\newblock {The Complexity of Partial Derivatives}.
\newblock {\em Theor. Comput. Sci.}, 22:317--330, 1983.

\bibitem[CELS18]{ChillaraEL018}
Suryajith Chillara, Christian Engels, Nutan Limaye, and Srikanth Srinivasan.
\newblock {A Near-Optimal Depth-Hierarchy Theorem for Small-Depth Multilinear
  Circuits}.
\newblock In Mikkel Thorup, editor, {\em 59th {IEEE} Annual Symposium on
  Foundations of Computer Science, {FOCS} 2018, Paris, France, October 7-9,
  2018}, pages 934--945. {IEEE} Computer Society, 2018.

\bibitem[CKSV22]{ChatterjeeKSV22}
Prerona Chatterjee, Mrinal Kumar, Adrian She, and Ben~Lee Volk.
\newblock Quadratic lower bounds for algebraic branching programs and formulas.
\newblock {\em Comput. Complex.}, 31(2):8, 2022.
\newblock Conference version appeared in the proceedings of CCC 2020.

\bibitem[CLS19]{ChillaraL019}
Suryajith Chillara, Nutan Limaye, and Srikanth Srinivasan.
\newblock Small-depth multilinear formula lower bounds for iterated matrix
  multiplication with applications.
\newblock {\em {SIAM} J. Comput.}, 48(1):70--92, 2019.
\newblock Conference version appeared in the proceedings of STACS 2018.

\bibitem[DMPY12]{DvirMPY12}
Zeev Dvir, Guillaume Malod, Sylvain Perifel, and Amir Yehudayoff.
\newblock Separating multilinear branching programs and formulas.
\newblock In Howard~J. Karloff and Toniann Pitassi, editors, {\em Proceedings
  of the 44th Symposium on Theory of Computing Conference, {STOC} 2012, New
  York, NY, USA, May 19 - 22, 2012}, pages 615--624. {ACM}, 2012.

\bibitem[FK09]{FortnowK09}
Lance Fortnow and Adam~R. Klivans.
\newblock Efficient learning algorithms yield circuit lower bounds.
\newblock {\em J. Comput. Syst. Sci.}, 75(1):27--36, 2009.

\bibitem[FKS16]{ForbesKS16}
Michael~A. Forbes, Mrinal Kumar, and Ramprasad Saptharishi.
\newblock Functional lower bounds for arithmetic circuits and connections to
  boolean circuit complexity.
\newblock In {\em 31st Conference on Computational Complexity, {CCC} 2016, May
  29 to June 1, 2016, Tokyo, Japan}, pages 33:1--33:19, 2016.

\bibitem[FLMS15]{FournierLMS15}
Herv{\'{e}} Fournier, Nutan Limaye, Guillaume Malod, and Srikanth Srinivasan.
\newblock {Lower Bounds for Depth-4 Formulas Computing Iterated Matrix
  Multiplication}.
\newblock {\em {SIAM} J. Comput.}, 44(5):1173--1201, 2015.
\newblock Conference version appeared in the proceedings of STOC 2014.

\bibitem[GKKS14]{GKKS14}
Ankit Gupta, Pritish Kamath, Neeraj Kayal, and Ramprasad Saptharishi.
\newblock {Approaching the Chasm at Depth Four}.
\newblock {\em J. {ACM}}, 61(6):33:1--33:16, 2014.
\newblock Conference version appeared in the proceedings of CCC 2013.

\bibitem[GKKS16]{GKKS16}
Ankit Gupta, Pritish Kamath, Neeraj Kayal, and Ramprasad Saptharishi.
\newblock {Arithmetic Circuits: {A} Chasm at Depth 3}.
\newblock {\em {SIAM} J. Comput.}, 45(3):1064--1079, 2016.
\newblock Conference version appeared in the proceedings of FOCS 2013.

\bibitem[GKS20]{GargKS20}
Ankit Garg, Neeraj Kayal, and Chandan Saha.
\newblock Learning sums of powers of low-degree polynomials in the
  non-degenerate case.
\newblock In Sandy Irani, editor, {\em 61st {IEEE} Annual Symposium on
  Foundations of Computer Science, {FOCS} 2020, Durham, NC, USA, November
  16-19, 2020}, pages 889--899. {IEEE}, 2020.

\bibitem[GL19]{GesmundoL19}
Fulvio Gesmundo and Joseph~M. Landsberg.
\newblock Explicit polynomial sequences with maximal spaces of partial
  derivatives and a question of k. mulmuley.
\newblock {\em Theory Comput.}, 15:1--24, 2019.

\bibitem[GST20]{GuptaST20}
Nikhil Gupta, Chandan Saha, and Bhargav Thankey.
\newblock A super-quadratic lower bound for depth four arithmetic circuits.
\newblock In Shubhangi Saraf, editor, {\em 35th Computational Complexity
  Conference, {CCC} 2020, July 28-31, 2020, Saarbr{\"{u}}cken, Germany (Virtual
  Conference)}, volume 169 of {\em LIPIcs}, pages 23:1--23:31. Schloss Dagstuhl
  - Leibniz-Zentrum f{\"{u}}r Informatik, 2020.

\bibitem[HY11]{HrubesY11b}
Pavel Hrubes and Amir Yehudayoff.
\newblock Homogeneous formulas and symmetric polynomials.
\newblock {\em Comput. Complex.}, 20(3):559--578, 2011.

\bibitem[Kal85]{Kalorkoti85}
K.~Kalorkoti.
\newblock {A Lower Bound for the Formula Size of Rational Functions}.
\newblock {\em {SIAM} J. Comput.}, 14(3):678--687, 1985.

\bibitem[Kay12]{Kayal12a}
Neeraj Kayal.
\newblock An exponential lower bound for the sum of powers of bounded degree
  polynomials.
\newblock {\em Electronic Colloquium on Computational Complexity {(ECCC)}},
  19:81, 2012.

\bibitem[KLSS17]{KayalLSS17}
Neeraj Kayal, Nutan Limaye, Chandan Saha, and Srikanth Srinivasan.
\newblock {An Exponential Lower Bound for Homogeneous Depth Four Arithmetic
  Formulas}.
\newblock {\em {SIAM} J. Comput.}, 46(1):307--335, 2017.
\newblock Conference version appeared in the proceedings of FOCS 2014.

\bibitem[KNS20]{KayalNS20}
Neeraj Kayal, Vineet Nair, and Chandan Saha.
\newblock {Separation Between Read-once Oblivious Algebraic Branching Programs
  (ROABPs) and Multilinear Depth-three Circuits}.
\newblock {\em {ACM} Trans. Comput. Theory}, 12(1):2:1--2:27, 2020.
\newblock Conference version appeared in the proceedings of STACS 2016.

\bibitem[Koi12]{Koiran12}
Pascal Koiran.
\newblock {Arithmetic circuits: The chasm at depth four gets wider}.
\newblock {\em Theor. Comput. Sci.}, 448:56--65, 2012.

\bibitem[KS14a]{KumarS14}
Mrinal Kumar and Shubhangi Saraf.
\newblock The limits of depth reduction for arithmetic formulas: it's all about
  the top fan-in.
\newblock In {\em Symposium on Theory of Computing, {STOC} 2014, New York, NY,
  USA, May 31 - June 03, 2014}, pages 136--145, 2014.

\bibitem[KS14b]{KumarS14a}
Mrinal Kumar and Shubhangi Saraf.
\newblock Superpolynomial lower bounds for general homogeneous depth 4
  arithmetic circuits.
\newblock In {\em Automata, Languages, and Programming - 41st International
  Colloquium, {ICALP} 2014, Copenhagen, Denmark, July 8-11, 2014, Proceedings,
  Part {I}}, pages 751--762, 2014.

\bibitem[KS16]{KayalS16}
Neeraj Kayal and Chandan Saha.
\newblock {Lower Bounds for Depth-Three Arithmetic Circuits with small bottom
  fanin}.
\newblock {\em Computational Complexity}, 25(2):419--454, 2016.
\newblock Conference version appeared in the proceedings of CCC 2015.

\bibitem[KS17a]{KayalS17}
Neeraj Kayal and Chandan Saha.
\newblock Multi-k-ic depth three circuit lower bound.
\newblock {\em Theory Comput. Syst.}, 61(4):1237--1251, 2017.
\newblock The conference version appeared in the proceedings of STACS, 2015.

\bibitem[KS17b]{KumarS17}
Mrinal Kumar and Shubhangi Saraf.
\newblock {On the Power of Homogeneous Depth 4 Arithmetic Circuits}.
\newblock {\em {SIAM} J. Comput.}, 46(1):336--387, 2017.
\newblock Conference version appeared in the proceedings of FOCS 2014.

\bibitem[KS19a]{KayalS19}
Neeraj Kayal and Chandan Saha.
\newblock Reconstruction of non-degenerate homogeneous depth three circuits.
\newblock In Moses Charikar and Edith Cohen, editors, {\em Proceedings of the
  51st Annual {ACM} {SIGACT} Symposium on Theory of Computing, {STOC} 2019,
  Phoenix, AZ, USA, June 23-26, 2019}, pages 413--424. {ACM}, 2019.

\bibitem[KS19b]{KumarS19}
Mrinal Kumar and Ramprasad Saptharishi.
\newblock The computational power of depth five arithmetic circuits.
\newblock {\em {SIAM} J. Comput.}, 48(1):144--180, 2019.

\bibitem[KS22]{KushS22}
Deepanshu Kush and Shubhangi Saraf.
\newblock Improved low-depth set-multilinear circuit lower bounds.
\newblock In Shachar Lovett, editor, {\em 37th Computational Complexity
  Conference, {CCC} 2022, July 20-23, 2022, Philadelphia, PA, {USA}}, volume
  234 of {\em LIPIcs}, pages 38:1--38:16. Schloss Dagstuhl - Leibniz-Zentrum
  f{\"{u}}r Informatik, 2022.

\bibitem[KSS14]{KayalSS14}
Neeraj Kayal, Chandan Saha, and Ramprasad Saptharishi.
\newblock A super-polynomial lower bound for regular arithmetic formulas.
\newblock In {\em Symposium on Theory of Computing, {STOC} 2014, New York, NY,
  USA, May 31 - June 03, 2014}, pages 146--153, 2014.

\bibitem[KST16a]{KayalST16}
Neeraj Kayal, Chandan Saha, and S{\'{e}}bastien Tavenas.
\newblock {An Almost Cubic Lower Bound for Depth Three Arithmetic Circuits}.
\newblock In {\em 43rd International Colloquium on Automata, Languages, and
  Programming, {ICALP} 2016, July 11-15, 2016, Rome, Italy}, pages 33:1--33:15,
  2016.

\bibitem[KST16b]{KayalST16a}
Neeraj Kayal, Chandan Saha, and S{\'{e}}bastien Tavenas.
\newblock On the size of homogeneous and of depth four formulas with low
  individual degree.
\newblock In {\em Proceedings of the 48th Annual {ACM} {SIGACT} Symposium on
  Theory of Computing, {STOC} 2016, Cambridge, MA, USA, June 18-21, 2016},
  pages 626--632, 2016.

\bibitem[LLS19]{lls19}
Guillaume Lagarde, Nutan Limaye, and Srikanth Srinivasan.
\newblock {Lower Bounds and {PIT} for Non-commutative Arithmetic Circuits with
  Restricted Parse Trees}.
\newblock {\em Computational Complexity}, 28(3):471--542, 2019.

\bibitem[LMP19]{LagardeMP19}
Guillaume Lagarde, Guillaume Malod, and Sylvain Perifel.
\newblock Non-commutative computations: lower bounds and polynomial identity
  testing.
\newblock {\em Chicago Journal of Theoretical Computer Science}, (2):1--19,
  2019.

\bibitem[LST21]{lst1}
Nutan Limaye, Srikanth Srinivasan, and S{\'{e}}bastien Tavenas.
\newblock {Superpolynomial Lower Bounds Against Low-Depth Algebraic Circuits}.
\newblock In {\em 62nd {IEEE} Annual Symposium on Foundations of Computer
  Science, {FOCS} 2021, Denver, CO, USA, February 7-10, 2022}, pages 804--814.
  {IEEE}, 2021.
\newblock A full version of the paper can be found at
  \url{https://eccc.weizmann.ac.il/report/2021/081}.

\bibitem[LST22]{lst3}
Nutan Limaye, Srikanth Srinivasan, and S{\'{e}}bastien Tavenas.
\newblock On the partial derivative method applied to lopsided set-multilinear
  polynomials.
\newblock In Shachar Lovett, editor, {\em 37th Computational Complexity
  Conference, {CCC} 2022, July 20-23, 2022, Philadelphia, PA, {USA}}, volume
  234 of {\em LIPIcs}, pages 32:1--32:23. Schloss Dagstuhl - Leibniz-Zentrum
  f{\"{u}}r Informatik, 2022.

\bibitem[Nis91]{Nisan91}
Noam Nisan.
\newblock {Lower Bounds for Non-Commutative Computation (Extended Abstract)}.
\newblock In Cris Koutsougeras and Jeffrey~Scott Vitter, editors, {\em
  Proceedings of the 23rd Annual {ACM} Symposium on Theory of Computing, May
  5-8, 1991, New Orleans, Louisiana, {USA}}, pages 410--418. {ACM}, 1991.

\bibitem[NW97]{NisanW97}
Noam Nisan and Avi Wigderson.
\newblock {Lower Bounds on Arithmetic Circuits Via Partial Derivatives}.
\newblock {\em Computational Complexity}, 6(3):217--234, 1997.
\newblock Conference version appeared in the proceedings of FOCS 1995.

\bibitem[Raz03]{Raz03}
Ran Raz.
\newblock {On the Complexity of Matrix Product}.
\newblock {\em {SIAM} J. Comput.}, 32(5):1356--1369, 2003.
\newblock Conference version appeared in the proceedings of STOC 2002.

\bibitem[Raz06]{Raz06}
Ran Raz.
\newblock {Separation of Multilinear Circuit and Formula Size}.
\newblock {\em Theory of Computing}, 2(6):121--135, 2006.
\newblock Conference version appeared in the proceedings of FOCS 2004.

\bibitem[Raz10]{Raz10}
Ran Raz.
\newblock {Elusive Functions and Lower Bounds for Arithmetic Circuits}.
\newblock {\em Theory of Computing}, 6(1):135--177, 2010.
\newblock Conference version appeared in the proceedings of STOC 2008.

\bibitem[Raz13]{Raz13}
Ran Raz.
\newblock {Tensor-Rank and Lower Bounds for Arithmetic Formulas}.
\newblock {\em J. {ACM}}, 60(6):40:1--40:15, 2013.
\newblock Conference version appeared in the proceedings of STOC 2010.

\bibitem[Rez92]{Rez92}
Bruce Reznick.
\newblock Sums of even powers of real linear forms.
\newblock {\em Memoirs of the AMS}, 96:463, 1992.

\bibitem[RSY08]{RazSY08}
Ran Raz, Amir Shpilka, and Amir Yehudayoff.
\newblock {A Lower Bound for the Size of Syntactically Multilinear Arithmetic
  Circuits}.
\newblock {\em {SIAM} J. Comput.}, 38(4):1624--1647, 2008.
\newblock Conference version appeared in the proceedings of FOCS 2007.

\bibitem[RY08]{RazY08}
Ran Raz and Amir Yehudayoff.
\newblock {Balancing Syntactically Multilinear Arithmetic Circuits}.
\newblock {\em Computational Complexity}, 17(4):515--535, 2008.

\bibitem[RY09]{RazY09}
Ran Raz and Amir Yehudayoff.
\newblock {Lower Bounds and Separations for Constant Depth Multilinear
  Circuits}.
\newblock {\em Computational Complexity}, 18(2):171--207, 2009.
\newblock Conference version appeared in the proceedings of CCC 2008.

\bibitem[SS97]{ShoupS97}
Victor Shoup and Roman Smolensky.
\newblock {Lower Bounds for Polynomial Evaluation and Interpolation Problems}.
\newblock {\em Computational Complexity}, 6(4):301--311, 1997.
\newblock Conference version appeared in the proceedings of FOCS 1991.

\bibitem[Str73a]{S73}
Volker Strassen.
\newblock Die berechnungskomplexi\"{a}t von elementarysymmetrischen funktionen
  und von iterpolationskoeffizienten.
\newblock {\em Numerische Mathematik}, 20:238--251, 1973.

\bibitem[Str73b]{S73a}
Volker Strassen.
\newblock Vermeidung von divisionen.
\newblock {\em The Journal f{\"{u}}r die Reine und Angewandte Mathematik},
  264:182--202, 1973.

\bibitem[SW01]{ShpilkaW01}
Amir Shpilka and Avi Wigderson.
\newblock Depth-3 arithmetic circuits over fields of characteristic zero.
\newblock {\em Computational Complexity}, 10(1):1--27, 2001.
\newblock Conference version appeared in the proceedings of CCC 1999.

\bibitem[Tav15]{Tavenas15}
S{\'{e}}bastien Tavenas.
\newblock Improved bounds for reduction to depth 4 and depth 3.
\newblock {\em Inf. Comput.}, 240:2--11, 2015.
\newblock Conference version appeared in the proceedings of MFCS 2013.

\bibitem[TLS22]{lst2}
S{\'{e}}bastien Tavenas, Nutan Limaye, and Srikanth Srinivasan.
\newblock Set-multilinear and non-commutative formula lower bounds for iterated
  matrix multiplication.
\newblock In Stefano Leonardi and Anupam Gupta, editors, {\em {STOC} '22: 54th
  Annual {ACM} {SIGACT} Symposium on Theory of Computing, Rome, Italy, June 20
  - 24, 2022}, pages 416--425. {ACM}, 2022.

\bibitem[Val79]{Valiant79a}
Leslie~G. Valiant.
\newblock {Completeness Classes in Algebra}.
\newblock In {\em Proceedings of the 11h Annual {ACM} Symposium on Theory of
  Computing, April 30 - May 2, 1979, Atlanta, Georgia, {USA}}, pages 249--261,
  1979.

\bibitem[Vol16]{Volkovich16}
Ilya Volkovich.
\newblock {A Guide to Learning Arithmetic Circuits}.
\newblock In {\em Proceedings of the 29th Conference on Learning Theory, {COLT}
  2016, New York, USA, June 23-26, 2016}, pages 1540--1561, 2016.

\bibitem[VSBR83]{ValiantSBR83}
Leslie~G. Valiant, Sven Skyum, S.~Berkowitz, and Charles Rackoff.
\newblock {Fast Parallel Computation of Polynomials Using Few Processors}.
\newblock {\em {SIAM} J. Comput.}, 12(4):641--644, 1983.

\bibitem[Yau16]{Yau16}
Morris Yau.
\newblock Almost cubic bound for depth three circuits in {VP}.
\newblock {\em Electronic Colloquium on Computational Complexity {(ECCC)}},
  23:187, 2016.

\end{thebibliography}

	\appendix
        \addtocontents{toc}{\protect\setcounter{tocdepth}{1}}
        \section{Proofs from Section \ref{sec:prelim}}\label{app:prelim}

\subsection{Proof of Lemma~\ref{lem:binomial-coeffs}}
\label{app:lem:binomial-coeffs}
	\begin{enumerate}
	    \item 
	    \begin{align*}
	    M(a,b) = \frac{(a+b-1)\cdots(a)}{b!} \ge \frac{a^b}{b^b} \text{~, and}
	    \end{align*}
	    \begin{align*}
	    M(a,b) & \le \frac{(a+b-1)^b}{\paren{\frac{b}{e}}^b} \tag{using $b! \ge \paren{b/e}^b$}\\
	    & \le \frac{(2a)^b\!\cdot\!e^b}{b^b} \le \paren{\frac{6a}{b}}^b.
	    \end{align*}
	    \item 
	    \begin{align*}
	    \frac{M(a,b+c)}{M(a,b)} = \paren{\frac{a+b+c-1}{b+c}}\cdots\paren{\frac{a+b}{b+1}}.
	    \end{align*}
	    The bounds follow from the fact that each of the above $c$ many fractions lies between $\frac{a}{2b}$ and $\frac{2a}{b}$. 
	    \item \begin{align*}
	    \frac{M(c,d)}{M(b,d)} = \paren{\frac{c+d-1}{b+d-1}}\cdots\paren{\frac{c}{b}}.
	    \end{align*}
	    The lower bound follows from the fact that each of the above $d$ many fractions is at least $\frac{c}{b}$. 
	\end{enumerate} \qed

\begin{algorithm}
	\caption{Canonical tree of a binary tree}
	\label{alg:canon}
	\begin{algorithmic}[1]
		\Function{\canon}{$\cT$}
		\If{$\cT$ is an empty tree} \Return $\cT$.
		\EndIf
		\State $v \gets \textnormal{root node of $\cT$}$.
		\Comment{If $v$ has no left (or right) child, $v_L$ (resp. $v_R$) defined below is treated as an empty node and the corresponding subtree is considered empty.}
		\State $v_{L} \gets \textnormal{left child of } v$. 
		\State $v_{R} \gets \textnormal{right child of } v$.
		\Comment{Recursively ``canonizing'' the left and right subtrees.}
		\State $\cT_{v_L} \gets \canon(\cT_{v_L})$.
		\State $\cT_{v_R} \gets \canon(\cT_{v_R})$.\label{line:two-canons}
        \Comment{\emph{encoding} is any fixed 1-1 map from binary trees to positive integers.}
		\If{$\leaves(v_L)  > \leaves(v_R)$ or ($\leaves(v_L) = \leaves(v_R)$ and $\encoding(\cT_{v_L}) > \encoding(\cT_{v_R})$)} 
        \Comment{The left and right subtrees at the root are swapped.}
		\State \swap($\cT_{v_L},\cT_{v_R}$). \label{line:swap-right-left-trees}
		
		\EndIf 
		\State \Return $\cT$
		\EndFunction
	\end{algorithmic}
\end{algorithm}

\subsection{Proof of Proposition~\ref{clm:prelim}}
\label{app:clm:prelim}
    The lemma is trivially true for empty trees, so we will assume that $\cT$ has at least one leaf.
    \begin{enumerate}
        \item In Algorithm \ref{alg:canon}, we swap the left and right subtrees of a node when the left-subtree has more leaves than the right-subtree and do not swap if the right-subtree has more nodes than the left-subtree. Because of this, it follows from a simple inductive argument that $\canon(\cT)$ is right-heavy and is isomorphic to $\cT$.
        
        \item Suppose $\canon(\cT) = \canon(\widetilde{\cT})$. Using Item 1, we get that $\cT$ and $\widetilde{\cT}$ are isomorphic as both of them are isomorphic to $\canon(\cT)$. Conversely, suppose that $\cT$ is isomorphic to $\widetilde{\cT}$ via a bijection $\phi$. We shall use induction on the number of leaves of $\cT$ and $\widetilde{\cT}$; if both of them have just one leaf, then we trivially have that $\canon(\cT) = \canon(\widetilde{\cT})$.  Denoting the left and right children of the root $v$ of $\cT$ by $v_L$ and $v_R$, we get that $\cT_{v_L}$ and $\widetilde{\cT}_{\phi(v_L)}$ are isomorphic under $\phi$ and so are $\cT_{v_R}$ and $\widetilde{\cT}_{\phi(v_R)}$. As these trees have fewer leaves than $\cT$ and $\widetilde{\cT}$, we have $\canon(\cT_{v_L}) = \canon(\widetilde{\cT}_{\phi(v_L)})$ and $\canon(\cT_{v_R}) = \canon(\widetilde{\cT}_{\phi(v_R)})$ by induction. Hence, after line \ref{line:swap-right-left-trees} in the execution of Algorithm \ref{alg:canon} on inputs $\cT$ and $\widetilde{\cT}$, either $\canon(\cT_{v_L})$ and $\canon(\widetilde{\cT}_{\phi(v_L)})$ are both the left child of $v$ and $\phi(v)$ respectively, or both are the right child. The same is true for $\canon(\cT_{v_R})$ and $\canon(\widetilde{\cT}_{\phi(v_R)})$. Thus $\canon(\cT) = \canon(\widetilde{\cT})$.
        
        
        \item Note that $\phi$ also induces an isomorphism from $\cT_v$ to $\canon(\cT)_{\phi(v)}$. Hence from Item 2, $\canon(\cT_v) = \canon(\canon(\cT)_{\phi(v)})$. By our design of the function \canon(), note that all subtrees of a canonical tree are themselves canonical trees i.e., there exists a binary tree $\cT'$ such that $\canon(\cT') = \canon(\cT)_{\phi(v)}$. Thus, $\canon(\cT_v) = \canon(\canon(\cT)_{\phi(v)})=\canon(\canon(\cT'))=\canon(\cT')=\canon(\cT)_{\phi(v)}$. 
    \end{enumerate} \qed

\subsection{$\APP$ vs skewed partials} \label{subsec: app vs skewp}

In this section, we show that there are some lower bounds that can be proved using $\APP$ but not with the skewed partials measure ($\SkewP$). Consider circuits of the form $C = Q_1^{e_1} + \cdots + Q_s^{e_s}$, where $Q_1, \ldots, Q_{s}$ are arbitrary polynomials of degree at most $d \leq \frac{n}{2e}$. How large an $s$ do we need for $C$ to compute the monomial $P := x_1\cdots x_n$? \cite{Kayal12a} introduced the $\SP$ measure to show that $s = 2^{\Omega\paren{\frac{n}{d}}}$. Here, we prove the same using the $\APP$ measure. However, this lower bound cannot be proved using $\SkewP$.
Let $\F$ be a field of size greater than $n$. \footnote{This restriction on $\F$ is not required. If $|\F| < n$, we take any extension $\K$ of size more than $n$, consider $C$ and $P$ over $\K$, and prove the lower bound on $s$. Observe that the bound will continue to hold over $\F$.} \\

\noindent\textbf{Analysing $\SkewP(P)$.} The $\SkewP$ measure is defined in \cite{KayalNS20}. If $\vecx = \vecy \uplus \vecz$, then $$\SkewP_{\vecy,k}(P) := \dim\spacespanned{\brac{\der{P}{m}}_{\vecy = 0}: m \text{ is a monomial of degree } k \text{ in } \vecy}.$$ Observe that $\SkewP_{\vecy, k}(P) \leq 1$ for all $\vecy \subseteq \vecx$ and $k$. Hence, we cannot hope to get a lower bound on $s$ using skewed partials. \\

\noindent\textbf{Analysing $\APP$ of $C$ and $P$.}  Let $k = \floor{\frac{n}{2ed}}$ and $n_0 = k+1$. We begin by proving an upper bound on $\APP_{n_0,k}(C)$.
\begin{claim}
    $\APP_{n_0,k}(C) \leq s\cdot \binom{n_0 + dk - k}{n_0}$.
\end{claim}
\begin{proof}
    Observe that for all $i \in [s]$, 
    $$\partialf^k\paren{Q_i^{e_i}} \subseteq \set{\vecx^{\leq d(k - 1)}Q_i^{\max\set{e_i - k,0}}}.\footnote{Here $\vecx^{\leq d(k-1)}$ denotes the set of all monomials of degree at most $d(k-1)$.}$$
    Thus, $\APP_{k, n_0}\paren{Q_i^{e_i}} \leq \binom{n_0 + dk - k}{n_0}$. The claim follows from the sub-additivity of $\APP$.
\end{proof}
We now compute  $\APP_{k,n_0}(P)$.
\begin{claim}
    $\APP_{k,n_0}(P) = \binom{n}{k}$.
\end{claim}
\begin{proof}
    Fix distinct $\alpha_1, \ldots, \alpha_n \in \F$ and let $\vecz = \set{z_1, \ldots, z_{n_0}}$ be a fresh set of variables. Let $L$ map $x_i$ to $\ell_i(\vecz) := z_1 + \alpha_i\cdot z_2 + \alpha_i^2 \cdot z_3 + \cdots + \alpha_i^{n_0 - 1}z_{n_0}$ for all $i \in [n]$. Observe that any $n_0 = k+1$ many linear forms in $\set{\ell_1, \ldots, \ell_{n}}$ are linearly independent. Now,
    $$\partialf^k(P) = \set{\prod_{i \in S}x_i : S \subseteq [n], |S| = n-k}.$$
    We now argue that the polynomials of the set $\pi_{L}\paren{\partialf^k (P)} = \set{\pi_L(\prod_{i \in S}x_i): S \subseteq [n], |S| = n-k} = \set{\prod_{i \in S}\ell_i: S \subseteq [n], |S| = n-k}$ are linearly independent; this would prove the claim.

    Consider any $\F$-linear combination
    $$\sum_{|S| = n-k} \beta_S \cdot \prod_{i \in S}\ell_i = 0$$ and fix an arbitrary $S'$ of size $n-k$. Observe that for every $S \neq S'$, at least one of the linear forms $\set{\ell_j : j \notin S'}$ divides $\prod_{i \in S}\ell_i$. Thus, if $I$ is the ideal generated by $\set{\ell_j : j \notin S'}$, then
    $$\beta_{S'}\prod_{i \in S'}\ell_i = 0 \mod I.$$ Since $I$ is an ideal generated by  linear forms, $\F[\vecz]/I$ is a polynomial ring. This implies that $\beta_{S'} = 0$ or $\prod_{i \in S'}\ell_i = 0 \mod I$. The latter is not true: It implies that there exists an $i \in S'$ such that $\ell_i = 0 \mod I$, i.e., $\ell_i$ is an $\F$-linear combination of $\set{\ell_j : j \notin S'}$. However, as $|[n] \setminus S'| = k$, this contradicts the fact that every $k+1$ linear forms in $\set{\ell_1, \ldots, \ell_n}$ are linearly independent. Thus, $\beta_{S'} = 0$.  Repeating this argument for all $S' \subset [n]$ such that $|S'| = n-k$, we get that the elements of $\set{\prod_{i \in S}\ell_i: S \subseteq [n], |S| = n-k}$ are linearly independent.
\end{proof}

From the above two claims, and using $k = \floor{\frac{n}{2ed}}$, $n_0 = k+1$, we get,
$$s \geq \frac{\binom{n}{k}}{\binom{n_0 + dk - k}{n_0}} \geq  \frac{\paren{\frac{n}{k}}^k}{\paren{\frac{e(n_0 + kd-k))}{n_0}}^{n_0}} \geq \frac{\paren{\frac{n}{k}}^k}{\paren{\frac{e(kd + 1)}{k+1}}^{k+1}} \geq \frac{\paren{\frac{n}{k}}^k}{n^{O(1)}\paren{\frac{n/2}{n/2ed}}^{k}} \geq \frac{1}{n^{O(1)}}\paren{\frac{n}{ked}}^k \geq \frac{2^k}{n^{O(1)}} \geq 2^{\Omega\paren{\frac{n}{d}}}.$$

	\section{Proofs from Section \ref{sec:der-structural}}\label{app:der-structural}

\subsection{Proof of Claim~\ref{clm:ders-as-rs}}
\label{app:clm:ders-as-rs}
        \begin{flalign*}
		&&\partial_{\mu([k])} \paren{\prod \limits_{i\in [t]} Q_i}
		 & = \sum_{\substack{\kappa:[t]\to 2^{[k]}\text{~s.t.~}\\ \sqcup_{i \in [t]} \kappa_i = [k]}} \prod \limits_{i \in [t]} {\partial_{\mu(\kappa_i)} Q_i} \\
		  && &= \sum_{\substack{S \subseteq [t]}} \sum_{\substack{\kappa:[t]\to2^{[k]}\text{~s.t.~}\\ \sqcup_{i \in [t]} \kappa_i = [k] \\ \set{i \in [t]: |\kappa_i| \le \frac{k}{d}.d_i} = S}} \prod \limits_{i \in [t]} {\partial_{\mu(\kappa_i)} Q_i} \\
		 && & = \sum_{\substack{S \subseteq [t]}}~\sum_{\substack{\widetilde{\kappa}:\overline{S}\to 2^{[k]}\text{ s.t. }\\ \forall i \in \overline{S}, |\widetilde{\kappa}_i|>\frac{k}{d}.d_i}}~ \sum_{\substack{\kappa' : S \to 2^{[k]} \text{ s.t. }\\ \kappa = \kappa' \sqcup \widetilde{\kappa}\\ \sqcup_{i \in [t]} \kappa_i = [k]\\ \forall i \in S, |\kappa'_i| \le \frac{k}{d}.d_i}} \prod_{i \in [t]} \partial_{\mu(\kappa_i)} Q_i \\
		  && &= \sum_{\substack{S \subseteq [t]}}~\sum_{\substack{\widetilde{\kappa}:\overline{S}\to 2^{[k]}\text{ s.t. }\\ \forall i \in \overline{S}, |\widetilde{\kappa}_i|>\frac{k}{d}.d_i}}~ \sum_{\substack{\kappa : [t] \to 2^{[k]} \text{ s.t. }\\ \kappa \text{ extends } \widetilde{\kappa}\\ \sqcup_{i \in [t]} \kappa_i = [k]\\ \forall i \in S, |\kappa_i| \le \frac{k}{d}.d_i}} \prod_{i \in [t]} \partial_{\mu(\kappa_i)} Q_i \\
		 && &= \sum_{\substack{S \subseteq [t]}}~\sum_{\substack{\widetilde{\kappa}:\overline{S}\to 2^{[k]}\text{ s.t. }\\ \forall i \in \overline{S}, |\widetilde{\kappa}_i|>\frac{k}{d}.d_i}} R_{S,\widetilde{\kappa}}.  && \text{ \raggedleft{(by the definition of $R_{S,\widetilde{\kappa}}$ in \eqref{eqn:poly-defn})}} 
	\end{flalign*} \qed

\subsection{Proof of Claim~\ref{clm:cx}}
\label{app:clm:cx}
	\begin{flalign*}
	&& U_{S,\kappa} &= \paren{\partial_{\mu(P_S)} \prod \limits_{i \in S}Q_i} \cdot \prod_{i \in \overline{S}} \partial_{\mu(\kappa_i)} Q_i\\
        && &= \sum_{\substack{\kappa':S \to 2^{P_S}\text{~s.t.~}\\\sqcup_{i \in S} \kappa'_i = P_S}} \prod \limits_{i \in S} {\partial_{\mu(\kappa'_i)} Q_i}\cdot \prod_{i \in \overline{S}} \partial_{\mu(\kappa_i)} Q_i\\
	&& & = \sum \limits_{T \subseteq S} ~\sum_{\substack{\kappa':S \to 2^{P_S}\text{~s.t.~}\\ \sqcup_{i \in S} {\kappa'}_i = P_S \\ \set{i \in S : |{\kappa'_i}| \le \frac{k}{d}.d_i} = T}}\prod \limits_{i \in S} {\partial_{\mu({\kappa'}_i)} Q_i}\cdot \prod \limits_{i \in \overline{S}} \partial_{\mu(\kappa_i)} Q_i && \text{(reordering the summation based on $T$)}\\
		&& & = \sum \limits_{T \subseteq S} ~\sum_{\substack{\kappa'':S \setminus T \to 2^{P_S} \text{~and~} \kappa''':T \to 2^{P_S}\text{~s.t.}\\\kappa' = \kappa'' \sqcup \kappa''' \\ \sqcup_{i \in S} {\kappa'}_i = P_S \\  \forall i \in S \setminus T, |\kappa''_i| > \frac{k}{d} \cdot d_i \\ \forall i\in T, |\kappa'''_i| \le \frac{k}{d}\cdot d_i \\}} ~\prod_{i \in S} \partial_{\mu(\kappa'_i)} Q_i\cdot \prod_{i \in \overline{S}} \partial_{\mu(\kappa_i)} Q_i \\
		&& & = \sum \limits_{T \subseteq S}~ \sum_{\substack{\kappa'':S\setminus T \to 2^{P_S} \text{~s.t.}\\ \forall i \in S \setminus T, |\kappa''_i| > \frac{k}{d}\cdot d_i}}~ \sum_{\substack{\kappa''':T \to 2^{P_S} \text{~s.t.}\\\kappa^* = \kappa'' \sqcup \kappa''' \sqcup \widetilde{\kappa}\\ \forall i \in T, |\kappa^*_i| \le \frac{k}{d}\cdot d_i \\ \sqcup_{i \in T} \kappa'''_i \sqcup \sqcup_{i \in S \setminus T} \kappa''_i = P_S}} \prod_{i \in [t]} \partial_{\mu(\kappa^*_i)} Q_i \\
		\intertext{\raggedleft{(as $\kappa^*$ extends $\kappa'=\kappa'' \sqcup \kappa'''$ and $\kappa$ extends $\widetilde{\kappa}$ )}}\\
		&& & = \sum \limits_{T \subseteq S}~ \sum_{\substack{\kappa'':S\setminus T \to 2^{P_S} \text{~s.t.}\\ \forall i \in S \setminus T, |\kappa''_i| > \frac{k}{d}\cdot d_i}}~ \sum_{\substack{\kappa^*:[t] \to 2^{[k]} \text{~s.t.}\\\kappa^* = \kappa'' \sqcup \kappa''' \sqcup \widetilde{\kappa}\\ \forall i \in T, |\kappa^*_i| \le \frac{k}{d}\cdot d_i \\ \sqcup_{i \in [t]} \kappa^*_i = [k]}} \prod_{i \in [t]} \partial_{\mu(\kappa^*_i)} Q_i \\
		\intertext{\raggedleft{(because $\sqcup_{i \in [t]} \kappa^*_i = \paren{\sqcup_{i \in S\setminus T} \kappa''_i \sqcup \sqcup_{i \in T} \kappa'''_i} \sqcup \sqcup_{i \in \overline{S} } \widetilde{\kappa}_i = P_S \sqcup \sqcup_{i \in \overline{S}} \kappa_i = \sqcup_{i \in [t]} \kappa_i = [k]$ from \eqref{eqn:s-pi})}}\\
		&& & = \sum \limits_{T \subseteq S}~ \sum_{\substack{\kappa'':S\setminus T \to 2^{P_S} \text{~s.t.}\\ \forall i \in S \setminus T, |\kappa''_i| > \frac{k}{d}\cdot d_i}} R_{T,\kappa'' \sqcup \widetilde{\kappa}} 
		\intertext{\raggedleft{($R_{T,\kappa'' \sqcup \widetilde{\kappa}}$ is well-defined because $\kappa^*$ extends $\kappa'' \sqcup \widetilde{\kappa}$ and \eqref{eqn:s-pi})}}\\
		&& & = R_{S,\widetilde{\kappa}} + \sum_{\substack{T \subsetneq S \text{~and~} \kappa'':S\setminus T \to 2^{P_S} \text{~s.t.}\\ \forall i \in S \setminus T, |\kappa''_i| > \frac{k}{d}.d_i}} R_{T,\kappa'' \sqcup \widetilde{\kappa}}. \label{eqn:decompose3}&& \text{(separating out the case $T=S$)}
	\end{flalign*} \qed

\subsection{Proof of Lemma~\ref{lem:measure-upper-bd} (continued)}
\label{app:lem:measure-upper-bd}
    \begin{align*}
	    \APP_{k,n_0}(Q) & = \max \limits_{L : \vecx \to \spacespanned{\vecz}} \dim \spacespanned{\pi_L\paren{\partialf^k(Q_1\cdots Q_t)}}\\
	    & \le \max \limits_{L : \vecx \to \spacespanned{\vecz}} \dim \spacespanned{\pi_L \paren{\sum \limits_{\substack{S \subseteq [t]; ~k_0, \ell_0 \ge 0 \\ k_0 + \frac{k}{d-k}\cdot \ell_0 \le k - \entropy_k(d_1,\dots,d_t)}} \spacespanned{\vecx^{\ell_0}\cdot \partialf^{k_0}\paren{ \prod \limits_{i \in S} Q_i}}} }  \tag{from Lemma ~\ref{lem:der-structure}} \\
	    & \le \max \limits_{L : \vecx \to \spacespanned{\vecz}} \dim \spacespanned{ \sum \limits_{\substack{S \subseteq [t]; ~k_0, \ell_0 \ge 0 \\ k_0 + \frac{k}{d-k}\cdot \ell_0 \le k - \entropy_k(d_1,\dots,d_t)}} \spacespanned{\pi_L\paren{\vecx^{\ell_0}\cdot \partialf^{k_0}\paren{ \prod \limits_{i \in S} Q_i}}} } \tag{as $\pi_L$ distributes over addition} \\
	    & \le \max \limits_{L : \vecx \to \spacespanned{\vecz}} \sum \limits_{\substack{S \subseteq [t]; ~k_0, \ell_0 \ge 0 \\ k_0 + \frac{k}{d-k}\cdot \ell_0 \le k - \entropy_k(d_1,\dots,d_t)}} \dim  \spacespanned{\pi_L\paren{\vecx^{\ell_0}}\cdot\pi_L\paren{\partialf^{k_0}\paren{ \prod \limits_{i \in S} Q_i}}} \tag{using Proposition~\ref{prop:sub} (Item 4) and $\pi_L$ distributes over multiplication}\\
	    & \le \max \limits_{L : \vecx \to \spacespanned{\vecz}} \sum \limits_{\substack{S \subseteq [t]; ~k_0, \ell_0 \ge 0 \\ k_0 + \frac{k}{d-k}\cdot \ell_0 \le k - \entropy_k(d_1,\dots,d_t)}} \dim  \spacespanned{\pi_L\paren{\vecx^{\ell_0}}}\cdot \dim \spacespanned{\pi_L\paren{\partialf^{k_0}\paren{ \prod \limits_{i \in S} Q_i}}} \tag{from Proposition~\ref{prop:sub} (Item 5)}\\
	    & \le \max \limits_{L : \vecx \to \spacespanned{\vecz}} \sum \limits_{\substack{S \subseteq [t]; ~k_0, \ell_0 \ge 0 \\ k_0 + \frac{k}{d-k}\cdot \ell_0 \le k - \entropy_k(d_1,\dots,d_t)}} \abs{\pi_L\paren{\vecx^{\ell_0}}}\cdot \abs{\pi_L\paren{\partialf^{k_0}\paren{ \prod \limits_{i \in S} Q_i}}} \\
	    & \le \max \limits_{L : \vecx \to \spacespanned{\vecz}} \sum \limits_{\substack{S \subseteq [t]; ~k_0, \ell_0 \ge 0\\ k_0 + \frac{k}{d-k}\cdot \ell_0 \le k - \entropy_k(d_1,\dots,d_t)}} \abs{ \vecz^{\ell_0}}\cdot \abs{ \pi_L\paren{\partialf^{k_0}\paren{ \prod \limits_{i \in S} Q_i}}} \tag{as $L$ is a map from $\vecx$ to $\spacespanned{\vecz}$, $\pi_L(m) \in \vecz^{\ell_0}$ for any monomial $m$ over $\vecx$ of degree $\ell_0$} \\
	    & \le 2^t\!\cdot\!d^2\!\cdot\!\max \limits_{\substack{k_0, \ell_0 \ge 0 \\ k_0 + \frac{k}{d-k}\cdot \ell_0 \le k - {\entropy_k(d_1,\dots,d_t)}}} M(n_0, \ell_0)\!\cdot\!M(n,k_0).
	\end{align*} \qed	
	\section{Proofs from Section \ref{sec:low-depth-lb}}\label{app:low-depth-lb}
\subsection{Proof of Claim~\ref{clm:entropy-sub}}
\label{app:entropy-sub}
    	We prove the claim by analysing the following three sub-cases.\\
    	
	\noindent \textbf{Case (i):} $m_j < k_j$. Then, $\eta \ge \abs{k_j - \alpha\!\cdot\!m_j} = k_j - \alpha\!\cdot\! m_j \ge k_j - m_j \ge 1$. 
 
    Now, let $\alpha_1:= \sum\limits_{\nu=0}^{\delta}  \frac{(-1)^\nu}{\tau^{2^\nu-1}}$, $\alpha_2:= \frac{(-1)^{\delta+1}}{\tau^{2^{\delta+1}-1}}$ and $\alpha_3 := \sum\limits_{\nu=\delta+2}^{\Delta-1}  \frac{(-1)^\nu}{\tau^{2^\nu-1}}$. Then, let $\alpha_4:= \tau^{2^\delta-1}\!\cdot\! \paren{k_j - m_j\!\cdot\! \alpha_1}$. Noting that $\alpha = \alpha_1 + \alpha_2 + \alpha_3$ we have, 
	
	\begin{align*}
	    \eta = \tau^{2^\delta-1}\!\cdot\! \abs{k_j - \alpha\!\cdot\! m_j} & = \abs{\tau^{2^\delta-1}\!\cdot\! k_j - \tau^{2^{\delta}-1}\!\cdot\! m_j\!\cdot\! (\alpha_1+\alpha_2+\alpha_3)} \tag{as $\alpha = \alpha_1 + \alpha_2 + \alpha_3$ by definition}\\
	    & \ge \abs{\alpha_4 - \tau^{2^\delta-1}\!\cdot\! m_j\!\cdot\! \frac{(-1)^{\delta+1}}{\tau^{2^{\delta+1}-1}} - \tau^{2^{\delta}-1}\!\cdot\! m_j\!\cdot\! \alpha_3 } \\
	    & \ge \abs{\alpha_4 - \tau^{2^\delta-1}\!\cdot\! m_j\!\cdot\! \frac{(-1)^{\delta+1}}{\tau^{2^{\delta+1}-1}}} - \abs{\tau^{2^{\delta}-1}\!\cdot\! m_j\!\cdot\! \alpha_3}\\
	    & \ge \abs{\abs{\alpha_4} - \frac{m_j}{\tau^{2^\delta}}} - \abs{\tau^{2^{\delta}-1}\!\cdot\! m_j\!\cdot\! \alpha_3} \\
	    & = \abs{\abs{\alpha_4} - \frac{m_j}{\tau^{2^\delta}}} - \abs{\sum\limits_{\nu=\delta+2}^{\Delta-1}\!\cdot\! \frac{(-1)^\nu\!\cdot\! \tau^{2^\delta-1}\!\cdot\! m_j}{\tau^{2^\nu-1}}}\\
	    & \ge \abs{\abs{\alpha_4} - \frac{m_j}{\tau^{2^\delta}}} - \abs{\frac{\tau^{2^\delta-1}\!\cdot\! m_j}{\tau^{2^{\delta+2}-1}}} \tag{taking only the leading term of the summation}\\
	    & \ge \abs{\abs{\alpha_4} - \frac{m_j}{\tau^{2^\delta}}} - \abs{\frac{\tau^{2^\delta-1}\!\cdot\! \tau^{2^\delta}}{\tau^{2^{\delta+2}-1}}} \tag{since $m_j \le d_j \approx_2 \tau^{2^{\delta}}$}\\
	    & \ge \abs{\abs{\alpha_4} - \frac{m_j}{\tau^{2^\delta}}} - \frac{1}{\tau^2} \\
	    & = \abs{\abs{\alpha_4} - \frac{m_j}{\tau^{2^\delta}}} - o(1).
	\end{align*}
        Notice that, as $m_j \leq d_j \approx_2 \tau^{2^\delta}$, $\frac{m_j}{\tau^{2^\delta}} \leq 1$.
	
	\noindent \textbf{Case (ii):} $k_j \le m_j \le 6\!\cdot\! k_j$. Note that $$m_j = \frac{6\!\cdot\! m_j + m_j}{7} \le \frac{6\!\cdot\! m_j + 6\!\cdot\! k_j}{7} = \frac{6}{7}\!\cdot\! d_j \le \frac{6}{7}\!\cdot\! \tau^{2^\delta},\text{~and}$$ 
	$$m_j \ge \frac{m_j + k_j}{2} = \frac{d_j}{2} \ge \frac{1}{4}\!\cdot\! \tau^{2^\delta}.$$
	Thus $\frac{m_j}{\tau^{2^\delta}} \in \brac{\frac{1}{4}, \frac{6}{7}}$. On the other hand, $\alpha_4= \tau^{2^\delta-1}\!\cdot\! k_j - \tau^{2^\delta-1}\!\cdot\! m_j\!\cdot\! \alpha_1$ is an integer since the denominators of all the terms in $\alpha_1$ divide $\tau^{2^\delta-1}$. Therefore $\frac{m_j}{\tau^{2^\delta}}$ is at least $\min\set{1/4,3/4, 6/7,1/7}=1/7$ distance from any integer, and from $\abs{\alpha_4}$ in particular. That is, $\abs{\abs{\alpha_4} - \frac{m_j}{\tau^{2^\delta}}} \ge 1/7$ and $\eta \ge \abs{\abs{\alpha_4} - \frac{m_j}{\tau^{2^{\delta}}}} - o(1) \ge \Omega(1)$.\\
	
	\noindent \textbf{Case (iii):} $m_j > 6\!\cdot\! k_j$. Then,
	\begin{align*}
	    -k_j + m_j\!\cdot\! \alpha_1 
	    & = -k_j + m_j\!\cdot\! \paren{\sum\limits_{\nu=0}^\delta \frac{(-1)^\nu}{\tau^{2^\nu-1}}}\\
	    & = -k_j + m_j -  m_j  \!\cdot\! \paren{\sum\limits_{\nu=1}^\delta \frac{(-1)^{\nu-1}}{\tau^{2^\nu-1}}}\\
	    & \ge -k_j+m_j - \frac{m_j}{\tau}\\
	    & \ge \frac{m_j}{2} - k_j \tag{as $\tau = \omega(1)$}\\
	    & \ge \frac{m_j}{2} - \frac{m_j}{6}
	     = \frac{m_j}{3}
	     \ge \frac{2}{7}\cdot \paren{m_j + k_j} = \frac{2\cdot d_j}{7}
	     \ge \frac{\tau}{7}
	     \ge 2.
	\end{align*}
		Hence, $ \abs{\abs{\alpha_4} - \frac{m_j}{\tau^{2^\delta}}} = \abs{\tau^{2^\delta-1}\!\cdot\! \paren{-k_j + m_j\!\cdot\! \alpha_1} - \frac{m_j}{\tau^{2^\delta}}} \ge 2 - 1 =1$ and $\eta \ge \Omega(1)$.
\qed

\subsection{Proof of Lemma~\ref{lem:app-final-lb}}
\label{app:lem:app-final-lb}
    Using Lemma~\ref{lem:measure-upper-bd} (Item 2) and the fact that $\APP$ is sub-additive (Proposition~\ref{prop:sub-additive}), we get
    $$\APP_{k,n_0}(P) \le \sum \limits_{i=1}^s \APP_{k,n_0}(Q_{i,1}\cdots Q_{i,t_i}) \le s\cdot 2^t\cdot d^2 \cdot \max \limits_{\substack{k_0,\ell_0 \ge 0 \\ k_0 + \frac{k}{d-k}\cdot\ell_0~\le~k - \gamma}} M(n,k_0)\!\cdot\!M(n_0,\ell_0).$$
    On the other hand, we have $\APP_{k,n_0}(P) \ge 2^{-O(d)}\!\cdot\!M(n,k)$.
    Putting these two together, we get for some integers $k_0,\ell_0 \ge 0$ satisfying
    \begin{align*}
        k_0 + \frac{k}{d-k}\cdot \ell_0 \le k - \gamma,
        \label{eqn:d1-d2-upper}
    \end{align*}
    that,
    \begin{align*}
        s \ge   2^{-O(d)}\!\cdot\!2^{-t}\!\cdot\!d^{-2}\!\cdot\!\frac{M(n,k)}{M(n,k_0)\!\cdot\!M(n_0,\ell_0)}
         \ge 2^{-O(d)}\!\cdot\!\paren{\frac{n}{d}}^{\Omega(\gamma)}. \tag{Using Lemma \ref{lem:binomial-coeffs} (item 1), absorbing $6^{\ell_0}$ in $2^{-O(d)}$, and borrowing calculations beginning from \eqref{eqn:calcs}}\\    
    \end{align*}\qed

\subsection{Proof of Lemma~\ref{lem:hard-poly}}
\label{app:lem:hard-poly}
    We construct the word $\vecw$ as follows. Let $h' = \frac{h\cdot k}{d-k}\in \brac{\frac{h}{29},h}$. The word $\vecw$ shall consist of the following elements (the ordering of these elements shall be fixed shortly): $h,\dots,h$ ($k$ times), $-\floor{h'},\dots,-\floor{h'}$ ($k_1$ times), $-\ceil{h'},\dots,-\ceil{h'}$ ($k_2$ times), where $k_1:=(d-k)\ceil{h'}-kh$ and $k_2:=d-k-k_1$. We note that $k_1,k_2 \in \Z_{\ge 0}$ and $k+k_1+k_2=d$. Assuming $\floor{h'} = \ceil{h'}-1$ (even if $h'\in \Z$, the calculations are similar), the total sum of the weights is 
    \begin{align}
        \sum_{i\in[d]} w_i & = kh-k_1\floor{h'}-k_2\ceil{h'} = kh-k_1(\ceil{h'}-1)-k_2\ceil{h'} = kh - k_1\ceil{h'}+k_1 - k_2\ceil{h'} \nonumber \\
        & =kh-k_1\ceil{h'}+(d-k)\ceil{h'}-kh-k_2\ceil{h'}=0.
        \label{eqn:zero-sum}
    \end{align}
    
    Now we fix the ordering of the above weights. For $i=1$ to $d$ in this order, if the sum $\sum_{j \in [i-1]} w_j$ is non-negative (for example, this happens for $i=1$), set $w_i$ to be an arbitrary negative weight that is available, otherwise set it to be the positive weight $h$ (if available). 
    
    If the above procedure never runs out of positive or negative weights at any step $i\in[d]$, then for all $i\in [d]$, $\abs{w_1 + \dots + w_i} \le h$. In other words, $\vecw$ is $h$\unbiased. Now suppose the procedure runs out of negative weights at an index $i \in [d]$. This means that the sum $\sum_{j \in [i-1]} w_j$ is non-negative but there are no negative weights available among the unused weights. But then, the total sum of the weights would be equal to $\sum_{j \in [i-1]} w_j$ plus the sum of unused weights, which is greater than $0$, contradicting \eqref{eqn:zero-sum}. We get a similar contradiction if there are insufficient positive weights at any point.
    For the rest of the proof, we fix $\vecw$ to be the above word. Then, \begin{equation}
        k\!\cdot\! 2^h \le n \le d\!\cdot\! 2^h, \text{~so~} 2^h \approx_{30} \paren{\frac{n}{k}}.
        \label{eqn:k-n-d}
    \end{equation}
    
    Denoting the variables of $P_{\vecw}$ by $\vecx=\vecy \sqcup \vecz$, where $\vecy$ are the positive variables and $\vecz$ are the negative variables, we take $$n_0:=\abs{\vecz} \approx_2 (d-k)\!\cdot\! 2^{\ceil{h'}}.$$ Note that $n_0  \approx_2 (d-k)\!\cdot\! 2^{\ceil{h'}} \approx 2(d-k)\!\cdot\!2^{h'} = 2(d-k)\!\cdot\! 2^{\frac{hk}{d-k}} \le 2k\!\cdot\! 2^{h} = 2(n-n_0)$ where the last inequality follows from the fact that $\frac{d-k}{k}\cdot 2^{\frac{hk}{d-k}}$ is an increasing function of $k$ when $k\in \brac{\frac{d}{30},\frac{d}{2}}$ and $h > 100$. That is, $n_0 \le 2n/3$ and $n_0 \approx 2(d-k)\!\cdot\! \paren{\frac{n}{k}}^{\frac{k}{d-k}}$ by \eqref{eqn:k-n-d} and $k \leq \frac{d}{2}$. Also, $n_0 \geq \frac{d-k}{2}\cdot 2^{\ceil{h'}} \geq \frac{d-k}{2}\cdot 2^{h'} \geq \frac{d-k}{2}\cdot 2^{\frac{h}{29}} \geq \frac{d-k}{2}\cdot 2^3 \geq 2d$ as $h > 100$ and $k \leq \frac{d}{2}$. Define a map $L:\vecx \to \spacespanned{\vecz}$ as follows:
    
    \begin{equation*}
        L(x) = \begin{cases}
        0, \text{~if~} x \in \vecy,\\
        x, \text{~if~} x \in \vecz.
        \end{cases}
    \end{equation*}
    
    We can lower bound the $\APP$ measure by using $L$ and considering only the derivatives with respect to the set-multilinear monomials over all the positive sets, i.e.,  $\cM_+(\vecw)$. By the definition of the polynomial $P_{\vecw}$ and because $\sum_{i\in [d]} w_i =0$, for every $m_+ \in \cM_+$, there exists a unique $m_- \in \cM_-$ such that $m_+\cdot m_-$ is a monomial in $P_\vecw$ and vice versa.\footnote{Recall the definition of $P_{\vecw}$ from Section \ref{sec:prelim}. Because $|\cM_+| = |\cM_-|$, the bit representations of $m_+$ and $m_-$ are the same. However, they can have different degrees.} Hence the set of all derivatives of $P_\vecw$ with respect to monomials in $\cM_+$ is exactly $\cM_-$, yielding
    \begin{equation}\partialf^k \paren{P_{\vecw}} \supseteq \cM_-(\vecw). \label{eqn:some-ders} \end{equation}
    
    \noindent Using the fact that $\sum_{i \in [d]} w_i =0$ and \eqref{eqn:k-n-d}, the size of $\cM_-(\vecw)$ is 
    \begin{align}
        \abs{\cM_-(\vecw)} = 2^{\sum_{i\in [d]: w_i < 0} \abs{w_i}} = 2^{hk} \ge 2^{-O(k)}\!\cdot\! \paren{\frac{n}{k}}^k \ge 2^{-O(d)}\!\cdot\! M(n,k).
        \label{eqn:neg-monomials}
    \end{align}
    The last bound follows from Lemma~\ref{lem:binomial-coeffs} (Item 1), as $n \ge n_0 \ge d \ge k$. As the substitution $\pi_L$ does not affect negative variables, thus,
    \begin{align*}
        \APP_{k,n_0}(P_{\vecw}) & \ge \dim{\spacespanned{\pi_L\paren{\partialf^k \paren{ P_{\vecw}}}}}
         \ge \dim \spacespanned{\pi_L(\cM_-(\vecw))}  = \dim \spacespanned{\cM_-(\vecw)} \ge 2^{-O(d)}\!\cdot\! M(n,k).
    \end{align*}
    
    We now analyze the shifted partials of the same polynomial with $\ell := \floor{\frac{n\cdot d}{n_0}}$. Recall that $n_0 \leq 2n/3$.
    \begin{align*}
        \SP_{k,\ell}(P_{\vecw}) & \ge \dim \spacespanned{\vecx^\ell\cdot  \partialf^k \paren{P_{\vecw}}}\\
        & \ge \dim \spacespanned{\vecy^\ell\cdot \cM_-(\vecw) } \tag{as $\vecx \supseteq \vecy$ and \eqref{eqn:some-ders}}\\
        & \ge \abs{\vecy^\ell\cdot \cM_-(\vecw)}\\
        & = \abs{\vecy^\ell}\cdot \abs{\cM_-(\vecw)} \tag{since $\cM_-(\vecw) \subseteq \vecz^{d-k}$ and $\vecy \cap \vecz = \Phi$}\\
        & = M(n-n_0,\ell)\!\cdot\! 2^{-O(d)}\!\cdot\! M(n,k) \tag{using $\abs{\vecy} = \abs{\vecx} - \abs{\vecz}$ and \eqref{eqn:neg-monomials}}\\
        & \ge M(n, \ell)\!\cdot\! \paren{1-\frac{n_0}{n}}^\ell\!\cdot\! 2^{-O(d)}\!\cdot\! M(n,k) \tag{using Lemma~\ref{lem:binomial-coeffs} (Item 3)}\\
        & \ge M(n, \ell)\!\cdot\! \paren{1-\frac{n_0}{n}}^{\frac{n}{n_0}\cdot d}\!\cdot\!2^{-O(d)}\!\cdot\! M(n,k)\\
        & \ge 2^{-O(d)}\!\cdot\! M(n,k)\!\cdot\! M(n,\ell) \tag{since $(1-x)^{1/x} \ge 1/3\sqrt{3}$ for $x:=n_0/n \le 2/3$}.
    \end{align*}\qed

\subsection{Proof of Lemma \ref{lemma:NW_hard_poly}}\label{subapp:proof_NW_hard_poly}

We begin by obtaining bounds on the value of $\ell$.

\begin{claim}\label{clm: bounds on ell}
    $n_0 = o(qd)$, $d^2 = o(\ell)$ and $\ell = o(qd)$. 
\end{claim}
\begin{proof}
    \begin{flalign*}
        && n_0  &= 2(d-k)\paren{\frac{qd}{k}}^{\frac{k}{d-k}} \\
        && &\leq 2(d-k)\paren{\frac{qd}{k}}^{\frac{d/2 - \sqrt{d}/8}{d/2 + \sqrt{d}/8}}&& \paren{\text{because $k = o(qd)$ and $k\leq \frac{d}{2} - \frac{\sqrt{d}}{8}$}}\\
        && & = 2(d-k)\cdot\frac{qd}{k}\cdot \paren{\frac{k}{qd}}^{\frac{2}{4\sqrt{d}+1}}\\
        && &  \leq 2d\cdot qd \cdot \frac{1}{(qd)^{\frac{1}{2.5\sqrt{d}}}}\\
        && & \leq 2d\cdot qd \cdot \frac{1}{2^{\frac{\log qd}{2.5\sqrt{d}}}}.
    \end{flalign*}
    As $d \leq \frac{1}{150}\paren{\frac{\log n}{\log \log n}}^{2}$ and $qd \geq \frac{n}{4}$, $\frac{\log qd}{2.5\sqrt{d}} \geq \frac{12\log\log n}{3} = 4 \log \log n$. As $\log d^2 \leq 4 \log \log n - \omega(1)$, $\frac{\log qd}{2.5\sqrt{d}} = \log d^2 + \omega(1)$ and $2^{\frac{\log qd}{2.5\sqrt{d}}} = \omega(d^2)$. Thus, $n_0 = o(q) = o(qd)$.\footnote{In this proof, we need $n_0 = o(q)$. However, we require $n_0 = o(n)$, in Section \ref{subsec: calcs for nsm hard poly} and so we have mentioned $n_0 = o(qd)$ in the statement of the claim.} Now, 
    $\ell \geq \frac{qd^2}{n_0} - 1 \geq \frac{qd^2}{o(q)} - 1 = \omega(d^2)$. Thus, $d^2 = o(\ell)$. Also,
\begin{flalign*}
    && \ell &\leq \frac{qd^2}{n_0} \\
        && &= \frac{qd^2}{2(d-k)}\paren{\frac{k}{qd}}^{\frac{k}{d-k}} \\
        && &\leq \frac{qd^2}{2(d-k)}\paren{\frac{k}{qd}}^{\frac{1}{29}} && \paren{\text{because $k = o(qd)$ and $\frac{d}{30} \leq k$}} \\
        && &\leq k\cdot (qd)^\frac{28}{29} && \paren{\text{as $k \leq \frac{d}{2}$}}  \\
        && &= o(n) && \paren{\text{because $k \leq \log^2 n$ and $qd \leq n$}}.
\end{flalign*}
\end{proof}
Let 
\[S = \set{\prod\limits_{i \in [k+1\ldots d]} x_{i, h(i)}: h \in \F_q[z], \deg(h) < k}\]
and 
\[T = \set{m : \exists \text{ monomials } m_1, \ m_2, \ \deg(m_1) = \ell, m_2 \in S \text{ and } m = m_1m_2}.\]
Observe that $T \subseteq \spacespanned{\vecx^{\ell}\partialf^{k}NW_{q,d,k}}$ and so, $\SP_{k, \ell}(NW_{q,d,k}) \geq |T|$. We obtain a lower bound on $|T|$. For $h\in \F_q[z]$ such that $\deg(h) < k$, let 
\[T_h = \set{m_1\prod\limits_{i \in [k+1\ldots d]} x_{i, h(i)} : \deg(m_1) = \ell}.\] 
Then, $T = \cup_{\substack{h(z) \in \F_q[z]:\\ \deg(h) < k}} T_h$. Thus, from the inclusion-exclusion principle,
\begin{align}\label{eqn:lb_size_T}
	|T| \geq \sum_{\substack{h \in \F_q[z]:\\ \deg(h) < k}}|T_h| - \sum_{\substack{h_1 \neq h_2 \in \F_q[z]:\\ \deg(h_1), \deg(h_2) < k}} \abs{T_{h_1}\cap T_{h_2}}.
\end{align}

\noindent\textbf{Lower bound on $\sum_{h} |T_h|$.} Fix an $h \in \F_q[z]$ such that $\deg(h) < k$. Then, since for monomials $m_1 \neq m_2$, $m_1\cdot \prod\limits_{i \in [k+1\ldots d]} x_{i, h(i)} \neq m_2 \cdot \prod\limits_{i \in [k+1\ldots d]} x_{i, h(i)}$, $|T_h| = {qd + \ell - 1 \choose qd - 1}$. Hence, 
\begin{align}\label{eqn:lb_qty_1_ie}
	\sum_{\substack{h \in \F_q[z]:\\ \deg(h) < k}}|T_h| = |S|^k\cdot {qd + \ell - 1 \choose qd - 1} = q^k\cdot {qd + \ell - 1 \choose qd - 1}.
\end{align}

\noindent\textbf{Upper bound on $\sum_{h_1 \neq h_2} |T_{h_1}\cap T_{h_2}|$.} For $h_1, h_2 \in \F[z]$ such that $\deg(h_1), \deg(h_2) < k$, we say that $|h_1 \cap h_2| = r$ if $|\set{h_1(k+1), \ldots, h_1(d)} \cap \set{h_2(k+1), \ldots, h_2(d)}| = r$. Now
\begin{align}\label{eqn:ub_size_T_1_T_2}
    \sum_{h_1 \neq h_2} |T_{h_1}\cap T_{h_2}| = \sum_{r = 0}^{k-1}\sum_{\substack{h_1\neq h_2: \\ |h_1 \cap h_2| = r}} |T_{h_1} \cap T_{h_2}|.
\end{align}
Fix $h_1$ and $h_2$ such that $|h_1 \cap h_2| = r$. Let $m_1 = \prod_{i \in [k+1..d]} x_{i, h_1(i)}$ and  $m_2 = \prod_{i \in [k+1..d]} x_{i, h_2(i)}$. A monomial $m \in T_{h_1} \cap T_{h_2}$ if and only if there exist degree $\ell$ monomials $m_1'$ and $m_2'$ such that $m = m_1'm_1 = m_2'm_2$.  Thus $\frac{m_2}{\gcd(m_1, m_2)}$ must divide $m_1'$. As $|h_1 \cap h_2| = r$, $\gcd(m_1, m_2)$ has degree $r$, and so $\frac{m_2}{\gcd(m_1, m_2)}$ has degree $d - k - r$. Hence the number of possible monomials $m_1'$, and thus the number of possible monomials $m$ is at most $\binom{qd+\ell-d+k+r-1}{qd-1}$. Now, the number of possible polynomials $h_1$ and $h_2$ such that $|h_1 \cap h_2| = r$ is at most $\binom{d-k}{r}q^{k-r} q^k = q^{2k-r}\binom{d-k}{r}$.\footnote{This is so because $|h_1 \cap h_2| = r$ implies that $h_1 - h_2 = (z-\alpha_1)\cdots(z-\alpha_r)\cdot g(z)$, where $\alpha_1, \ldots, \alpha_r$ are distinct elements in $[k+1..d]$ and $g(z)$ is a polynomial of degree at most $d-k$.} Hence,
\begin{align}\label{eqn:ub_T_1_T_2_r}
    \sum_{\substack{h_1\neq h_2: \\ |h_1 \cap h_2| = r}} |T_{h_1} \cap T_{h_2}| \leq q^{2k-r}\cdot \binom{d-k}{r}\binom{qd+\ell-d+k+r-1}{qd-1}.
\end{align}

\begin{claim}
    For $r \in [0..k-1]$, let $\chi(r) = q^{2k-r}\cdot \binom{d-k}{r}\binom{qd+\ell-d+k+r-1}{qd-1}$. Then $\chi(0) \geq \chi(r)$ for all $r \in [k-1]$.
\end{claim}
\begin{proof}
    We shall show that for all $r \in [0..k-2]$, $\frac{\chi(r+1)}{\chi(r)} < 1$; this will prove the claim. Fix any $r \in [0..k-2]$.
    \begin{align*}
       \frac{\chi(r+1)}{\chi(r)} &= \frac{q^{2k-r-1}\cdot\binom{d-k}{r+1}\binom{qd+\ell-d+k+r}{qd-1}}{q^{2k-r}\cdot\binom{d-k}{r}\binom{qd+\ell-d+k+r-1}{qd-1}} \\
        &= \frac{1}{q}\cdot\frac{\frac{(d-k)!}{(r+1)!(d-k-r-1)!}}{\frac{(d-k)!}{r!(d-k-r)!}}\cdot \frac{\frac{(qd+\ell-d+k+r)!}{(qd-1)!(\ell-d+k+r+1)!}}{\frac{(qd+\ell-d+k+r-1)!}{(qd-1)!(\ell-d+k+r)!}}\\
        &= \frac{1}{q}\cdot\frac{d-k-r}{r+1}\cdot\frac{qd+\ell-d+k+r}{\ell-d+k+r+1}\\
        &\leq \frac{d}{q}\cdot \frac{(1+o(1))qd}{(1-o(1))\ell} \tag{by Claim \ref{clm: bounds on ell}}\\
        &= (1 + o(1))\frac{d^2}{\ell} \\
        &= o(1), \tag{by Claim \ref{clm: bounds on ell}}
    \end{align*}
    where the second to last inequality follows from $k, r \geq 0$, $\ell, d, k, r = o(n)$ and $d, k, r = o(\ell)$ (Claim \ref{clm: bounds on ell}), and the last equality from the fact that $d^2 = o(\ell)$ (Claim \ref{clm: bounds on ell}).
\end{proof}

From Equations \eqref{eqn:ub_size_T_1_T_2}, \eqref{eqn:ub_T_1_T_2_r}, and the above claim, we get
\begin{align}\label{eqn:ub_qty_2_ie}
    \sum_{h_1\neq h_2}|T_{h_1} \cap T_{h_2}| \leq k\cdot q^{2k}\cdot \binom{qd+\ell-d+k-1}{qd-1}.
\end{align}
Thus from Equations \eqref{eqn:lb_size_T}, \eqref{eqn:lb_qty_1_ie}, and \eqref{eqn:ub_qty_2_ie},
\begin{align*}\label{eqn:lb_T}
    |T| &\geq q^k\cdot {qd + \ell - 1 \choose qd - 1} - k\cdot q^{2k}\cdot\binom{qd+\ell-d+k-1}{qd-1}\\
    &= q^k\cdot {qd + \ell - 1 \choose qd - 1} \paren{1 - \frac{k\cdot q^{2k}\cdot\binom{qd+\ell-d+k-1}{qd-1}}{q^k\cdot {qd + \ell - 1 \choose qd - 1}}}. \numberthis
\end{align*}
\begin{claim}
    $\frac{k\cdot q^{2k}\cdot\binom{qd+\ell-d+k-1}{qd-1}}{q^k\cdot {qd + \ell - 1 \choose qd - 1}}\leq \frac{1}{2}$.
\end{claim}
\begin{proof}
    \begin{flalign*}
         \frac{k\cdot q^{2k}\cdot\binom{qd+\ell-d+k-1}{qd-1}}{q^k\cdot {qd + \ell - 1 \choose qd - 1}} &= k \cdot q^k \cdot \frac{\frac{(qd+\ell-d+k-1)!}{(qd-1)!(\ell-d+k)!}}{\frac{(qd+\ell-1)!}{(qd-1)!\ell!}} \\
         &= k \cdot q^k \cdot \frac{\ell \cdot (\ell-1)\cdot (\ell-2)\cdots (\ell-d+k+1)}{(qd+\ell-1)\cdot(qd+\ell-2)\cdot(qd+\ell-3)\cdots(qd+\ell-d+k)} \\
         &= k \cdot q^k \cdot \frac{1}{\paren{\frac{qd-1}{\ell} + 1}\cdot\paren{\frac{qd-1}{\ell-1} + 1}\cdot\paren{\frac{qd-1}{\ell-2} + 1}\cdots {\paren{\frac{qd-1}{\ell-d+k+1} + 1}}}\\
         &\leq k \cdot q^k \cdot \frac{1}{\paren{\frac{qd-1}{\ell}}^{d-k}} \\
         & \leq k \cdot q^k \cdot \paren{\frac{\ell}{qd}}^{d-k}e^{\frac{2(d-k)}{qd}} \qquad\qquad\qquad \text{\raggedleft{(as $1-x \geq e^{-2x}$ for $x \in [0,1/2]$)}}\\ \\
         &= (1 + o(1))\cdot k \cdot q^k \cdot\paren{\frac{d}{n_0}}^{d-k} \qquad\qquad \text{\raggedleft{(as $d-k = o(qd)$ and $\ell \leq \frac{qd^2}{n_0}$)}}\\      
        &= (1 + o(1))\cdot k \cdot q^k \cdot \paren{\frac{d}{2(d-k)}}^{d-k}\cdot \paren{\frac{k}{qd}}^k \\
        &\leq (1 + o(1))\cdot k \cdot\paren{\frac{1}{2}}^k \qquad\qquad\qquad\qquad\qquad\qquad\quad \text{\raggedleft{(as $k \leq \frac{d}{2}$)}}\\
        &\leq \frac{1}{2},
    \end{flalign*}
    when $d \geq 120$.
\end{proof}

Thus, from Equation \eqref{eqn:lb_T} and $k = \Theta(d)$, we get
$$|T| \geq \frac{1}{2}\cdot q^k \cdot \binom{qd + \ell - 1}{qd - 1} \geq 2^{-O(d)} \cdot \paren{\frac{qd}{k}}^k\cdot\binom{qd + \ell - 1}{qd - 1} \geq 2^{-O(d)} \cdot \binom{qd + k - 1}{qd - 1}\cdot\binom{qd + \ell - 1}{qd - 1},$$
\noindent where the last inequality follows from Lemma \ref{lem:binomial-coeffs}. Recall that $\SP_{k,\ell}(NW_{q,d,k}) \geq |T|.$ Hence, $\SP_{k,\ell}(NW_{q,d,k}) \geq 2^{-O(d)}\cdot M(qd,k)\cdot M(qd, \ell)$.

\subsection{Details about the non-set-multilinear hard polynomial} \label{subsec: calcs for nsm hard poly}
Using an analysis similar to the one in the proof of Claim \ref{clm: bounds on ell}, it can be shown that $n_0 = o(n)$. Also, from the proof of Lemma \ref{lem:low-k-gamma}, $k \in \brac{\frac{d}{4}, \frac{d}{2}}$. Let us assume that $\Delta \geq 2$; the case of $\Delta = 1$ is simple and can be handled separately. Then, as $\frac{k}{d-k} \leq \alpha \leq 1 - \frac{1}{2\tau} \leq 1 - \frac{1}{2\sqrt{d}}$, $k \leq d - \frac{\sqrt{d}}{2} - k +\frac{k}{2\sqrt{d}}$, and hence $k \leq \frac{d}{2} - \frac{\sqrt{d}}{4} + \frac{k}{4\sqrt{d}} \leq \frac{d}{2} - \frac{\sqrt{d}}{4} + \frac{d}{8\sqrt{d}}  = \frac{d}{2} - \frac{\sqrt{d}}{8}$.

\begin{claim}
    $|\cM_y| \leq |\cM_z|$.
\end{claim}
\begin{proof}
    $|\cM_y| = \binom{n_1 + k - 1}{n_1 - 1}$ and $\cM_z = \binom{n_0 + d - k - 1}{n_0-1}$. Thus,
    \begin{align*}
        \frac{|\cM_z|}{|\cM_y|} &= \frac{(n_0+d-k-1)(n_0+d-k-2)\cdots n_0}{(n_1+k-1)(n_1+k-2)\cdots n_1}\cdot \frac{k!}{(d-k)!} \\
        &\geq \frac{k!}{(d-k)!}\cdot \frac{n_0^{d-k}}{n_1^k}\cdot \frac{1}{\paren{1 + \frac{k-1}{n_1}}^k} \\
        &\geq (1 - o(1))\cdot \frac{k!}{(d-k)!}\cdot \paren{(2-o(1))(d-k)}^{d-k}\paren{\frac{n}{k}}^k\cdot \frac{1}{n^k} \\
        \intertext{\raggedleft{(replacing $n_0$ by its value and as $n_0 = o(n), n_1 = \Theta(n) = \omega(k^2)$)}}
        &\geq (1.9)^{d-k} \frac{(1-o(1))\sqrt{2\pi k} \paren{\frac{k}{e}}^k}{(1+o(1))\sqrt{2\pi (d-k)} \paren{\frac{d-k}{e}}^{d-k}}\cdot \frac{(d-k)^{d-k}}{k^k} \\
        \intertext{\raggedleft{(using Sterling's approximation)}}
        &\geq (1.8)^{d-k} \cdot e^{d-2k}\cdot \sqrt{\frac{k}{d-k}} \\
        &\geq 1,
    \end{align*}
    for $d \geq 120$.
\end{proof}

\begin{claim}
    $M(n_1,k) \geq 2^{-O(k)}M(n,k)$.
\end{claim}
\begin{proof}
As $n_1 = \Theta(n) \geq k$, from Lemma \ref{lem:binomial-coeffs}, we get 
    $$M(n_1,k) \geq \paren{\frac{n_1}{k}}^k = \paren{\frac{n}{k}}^k\paren{1-\frac{n_0}{n}}^k \geq \paren{\frac{n}{k}}^k\cdot e^{-\frac{2kn_0}{n}} \geq 2^{-O(k)}\cdot \paren{\frac{6n}{k}}^k \geq 2^{-O(k)}\cdot M(n,k),$$
    where the third to last inequality follows from $n_0 = o(n)$ and the last inequality follows from $n \geq k$ and Lemma \ref{lem:binomial-coeffs}.
\end{proof}

\subsection{Can we avoid the $2^{-O(k)}$ factor loss in our lower bounds?} 
\label{app:2d-discuss}
In this section, we will work with the shifted partials measure, but the conclusion also holds for $\APP$. For any homogeneous polynomial $P\in \F[x_1,\dots,x_n]$ of degree $d$ and any choice of parameters $k < d$ and $\ell$, note that $$\SP(P) \le \min \set{M(n,k)\!\cdot\! M(n,\ell),  M(n,\ell+d-k)}.$$ Hence, the best possible lower bound \textit{with our current estimates for the measure} (from Lemma \ref{lem:measure-upper-bd}), for fixed values of $n,d,k,\gamma$, is at most 
\begin{align}
\Lambda := \max_{\ell}\set{\frac{\min\set{M(n,k)\!\cdot\! M(n,\ell),  M(n,\ell+d-k)}}{\max\limits_{\substack{k_0,\ell_0\geq 0: \\ k_0 + \frac{k}{d-k}\cdot\ell_0 \le k - \gamma}}M(n,k_0)\!\cdot\! M(n,\ell+\ell_0)}}.
\label{eqn:best-lb-1}
\end{align}

\noindent We claim that the above quantity is at most $2^{-\Omega\paren{k}}\cdot n^{O(\gamma)}$ as long as $d \le {n}^{1/4}$.
We will assume that $k \le d/2$ as the other case simply reduces to this case. Also, we assume that $\gamma = o(k)$ as otherwise the multiplicative factor of $2^{-\Omega\paren{k}}$ is irrelevant. 
The maximum of the R.H.S. in ~\eqref{eqn:best-lb-1} is attained for the choice of $\ell$ when $M(n,k) \!\cdot\! M(n,\ell)$ and $M(n,\ell+d-k)$ are the closest. Notice that $\frac{M(n,k)\cdot M(n, \ell)}{M(n, \ell + d - k)}$ increases along with $\ell$. Also, as $k \leq d/2$, there are values of $\ell$ for which the ratio is less than 1 and greater than 1. Thus we can fix $\ell$ (for a given $k$) such that $M(n,k) \!\cdot\! M(n,\ell) = M(n,\ell+d-k)$ (we use an exact equality here for brevity). 
Then,
\begin{align*}
\Lambda &\leq \frac{M(n,k)\!\cdot\! M(n,\ell)}{\max\limits_{\substack{k_0,\ell_0\geq 0: \\ k_0 + \frac{k}{d-k}\cdot\ell_0 \le k - \gamma}}M(n,k_0)\!\cdot\! M(n,\ell+\ell_0)} \\ &= \min\limits_{\substack{k_0,\ell_0\geq 0: \\ k_0 + \frac{k}{d-k}\cdot\ell_0 \le k - \gamma}}\frac{M(n,k)\!\cdot\! M(n,\ell)}{M(n,k_0)\!\cdot\! M(n,\ell+\ell_0)} \numberthis
\label{eqn:best-lb-2}
\end{align*}
We have two cases depending on how large the value of $\ell$ is. \\

\noindent \textbf{Case 1:} $\ell \ge d^2$.
\begin{align*}
    \frac{M(n,k) \!\cdot\! M(n,\ell)}{M(n,k_0)\!\cdot\! M(n,\ell+\ell_0)}
    & \le \frac{M(n,k)}{M(n,k_0)\!\cdot\! \paren{1+\frac{n-1}{\ell+\ell_0}}^{\ell_0}} \\
    & \le \frac{4\!\cdot\!M(n,k)}{M(n,k_0)\!\cdot\! \paren{1+\frac{n-1}{\ell+1}}^{\ell_0}} 
    \tag{$\ell_0 \leq d$ and $\ell \geq d^2$ implies that $\paren{\frac{\ell + \ell_0}{\ell + 1}}^{\ell_0} \leq 4$}\\
    & \le \frac{4\!\cdot\! M(n,k)}{M(n,k_0)\!\cdot\!M(n,k)^{\ell_0/\paren{d-k}}} \\
    \intertext{\raggedleft{(using $\paren{1+\frac{n-1}{\ell+1}}^{d-k} \geq \frac{M(n,\ell+d-k)}{M(n,\ell)}$ and $M(n,k)\!\cdot\!M(n,\ell) = M(n,\ell+d-k)$}}
    & \le \frac{4e\sqrt{2\pi k_0}\!\cdot\!\paren{1+o(1)}^k\!\cdot\! {\paren{en/k}^{k-\frac{k}{d-k}\cdot \ell_0}}}{ \paren{en/k_0}^{k_0}}. \tag{using Sterling's approximation and $k_0 = o(n)$}\\
\end{align*}

\noindent \textbf{Case 2:} $\ell \le d^2 \le \sqrt{n}$. In this regime, ignoring some polynomial factors, $M(n,k) = \paren{\frac{en}{k}}^k$, $M(n,\ell)=\paren{\frac{en}{\ell}}^\ell$ (unless $\ell$ is 0, in which case we may take $\ell=1$ as this does not alter the quantities we are considering by more than linear factors) and $M(n,\ell+d-k)=\paren{\frac{en}{\ell+d-k}}^{\ell+d-k}$. Similarly, $M(n,k_0) =  \paren{\frac{en}{k_0}}^{k_0}$ and $M(n,\ell+\ell_0) = \paren{\frac{en}{\ell+\ell_0}}^{\ell+\ell_0}$. These approximations follow as we can upper bound $\paren{1+\frac{x}{y}}^x$, and upper and lower bound $\paren{1-\frac{x}{y}}^x$ by a constant when $x^2 = O(y)$. Note that the real valued function defined by $f(x):= x\!\cdot\! \ln (en/x)$ over $x \in [1,\ell+d-k]$ is concave. Hence applying Jensen's inequality we obtain that $f(\ell+\ell_0) - f(\ell) \ge \frac{\ell_0}{d-k} \cdot \paren{f(\ell+d-k) -f(\ell)}$ or equivalently, $\frac{\paren{\frac{en}{\ell+\ell_0}}^{\ell+\ell_0}}{\paren{\frac{en}{\ell}}^\ell} \ge \paren{\frac{\paren{\frac{en}{\ell+d-k}}^{\ell+d-k}}{\paren{\frac{en}{\ell}}^\ell}}^{\frac{\ell_0}{d-k}} = \poly(n,d)\cdot \paren{\frac{en}{k}}^{\frac{k}{d-k}\cdot \ell_0}$ as $M(n,k)M(n,\ell) = M(n,\ell+d-k)$ and $\ell_0 \leq d-k$. Therefore, 

\begin{align*}
    \frac{M(n,k)\!\cdot\!M(n,\ell)}{M(n,k_0)\!\cdot\!M(n,\ell+\ell_0)} & \le \poly(n,d) \cdot \frac{\paren{\frac{en}{k}}^k\!\cdot\!\paren{\frac{en}{\ell}}^\ell}{\paren{\frac{en}{k_0}}^{k_0}\!\cdot\!\paren{\frac{en}{\ell+\ell_0}}^{\ell+\ell_0}}\\
    & \le \poly(n,d) \cdot \frac{\paren{en/k}^{k-\frac{k}{d-k}\cdot \ell_0}}{\paren{en/k_0}^{k_0}}.
\end{align*}

\noindent Thus, in both cases, ignoring some polynomial factors and $\paren{1+o(1)}^k \le 2^{o(k)}$ because they are asymptotically smaller than the $2^{\Omega(k)}$ factor, we get that

\begin{align*}
    \Lambda &\leq \min\limits_{\substack{k_0,\ell_0\geq 0: \\ k_0 + \frac{k}{d-k}\cdot\ell_0 \le k - \gamma}} \frac{\paren{en/k}^{k-\frac{k}{d-k}\cdot \ell_0}}{\paren{en/k_0}^{k_0}} \\
    &= \min\limits_{\substack{k_0,\ell_0\geq 0: \\ k_0 + \frac{k}{d-k}\cdot\ell_0 \le k - \gamma}} \frac{\paren{en/k}^{k-k_0-\frac{k}{d-k}\cdot \ell_0}}{\paren{k/k_0}^{k_0}}, \numberthis 
    \label{eqn:lambda-final}
\end{align*}
Thus to show that $\Lambda \leq 2^{-\Omega\paren{k}}\!\cdot\! n^{O(\gamma)}$, we only need to show that there exist values of $k_0$ and $\ell_0$ for which the fraction in the R.H.S. of \eqref{eqn:lambda-final} is at most $2^{-\Omega\paren{k}}\!\cdot\! n^{O(\gamma)}$. To find such $k_0, \ell_0$, we let $k_0$ be a function of $\ell_0$ defined as $k_0 = \floor{k-\gamma- \frac{k}{d-k}\cdot \ell_0}$. As we keep adding 1 to $\ell_0$ starting from $\floor{\frac{d-k}{3}}$ to $\floor{\frac{d-k}{2}}$, note that $k_0$ decreases from nearly $\floor{2k/3-\gamma}$ to nearly $\floor{k/2 -\gamma}$, with the decrease being at most  1 at a time (as $k \le d-k$). Therefore, as $\gamma = o(k)$, there must be a choice of $\ell_0 \in [0..(d-k)]$ such that $k_0 \approx_2 {2k/3}$. Further, for this value of $\ell_0$ and $k_0$, we also have $k-\gamma-1 \leq k_0 + \frac{k}{d-k}\cdot \ell_0 \leq k - \gamma$. Thus,~\eqref{eqn:lambda-final} yields $\Lambda \le 2^{-\Omega\paren{k}}\!\cdot\! n^{O(\gamma)}$.

	\section{Proofs from Section \ref{sec:upt-lb}}\label{app:upt-lb}
\subsection{Proof of Lemma \ref{lem:deg-seq}} \label{subsec:proof-deg-seq}
By induction on $d$. If $d=1$, then $\degseq(\cT)=(d_1)=(1)$, so all the conditions are trivially met. For $d\ge 2$, consider the node $v$ of $\cT$ defined at line~\ref{line:v-node} of Algorithm~\ref{alg:1}. Suppose the left and right children of $v$ are $v_L$ and $v_R$ respectively. 
	We have $$\leaves(v)= \leaves(v_L) + \leaves(v_R) \le 2\cdot \leaves(v_R) \le \frac{2d}{3}.$$ Then by line~\ref{line:d-d1}, $$e_1 = d - d_1 = \leaves(v) \in \bigg (\frac{d}{3}, \frac{2d}{3}\bigg] = \bigg (\frac{e_0}{3}, \frac{2e_0}{3}\bigg].$$ Note that if $v$ itself was a leaf, then $\leaves(v) \le \frac{2d}{3}$ still holds as $d\ge 2$. From line~\ref{alg:1:recurse}, it is evident that $(d_2,\dots,d_t) = \degseq(\cT_v)$. As the number of leaves in $\cT_v$ is $\leaves(v) < d$, by induction, we have that for all $i \in [t-2]$, $$f_i \in \bigg(\frac{f_{i-1}}{3}, \frac{2 f_{i-1}}{3}\bigg ]$$ where $f_i:= \leaves(v)-\sum \limits_{j=2}^{i+1} d_j$. But, notice that $f_i = d-d_1-\sum \limits_{j=2}^{i+1} d_j = e_{i+1}$. Hence, we have, for $i \in [t-1]$, $$e_i=f_{i-1} \in \bigg(\frac{f_{i-2}}{3}, \frac{2f_{i-2}}{3} \bigg] = \bigg(\frac{e_{i-1}}{3}, \frac{2e_{i-1}}{3} \bigg].$$ 
	
Also, by induction, we trivially get $d_t=1$ and $e_t=d-d_1-\sum \limits_{j=2}^{t} d_j= \leaves(v) - \sum \limits_{j=2}^{t} d_j=0$. Using $e_i \in \bigg(\frac{e_{i-1}}{3},\frac{2e_{i-1}}{3} \bigg]$ and $e_0=d$, we get $e_i \in \bigg( \frac{d}{3^i}, \frac{2^i\cdot d}{3^i}\bigg]$, so $1=d_t=e_{t-1}-e_t = e_{t-1}$. Hence, $\frac{d}{3^{t-1}} \leq 1 \leq \frac{2^{t-1}\cdot d}{3^{t-1}}$ and $\log_3 d + 1 \leq t \leq \log_{3/2} d  + 1$. \qed

\subsection{Proof of Lemma \ref{lem:log-pdt}} \label{subsec:proof-log-pdt}
The proof is by induction on the size of the formula $C$. In the base case, $C$ computes a variable, say $x_1$. Then, we have $d=s=\size(C)=t=d_1=1$ and $f=x_1$, all consistent with the lemma statement.
	
Let $v$ be the node in the canonical parse tree $\cT:=\cT(C)$ at line~\ref{line:v-node} in Algorithm~\ref{alg:1} -- that is, it is the last node in the rightmost path of $\cT$ which has more than $\frac{d}{3}$ leaves in its subtree. Now, let $\cP$ be an arbitrary parse tree of $C$ and let $\phi$ denote an isomorphism from $\cP$ to $\cT$ -- we know such an isomorphism exists from Proposition \ref{clm:prelim}. Let $u$ be the node in $\cP$ such that $\phi(u) = v$. Let $g$ be the multiplication gate in $C$ that corresponds to $u$.
	
Recall that $C_g$ and $C_{g \leftarrow y}$, respectively, denote the sub-formula at $g$ and the formula obtained by replacing the gate $g$ with $y$, where $y$ could be a new variable or a constant from the field. Due to homogeneity, $\deg(C_g) = \leaves(u) = \leaves(v)$. Also, note that $C_g$ is a UPT formula as it is a sub-formula of $C$. However, the formula $C_{g\leftarrow 0}$ need not be homogeneous, as we require all the leaves of homogeneous formulas to be labelled by variables. Nevertheless, one can easily eliminate the zero gates by simplifying the formula using the rules $g' \times 0 = 0$ and $g' + 0 = g'$; this ultimately results in a homogeneous formula. In fact, we argue below that it would be a UPT formula. 
	
Let $h$ be the first addition gate on the path from $g$ to the root of $C$. If no such gate exists, then $C_{g\leftarrow 0} = 0$ and thus can be considered to be a UPT formula with its canonical parse tree being the empty tree. Otherwise, let $h'$ be the child of $h$ such that $g$ is in the subformula rooted at $h'$. Then, $C_{g\leftarrow 0} = 0$ is equivalent to the formula obtained by removing the edge from $h'$ to $h$ in $C$ and the entire sub-formula at $h'$. Now, a parse tree is constructed by picking only one child of every addition gate. Since $C_{g\leftarrow 0} = 0$ is equivalent to the formula obtained by removing a child of an addition gate from $C$, every parse tree of $C_{g\leftarrow 0}$ is also a parse tree of $C$. Thus, $C_{g \gets 0}$ is a UPT formula. We will make use of the fact that $C_g$ and $C_{g\leftarrow 0}$ are UPT formulas later.
	
For a new variable $y$, the formula $C_{ g \leftarrow y}$ uses exactly one copy of $y$, so $C_{ g \leftarrow y}$ computes a polynomial in $\F[x_1,\dots,x_n,y]$ of the form $A_g\cdot y + B_g$, where $A_g$ and $B_g$ are polynomials in $\F[x_1,\dots,x_n]$. Substituting the variable $y$ to 0 in $C_{g\leftarrow y} = A_g\cdot y + B_g$, we see that $B_g$ is computed by the formula $C_{g\leftarrow 0}$. Plugging in $y$ to be the polynomial computed by $C_g$ results in the original formula $C$. Therefore,
\begin{equation}
	C = C_{g \leftarrow y} | _{ y = C_g} = A_g\cdot C_g + B_g = A_g\cdot C_g + C_{g \leftarrow 0}.
	\label{eqn:1}
\end{equation}
It is easy that $A_g$ is a homogeneous polynomial of degree $d-\leaves(v)$.
	
As $C_g$ and $C_{g\leftarrow 0}$ are smaller formulas compared to $C$, by induction, we have the expressions:
	$$C_g = \sum \limits_{i=1}^{s_1} Q^{(1)}_{i,1}\dots Q^{(1)}_{i,t_1},$$ and 
	$$C_{g\leftarrow 0} = \sum \limits_{i=1}^{s_2} Q^{(2)}_{i,1}\dots Q^{(2)}_{i,t_2},$$
for some homogeneous polynomials $\set{Q^{(1)}_{i,j}}_{i,j}$ and $\set{Q^{(2)}_{i,j}}_{i,j}$, $t_1\ge 1,t_2 \ge 1$, $s_1 \in [0..\size(C_g)]$, and $s_2 \in [0..\size(C_{g \leftarrow 0})]$. Plugging these expressions in \eqref{eqn:1}, we get 
\begin{equation}
    C = A_g\cdot C_g + C_{g\leftarrow 0} =  A_g\cdot \sum \limits_{i=1}^{s_1} Q^{(1)}_{i,1}\dots Q^{(1)}_{i,t_1} + \sum \limits_{i=1}^{s_2} Q^{(2)}_{i,1}\dots Q^{(2)}_{i,t_2}
		\label{eqn:2}
\end{equation}
Since the canonical parse tree of $C_g$ is
\begin{align}
    \cT(C_g) & = \canon(\cP_u) \tag{as $\cP$ is a parse tree of $C$, $\cP_{u}$ is a parse tree of $C_g$}\\
    & = \canon(\cP)_{\phi(u)} \tag{from~Proposition~\ref{clm:prelim}}\\
    & = \canon(\cP)_v \tag{as $\phi(u) = v$}\\
    & = \cT(C)_{v} \tag{as $\cT(C) = \canon(\cP)$}
\end{align}
and that of $C_{g \leftarrow 0}$ is $\cT(C)$, we have $\degseq({\cT(C)}_v) = (d_1^{(1)},\dots,d_{t_1}^{(1)})$ and $\degseq(\cT(C)) = (d_1^{(2)},\dots,d_{t_2}^{(2)})$, where $d^{(\kappa)}_{j}:= \deg\paren{Q^{(\kappa)}_{i,j}}$ for any $\kappa\in[2], i \in [s_\kappa], j\in [t_\kappa]$. From line~\ref{alg:1:recurse} of the function $\degseq$ for input $\cT(C)$, we have $\degseq(\cT(C)) = (d - \leaves(v), \degseq(\cT(C)_v))$; comparing these two sequences element-wise, we get, $d^{(2)}_1 = d-\leaves(v)$, and $d^{(2)}_{j+1}=d^{(1)}_j$ for $j$ in $[t_1]$, and $t_2=1+t_1$. From \eqref{eqn:2}, we have
\begin{equation*}
    C =  \sum \limits_{i=1}^{s} Q_{i,1}\dots Q_{i,{t}}~,
\end{equation*}
where $s:=s_1+s_2$, $t:=t_2$, and 
$$Q_{i,j}:=\begin{cases}
    A_g & \text{ for }i \in [s_1] \text{ and } j=1\\
    Q^{(1)}_{i,j-1} & \text{ for }i \in [s_1] \text{ and } j \in [2, t_2]\\
    Q^{(2)}_{i-s_1,j} & \text{ for }i \in [s_1+1, s_1+s_2] \text{ and } j \in [t_2]
\end{cases}$$
All that remains to be checked now is that the degrees of $Q_{i,j}$ for different $j$'s matches the degree sequence $\degseq(\cT(C))=(d^{(2)}_1,\dots,d^{(2)}_{t_2})$. 

For $i > s_1$, clearly $\deg(Q_{i,j})= \deg(Q^{(2)}_{i-s_1,j})= d^{(2)}_j$. For $i \le s_1$, we have $\deg(Q_{i,1}) = \deg(A_g) = d - \leaves(v) =  d^{(2)}_1$, and for $j\in [2, t_2]$, $\deg(Q_{i,j})=\deg(Q^{(1)}_{i,j-1})=d^{(1)}_{j-1}=d^{(2)}_{j}$. As desired, $s=s_1+s_2\le \size(C_g) + \size(C_{g\leftarrow0}) \le \size(C)$. \qed

\subsection{Proof of Claim \ref{claim:J}} \label{subsec:proof-J}
First, the function $\funJ:[3m] \to [t-2]$ is well-defined: Consider any $i \in [3m]$. Then the set $\set{j\in[0..t]:e_j \le 3^i}$ is not empty; it contains $t-2$, as $e_{t-2} < 3\cdot e_{t-1} = 3 \le 3^i$.\footnote{$e_{t-1} = 1$ is shown in the proof of Lemma \ref{lem:deg-seq}.} And, $0$ is not contained in this set as $$e_0 = d > 3^{3\floor{\frac{\log_3 d - 1}{3}}} = 3^{3m} \ge 3^i.$$ 
		
Let $j:=\funJ(i)\in[t-2]$. We show that $3^{i-1} < e_j$ by contradiction. Suppose that $e_j \le 3^{i-1}$. From the minimality condition in the definition of $\funJ$, we have $e_{j-1} > 3^i$. Hence, $e_{j} \le 3^{i-1} < \frac{e_{j-1}}{3}$, which contradicts Equation \eqref{eqn:e}.
It follows that $\funJ$ is injective because an integer cannot lie between two different pairs of consecutive powers of 3. \qed

\subsection{Proof of Claim \ref{fact:1}} \label{subsec:proof-1}
As $b_0$ and $b_1$ are close to each other, so are their integer parts. Indeed, $|\floor{b_1} - \floor{b_0}|$ is 0 or 1. We will only analyse the former case as the latter case can be easily reduced to the former. Hence we have, $|\fract{b_1} - \fract{b_0}| = |\paren{b_1-\floor{b_1}} - \paren{b_0-\floor{b_0}}| = |b_1 - b_0| \in \brac{\frac{1}{9}, \frac{8}{9}}$. As $\fract{b_0}$ and $\fract{b_1}$ are in the range $\brac{0,1}$, if neither of them lies in $\brac{\frac{1}{18}, \frac{17}{18}}$, the difference between them cannot be in the range $\brac{\frac{1}{9}, \frac{8}{9}}$. Therefore, at least one of $\fract{b_0}$ or $\fract{b_1}$ lies in $\brac{\frac{1}{18}, \frac{17}{18}}$. \qed

\section{Large-degree set-multilinear lower bound (using~\cite{lst1})}\label{sec:sml-lb}

\noindent One of the primary motivations for our alternate approach of the lower bounds of~\cite{lst1} is based on a hope that we can achieve exponential constant-depth homogeneous formula lower bounds (rather than simply superpolynomial) and resolve Open Problem 1.2 from Section~\ref{sec:intro}. A necessary condition for us to be able to achieve this is to at least get such lower bounds for {\em set-multilinear} constant-depth formulas. Although exponential set-multilinear (or even multilinear) constant-depth lower bounds already follow from~\cite{RazY09, ChillaraL019, ChillaraEL018}, these techniques do not give non-FPT lower bounds, unlike~\cite{lst1}.

In this section, we show how to obtain an exponential non-FPT lower bound against constant-depth set-multilinear formulas by making a small adjustment to the lower bound proof in~\cite{lst1}. However, we note that this change does not allow us to retain the iterated matrix multiplication as the hard polynomial, instead we are only able get the lower bound for some polynomial in $\VNP$. For the sake of simplicity, we shall present the idea only for depth-5 set-multilinear formulas; the parameters that work for larger depths are similar to the parameters we use in Section~\ref{sec:low-depth-lb}. We will assume that the reader is familiar with the original paper~\cite{lst1}, where the authors  handle degree at most $O(\log^2 n)$ (Lemma 13 in~\cite{lst1}), whereas the following claim handles the full range of degrees up to $n$.

\begin{claim}
    For any integers $d \le n$, there is a set-multilinear polynomial (family) $P \in \VNP$ of degree $d$ over at most $N=nd$ variables such that any set-multilinear formula of product-depth 2 computing it has size $n^{\Omega(\sqrt{d})}$. In particular, this gives a $N^{\Omega(N^{1/4})}$ lower bound.
\end{claim}

\begin{proof}
    The main idea to handle a larger degree is that the set sizes need not be powers of 2 if we are content with proving the lower bound for some polynomial in $\VNP$, and not necessarily for the $IMM$ polynomial (or its projection). 
    
    Let $\vecx = \vecx_1 \sqcup \dots \sqcup \vecx_d$ be the variable sets; of which some are ``positive'' and the rest are ``negative'' sets; in particular, let $\vecw=(w_1,\dots,w_d)$ be a tuple of (not necessarily integers) formed by $h$ and $-h_0$ such that all the prefix sums are at most $h$ in absolute value. According to the signs of the $w_i$'s, we classify the variable sets as positive or negative, and fix the sizes of the positive sets to be $n=2^h$, and that of negative sets to be $n_0 = 2^{h_0}$. We set $h=\log n$ and $h_0 = h-c\cdot h/\sqrt{d}$ for an appropriate constant $c \in [0.9,1]$ that we fix later so that $n_0$ is an integer. 
    
    Let $\cM_+$ be the set-multilinear monomials over the positive sets and similarly $\cM_-$ be the ``negative monomials'', both sets ordered in lexicographical order, where we express a monomial in a canonical way by ordering the variables in increasing order of the corresponding variable set indices. Then we define
    \[P := \sum_{\substack{{m_+ \in \cM_+},~{m_- \in \cM_-}\\ \rank(m_+) = \rank(m_-)}} m_+ \cdot m_{-},\]
    where $\rank(m)$ is the position of the positive or negative monomial $m$ in the corresponding monomial set. It is not hard to show that $P \in \VNP$ and $\relrk(P) \ge 2^{-O(h)}=n^{-O(1)}$. By following the same proof of Claim 14 in~\cite{lst1}, we conclude that any product-depth 2 set-multilinear formula $C$ must satisfy $\relrk(C) \le \frac{\size(C)}{n^{\Omega(\sqrt{d})}}$. Hence, $\size(C) \ge n^{\Omega(\sqrt{d})}$ if $C$ computes $P$. 
    
    All that remains to show is that there exists a value of $c\in [0.9,1]$ such that $n_0 = 2^{h_0} = 2^{h-c\cdot h/\sqrt{d}}$ is an integer. This is true because $n_0$ is a monotonous real function of $c$ with range 
    \begin{align*}2^{h-0.9 h/\sqrt{d}} - 2^{h-h/\sqrt{d}} 
    & = 2^{h\paren{1-1/\sqrt{d}}}\paren{2^{\frac{0.1 h}{\sqrt{d}}}-1}\\
    & \ge 2^{h\paren{1-1/\sqrt{d}}}\paren{\frac{0.1 h}{\log e\sqrt{d}}} \\
    & \ge n^{1-1/\sqrt{d}}/\sqrt{d} \tag{assuming $h$ is large enough} \\
    & \ge n^{1-1/\sqrt{d}}/\sqrt{n} \\
   & \ge 1.\end{align*}
    As the range of $n_0$ is at least 1, it must take an integral value for some setting of $c$.
\end{proof}

\section{Geometric intuition behind $\SP$ and $\APP$ measures}\label{sec:geometric}

\noindent In this section, we give the geometric intuition which led to the development of the shifted partials measure $\SP_{k,\ell}$ and the affine projection of partials measure $\APP_{k,n_0}$. We will see, however, that this intuition breaks down when $k$ is {\em large} and prior lower bounds did indeed only use small values of $k$. By using higher values of $k$ here, we depart from this geometric intuition. \\

\noindent{\bf Preliminaries: the algebra-geometry dictionary.}
Let $\mathbb{V} \subseteq \mathbb{F}^n$ be the zero set of polynomials $\inbrace{g_1(\vecx), g_2(\vecx), \ldots, g_s(\vecx)} \subseteq \FF[\vecx]^{=d}$. Let $P_{\ell}(\mathbb{V}) \subseteq \FF[\vecx]^{=\ell}$ be the vector space of polynomials of degree $\ell$ which vanish at every point in $\mathbb{V}$. Now, if $\V$ has a small codimension, then $\mathbb{V}$ {\em contains a lot of} points which in turn imposes a lot of (linear) constraints on $P_{\ell}(\mathbb{V})$ and so intuitively we should have that:  

\begin{proposition} {\bf (Informal/Qualitative).} 
    If codimension of $\mathbb{V}$ is small then dimension of $P_{\ell}(\mathbb{V})$ is also {\em small}. 
\end{proposition}

\noindent Now notice that for all $\ell \geq d$, the set of polynomials $\vecx^{(\ell-d)} \cdot \inbrace{g_1(\vecx), g_2(\vecx), \ldots, g_s(\vecx)}$ is contained in $P_{\ell}(\mathbb{V})$ and so intuitively we should have:

\begin{proposition} {\bf (Informal/Qualitative).} 
    If codimension of $\mathbb{V}$ is small then dimension of $\vecx^{(\ell-d)} \cdot \inbrace{g_1(\vecx), g_2(\vecx), \ldots, g_s(\vecx)}$ is also {\em small}\footnote{
        An asymptotic statement that captures this qualitative intuition is obtained through the notion of the Hilbert polynomial of the variety $\mathbb{V}$.
    }. 
\end{proposition}

\noindent We also have that 
\begin{fact}
    If $L \subseteq \FF^n$ is a random linear subspace, then the codimension of $\VV \cap L$ inside $L$ is the same as the codimension of $\VV$ inside $\FF^n$.  
\end{fact}

\noindent We can use the above intuition to formulate the shifted partials measure $\SP_{k,\ell}$ and the affine projection of partials measure $\APP_{k,n_0}$ as follows. \\

\noindent{\bf Formulating the measures.}
Arithmetic circuits and formulas admit various depth reductions and in particular, can be reduced to depth $4$ without increasing the size {\em too much}. Consider a term $T$ of such a depth $4$ formula. Specifically, let $T(\vecx) \in \F[\vecx]$ be a product of $t$ polynomials, i.e. $T = Q_{1}(\vecx) \cdot Q_{2}(\vecx) \cdot \ldots \cdot Q_{t}(\vecx) $. First observe that for any $k<t$, the zero set of $k$-th partials of $T$, denoted $\VV(\partialf^k T)$ is of codimension $(k+1)$. The above intuition then suggests that  
    $$ \SP_{k,\ell} := \dim \spacespanned{\vecx^{\ell} \cdot \partialf^{k} T} $$
should also be small. Also notice that for a subspace $L \subseteq \FF^n$, $\pi_L\paren{\partialf^k T}$ is the set of polynomials whose zero set is $L \cap \VV(\partialf^k T) $
and so for such a $T$ 
    $$ \APP_{k,n_0}(P) := \max \limits_{L : \vecx \to \spacespanned{\vecz} } \dim \spacespanned{ \pi_L\paren{\partialf^k P} } $$
should be small as well. Prior work \cite{GKKS14,KayalSS14,GargKS20} showed that when the degree of the factors of $T$ are also small, then these measures are indeed significantly small for such a $T$ (as compared to a random polynomial), thereby leading to lower bounds for certain classes of depth $4$ formulas (and also for some subclasses of higher depth formulas via appropriate depth reductions). But notice that the geometric intuition behind these measures holds only when $k$ is less than the number of factors $t$. In particular when $k \geq t$ then $\VV(\partialf^k T)$ could be empty\footnote{
    For example when $T = (x_1^3 + x_2^3 + \ldots + x_n^3)^t$, then $\VV(\partialf^k T)$ is empty for all $k \geq t$.
} and so its no longer clear why either $\SP_{k, \ell}(T) $ or $\APP_{k, n_0}(T)$ should be small in this regime. Our conceptual contribution is that even when $k>t$ then $\SP_{k, \ell}(T)$ and $\APP_{k, n_0}(T)$ can be small (this depends on $\gamma(k)$) and this can be used to more directly obtain the lower bounds in \cite{lst1}.

\end{document}